\documentclass[11pt]{article}
\usepackage{fullpage}
\usepackage{amsmath,amsthm,amssymb,xcolor,empheq,graphicx}
\usepackage{thm-restate}
\usepackage[hidelinks]{hyperref}
\usepackage{algpseudocode}

\MakeRobust{\Call}
\usepackage[linesnumbered,ruled,vlined]{algorithm2e}
\SetKwInOut{Input}{Input}
\SetKwProg{Procedure}{Procedure}{}{}

\usepackage{comment}
\usepackage{booktabs,multirow,bm}

\newtheorem{theorem}{Theorem}[section]
\newtheorem{lemma}[theorem]{Lemma}
\newtheorem{claim}[theorem]{Claim}
\newtheorem{corollary}[theorem]{Corollary}
\newtheorem{proposition}[theorem]{Proposition}

\theoremstyle{definition}
\newtheorem{definitiont}[theorem]{Definition}
\newtheorem{example}[theorem]{Example}
\newtheorem{remark}[theorem]{Remark}

\newcommand{\bM}{\mathbf{M}}

\newcommand{\cS}{\mathcal{S}}
\newcommand{\cA}{\mathcal{A}}
\newcommand{\cD}{\mathcal{D}}

\newcommand{\cM}{\mathcal{M}}
\newcommand{\cH}{\mathcal{H}}
\newcommand{\cG}{\mathcal{G}}
\newcommand{\cC}{\mathcal{C}}
\newcommand{\E}{\mathop{\mathbb{E}}}

\newcommand{\R}{\mathbb{R}}
\newcommand{\cost}{{\sf cost}}
\newcommand{\ham}{d_{\mathrm{Ham}}}
\newcommand{\eps}{\varepsilon}
\newcommand{\ynote}[1]{\textcolor{red}{(Yuichi: #1)}}

\allowdisplaybreaks

\usepackage{mathtools}
\DeclarePairedDelimiter{\norm}{\lVert}{\rVert}
\DeclarePairedDelimiter{\inprod}{\langle}{\rangle}

\title{Average Sensitivity of Graph Algorithms}
\author{
  Nithin Varma\\
  University of Haifa\\
  \texttt{nvarma@bu.edu}
  \and
  Yuichi Yoshida\\
  National Institute of Informatics\\
  \texttt{yyoshida@nii.ac.jp}
}

\date{}
\begin{document}
\maketitle

\begin{abstract}
  In modern applications of graph algorithms, where the graphs of interest are large and dynamic, it is unrealistic to assume that an input representation contains the full information of a graph being studied.
  Hence, it is desirable to use algorithms that, even when provided with only a (large) subgraph,
  output solutions that are close to the solutions output when the whole graph is available.
  We formalize this feature by introducing the notion of average sensitivity of graph algorithms, which is the average earth mover's distance between the output distributions of an algorithm on a graph and its subgraph obtained by removing an edge, where the average is over the edges removed and the distance between two outputs is the Hamming distance.

  In this work, we initiate a systematic study of average sensitivity.
  After deriving basic properties of average sensitivity such as composition, we provide efficient approximation algorithms with low average sensitivities for concrete graph problems, including the minimum spanning forest problem, the global minimum cut problem, the minimum $s$-$t$ cut problem, and the maximum matching problem.
  In addition, we prove that the average sensitivity of our global minimum cut algorithm is almost optimal, by showing a nearly matching lower bound.
  We also show that every algorithm for the 2-coloring problem has average sensitivity linear in the number of vertices.
  One of the main ideas involved in designing our algorithms with low average sensitivity is the following fact; if the presence of a vertex or an edge in the solution output by an algorithm can be decided locally, then the algorithm has a low average sensitivity, allowing us to reuse the analyses of known sublinear-time algorithms and local computation algorithms.
  Using this fact in conjugation with our average sensitivity lower bound for $2$-coloring, we show that every local computation algorithm for $2$-coloring has query complexity linear in the number of vertices, thereby answering an open question.
\end{abstract}

\thispagestyle{empty}
\setcounter{page}{0}
\newpage

%!TEX root=./stabilityOfAlgorithms.tex

\section{Introduction}

In modern applications of graph algorithms, where the graphs of interest are large and dynamic, it is unrealistic to assume that an input representation contains the full information of a graph being studied.
For example, consider a social network, where a vertex corresponds to a user of the social network service and an edge corresponds to a friendship relation.
It is reasonable to assume that users do not always update new friendship relations on the social network service, and that sometimes they do not fully disclose their friendship relations because of security or privacy reasons.
Hence, we can only obtain an approximation $G'$ to the true social network $G$.
This brings out the need for algorithms that can extract information on $G$ by solving a problem on $G'$.
Moreover, as the solutions output by a graph algorithm are often used in applications such as detecting communities~\cite{Newman:2004jh,Newman:2006iq}, ranking nodes~\cite{Page:1999wg}, and spreading influence~\cite{Kempe:2003iu}, the solutions output by an algorithm on $G'$ should be close to those output on $G$.

We assume that the $n$-vertex input graph $G'$ at hand is a randomly chosen (large) subgraph of an unknown true graph $G$.
Intuitively, a deterministic algorithm $\mathcal{A}$ is said to be \emph{stable-on-average} if the Hamming distance $d_{\mathrm{Ham}}\bigl(\mathcal{A}(G),\mathcal{A}(G')\bigr)$ is small, where $\mathcal{A}(G)$ and $\mathcal{A}(G')$ are outputs of $\mathcal{A}$ on $G$ and $G'$, respectively.
Here, outputs are typically vertex sets or edges sets and we assume that they are represented appropriately using binary strings.
More specifically, for an integer $k \geq 1$, we say that the \emph{$k$-average sensitivity} of a deterministic algorithm $\mathcal{A}$ is
\begin{align}
  \E_{\{e_1,\dots,e_k\} \sim \binom{E}{k}}\bigl[d_{\mathrm{Ham}}\bigl(\mathcal{A}(G),\mathcal{A}(G-\{e_1,\dots,e_k\})\bigr)\bigr]
  \label{eq:deterministic-sensitivity}
\end{align}
for every graph $G=(V,E)$, where $\{e_1,\ldots,e_k\}$ is sampled uniformly at random from $\binom{E}{k}$, the set of all subsets of $E$ of cardinality $k$, and where $G-F$ for a set of edges $F \subseteq E$ denotes the subgraph obtained from $G$ by removing $F$.
When $k = 1$, we call the $k$-average sensitivity simply \emph{average sensitivity}.
We say that algorithms with low average sensitivity are \emph{stable-on-average}.
Although we focus on graphs here, we note that our definition can also be extended to the study of combinatorial objects other than graphs such as strings and constraint satisfaction problems.

An algorithm that outputs the same solution regardless of the input has the least possible average sensitivity, even though it is definitely useless.
Hence, the key question in a study of average sensitivity is to reveal trade-offs between solution quality and average sensitivity for various problems.

\begin{example}\label{ex:vertices-with-high-degree}
  Consider the algorithm that, given a graph $G=(V,E)$ on $n$ vertices, outputs the set of vertices of degree at least $n/2$.
  As removing an edge changes the degree of exactly two vertices, the sensitivity of this algorithm is at most $2$.
\end{example}

\begin{example}\label{ex:s-t-shortest-path-deterministic}
  Consider the $s$-$t$ shortest path problem, where given a graph $G=(V,E)$ and two vertices $s,t \in V$, we are to output the set of edges in a shortest path from $s$ to $t$.
  Since the length of a shortest path is always bounded by $n$, where $n$ is the number of vertices, every deterministic algorithm has average sensitivity $O(n)$.
  Indeed, there exists a graph for which this trivial upper bound is tight.
  Think of a cycle of even length $n$ and two vertices $s,t$ in diametrically opposite positions.
  Consider an arbitrary deterministic algorithm $\mathcal{A}$, and assume that it outputs a path $P$ (of length $n/2$) among the two shortest paths from $s$ to $t$.
  With probability half, an edge in $P$ is removed, and $\mathcal{A}$ must output the other path $Q$ (of length $n/2$) from $s$ to $t$.
  Hence, the average sensitivity must be $1/2 \cdot (n/2) = \Omega(n)$.
  In this sense, there is no deterministic algorithm with nontrivial average sensitivity for the $s$-$t$ shortest path problem.
\end{example}

We also generalize our definition of average sensitivity to apply to randomized algorithms.
Let $\mathcal{A}(G)$ denote the output distribution of $\mathcal{A}$ on $G$.
Let $d_{\mathrm{EM}}(\mathcal{A}(G),\mathcal{A}(G'))$ denote the earth mover's distance between $\mathcal{A}(G)$ and $\mathcal{A}(G')$, where the distance between two outputs is measured by the Hamming distance.
Specifically, $d_{\mathrm{EM}}(\mathcal{A}(G),\mathcal{A}(G'))$ is equal to $\min_{\mathcal{D}} \left[\E_{(x,y) \sim \mathcal{D}} \left[d_{\text{Ham}}(x,y)\right]\right]$, where $\mathcal{D}$ denotes a distribution over pairs $(x,y)$ of outputs of $\mathcal{A}$ such that the left and right marginals of $\mathcal{D}$ are equal to $\mathcal{A}(G)$ and $\mathcal{A}(G')$, respectively.
%\nvnote{In the sentence above, should we use inf instead of min?}
Then, for an integer $k \geq 1$, the \emph{$k$-average sensitivity} of a randomized algorithm $\mathcal{A}$ is
\begin{align}
  \E_{\{e_1,\dots,e_k\} \sim \binom{E}{k}}\left[d_{\mathrm{EM}}\bigl(\mathcal{A}(G),\mathcal{A}(G-\{e_1,\dots,e_k\})\bigr)\right]
  \label{eq:randomized-sensitivity}
\end{align}
where $\{e_1,\ldots,e_k\}$ is sampled uniformly at random from $\binom{E}{k}$. %, the set of all subsets of $E$ of cardinality $k$.
%When $k = 1$, again, we say that the \emph{average sensitivity} is at most $\beta$.
Note that when the algorithm $\mathcal{A}$ is deterministic,~(\ref{eq:randomized-sensitivity}) matches the definition of the average sensitivity for deterministic algorithms.

\begin{remark}\label{rem:sensitivity-wrt-tv}
The $k$-average sensitivity of an algorithm $\cA$ with respect to the total variation distance can be defined as
$\E_{\{e_1,\ldots,e_k\} \sim \binom{E}{k}}\left[d_{\mathrm{TV}}\bigl(\mathcal{A}(G),\mathcal{A}(G-\{e_1,\dots,e_k\})\bigr)\right]$,
where $d_{\mathrm{TV}}(\cdot,\cdot)$ denotes the total variation distance between two distributions.
It is easy to observe that, if the $k$-average sensitivity of an algorithm with respect to the total variation distance is at most $\gamma(G)$, then its $k$-average sensitivity is bounded by ${\sf H}\cdot \gamma(G)$, where the ${\sf H}$ is the maximum over Hamming weights of all solutions output (with nonzero probability) by running $\mathcal{A}$ on $G = (V,E)$ and on all the graphs in $\{G-\{e_1,\dots,e_k\}: \{e_1,\dots,e_k\} \in \binom{E}{k}\}$.
\end{remark}

\begin{example}
  Randomness does not help improve the average sensitivity of algorithms for the $s$-$t$ shortest path problem.
  Think of the cycle graph given in Example~\ref{ex:s-t-shortest-path-deterministic}, and suppose that a randomized algorithm $\mathcal{A}$ outputs the $s$-$t$ paths $P$ and $Q$ with probability $p$ and $q = 1-p$, respectively.
  Then, the average sensitivity is $p \cdot 1/2 \cdot (n/2) + q \cdot 1/2 \cdot (n/2) = \Omega(n)$.
\end{example}

\subsection{Basic properties of average sensitivity}

Our definition of average sensitivity has many nice properties.
In this section, we discuss some such properties of average sensitivity that are useful in the design of our stable-on-average algorithms.
We denote by $\mathcal{G}$ the (infinite) set consisting of all graphs.
Given a graph $G = (V,E)$ and $e \in E$, we use $G - e$ as a shorthand for $G - \{e\}$. We use $n$ and $m$ to denote the number of vertices and edges in the input graph, respectively.

\paragraph{Bounds on $k$-average sensitivity from bounds on average sensitivity.}
This is one of the most important properties of our definition of average sensitivity.
It essentially says that a bound on the ($1$-)average sensitivity of an algorithm can be used to obtain a bound on the $k$-average sensitivity of that algorithm for $k \ge 1$.
In other words, it is enough to analyze the average sensitivity of an algorithm with respect to the removal of a single edge.

\begin{restatable}{theorem}{sadme}\label{thm:sensitivity-against-deleting-multiple-edges}
	Let $\mathcal{A}$ be an algorithm for a graph problem with average sensitivity at most $f(n,m)$.
	Then, for any integer $k \geq 1$, the algorithm $\mathcal{A}$ has $k$-average sensitivity at most $\sum_{i=1}^k f(n,m-i+1)$.
\end{restatable}

\noindent In particular, if the average sensitivity of an algorithm is bounded from above by a nondecreasing function of the number of edges, then its $k$-average sensitivity is at most $k$ times the upper bound on its average sensitivity.

\paragraph{Sequential composition.}
Another useful feature of our definition of average sensitivity is that one can obtain a stable-on-average algorithm by sequentially applying several stable-on-average subroutines.
The following two \emph{sequential composition theorems} formalize this feature.

\begin{restatable}[Sequential composition]{theorem}{scTaE}\label{thm:sequential-composition-TV-and-EM}
	Consider two randomized algorithms $\mathcal{A}_1: \mathcal{G} \to \mathcal{S}_1,\mathcal{A}_2: \mathcal{G} \times \mathcal{S}_1 \to \mathcal{S}_2$.
	Suppose that the average sensitivity of $\mathcal{A}_1$ with respect to the total variation distance is $\gamma_1(G)$ and the average sensitivity of $\mathcal{A}_2(\cdot,S_1)$ is $\beta_2^{(S_1)}(G)$ for any $S_1 \in \mathcal{S}_1$.
	Let $\mathcal{A}: \mathcal{G} \to \mathcal{S}_2$ be a randomized algorithm obtained by composing $\mathcal{A}_1$ and $\mathcal{A}_2$, that is, $\mathcal{A}(G) = \mathcal{A}_2(G,\mathcal{A}_1(G))$.
	Then, the average sensitivity of $\mathcal{A}$ is $\mathsf{H}\cdot \gamma_1(G) +\E_{S_1 \sim \cA_1(G)}\left[\beta_2^{(S_1)}(G)\right]$, where $\mathsf{H}$ denotes the maximum over Hamming weights of all the solutions output (with nonzero probability) by running $\cA$ on $G$ and all of the graphs in $\{G-e: e\in E\}$.
\end{restatable}

Our second composition theorem is for the average sensitivity with respect to the total variation distance.
This is also useful for analyzing the average sensitivity with respect to the earth mover's distance, as it can be bounded by the average sensitivity with respect to the total variation distance times the maximum over Hamming weights of solutions output, as in Remark~\ref{rem:sensitivity-wrt-tv}.

\begin{restatable}[Sequential composition w.r.t.\ the TV distance]{theorem}{c}\label{thm:composition}
	Consider $\ell$ randomized algorithms $\mathcal{A}_i: \mathcal{G} \times \prod_{j=1}^{i-1} \mathcal{S}_j \to \mathcal{S}_i$ for $i \in \{1,\dots,\ell\}$.
	Suppose that, for each $i \in \{1,\dots,\ell\}$, the average sensitivity of $\mathcal{A}_i(\cdot,S_1,\dots,S_{i-1})$ is $\gamma_i(G)$ with respect to the total variation distance for every $S_1 \in \mathcal{S}_1,\dots,S_{i-1} \in \mathcal{S}_{i-1}$.
	Consider a sequence of computations $S_1 = \mathcal{A}_1(G), S_2 = \mathcal{A}_2(G,S_1),\ldots,S_\ell=\mathcal{A}_\ell(G,S_1,\dots,S_{\ell-1})$.
	Let $\mathcal{A}: \mathcal{G} \to \mathcal{S}_\ell$ be a randomized algorithm that performs this sequence of computations on input $G$ and outputs $S_\ell$.
	Then, the average sensitivity of $\mathcal{A}$ with respect to the total variation distance is at most $\sum_{i=1}^\ell \gamma_i(G)$.
\end{restatable}

\paragraph{Parallel composition.}
It is often the case that there are multiple algorithms that solve the same problem albeit with different average sensitivity guarantees.
Such stable-on-average algorithms can be combined via \emph{parallel composition}, where we run these algorithms according to a distribution determined by the input graph.
The advantage of parallel composition is that the average sensitivity of the resulting algorithm might be better than that of the component algorithms.

\begin{restatable}[Parallel composition]{theorem}{asdma}\label{thm:avg-sensitivity-distribution-of-multiple-algorithms}
	Let $\cA_1, \dots, \cA_\ell$ be algorithms for a graph problem with average sensitivities $\beta_1(G), \dots, \beta_\ell(G)$, respectively.
	Let $\cA$ be an algorithm that, given a graph $G$, runs $\cA_i$ with probability $\rho_i(G)$ for $i \in \{1,\dots, \ell\}$, where $\sum_{i \in \{1,\dots, \ell\}} \rho_i(G) = 1$.
	Let $\mathsf{H}$ denote the maximum over Hamming weights of all solutions output (with nonzero probability) by running $\cA$ on $G$ and on all the graphs in $\{G-e:e \in E\}$.
	Then the average sensitivity of $\cA$ is at most $\sum_{i \in \{1,\dots, \ell\}}\rho_i(G)\cdot \beta_i(G) + \mathsf{H}\cdot \E_{e \sim E}\left[\sum_{i \in \{1,\dots, \ell\}}|\rho_i(G) - \rho_i(G-e)|\right]$.
\end{restatable}

\noindent In this paper, we use the above theorem extensively to combine algorithms with different average sensitivities.

\subsection{Connection to sublinear-time algorithms}

We show a relationship between the average sensitivity of an algorithm and the query complexity of a sublinear-time algorithm~\cite{Nguyen:2008fr,HassidimKNO09,YYI12} that simulates oracle access to the solution output by the former algorithm.
Roughly speaking, we show, in Theorem~\ref{thm:generic-transformation-from-local-oracle-to-stable-algorithm}, that the average sensitivity of an algorithm $\cA$ is bounded by the query complexity of another algorithm $\mathcal{O}$, which we call a \emph{solution oracle}, where $\mathcal{O}$ \emph{queries} the edges of the graph $G$ and simulates oracle access to the solution produced by $\cA$ on input $G$.
We first formalize the notion of a solution oracle.
%Add a connecting sentence.

\begin{definitiont}[Solution Oracle]\label{def:solution-oracle}
Consider a deterministic algorithm $\cA:\cG \to \cS$ for a graph problem, where each solution output by $\cA$ is a subset of the set of edges of the input graph.
An algorithm $\mathcal{O}$ is a \emph{solution oracle} for $\cA$ if $\mathcal{O}$ satisfies:
	\begin{itemize}
		\item $\mathcal{O}$ has access to a graph $G = (V,E)$, which is represented as adjacency lists, via \emph{neighbor queries}, where each query is of the form $(v,i)$ for $v \in V$ and $i \in [|V|]$ and the answer is the $i$-th neighbor of vertex $v$ in its adjacency list (and a special symbol if $i$ is larger than the degree of $v$),
		\item given an edge $e \in E$ as input, $\mathcal{O}$ queries $G$ and outputs whether $e$ is contained in the solution obtained by running $\cA$ on $G$.
	\end{itemize}
The solution oracle $\mathcal{O}$ of a randomized algorithm $\cA$ first generates a random string $\pi \in {\{0,1\}}^{r(|V|)}$ and then runs the solution oracle $\mathcal{O}_\pi$ of the deterministic algorithm $\cA_\pi$ obtained by fixing the randomness of $\cA$ to $\pi$.
\end{definitiont}
Note that an analogous definition can be made for algorithms that output a subset of vertices.

\begin{restatable}[Sublinearity implies low average sensitivity]{theorem}{gtflotsa}\label{thm:generic-transformation-from-local-oracle-to-stable-algorithm}
	Consider a randomized algorithm $\cA:\cG \to \cS$ for a graph problem, where each solution output by $\cA$ is a subset of the set of edges in the input graph.
	Assume that there exists a solution oracle $\mathcal{O}$ for $\cA$ such that $\mathcal{O}$ makes at most $q(G)$ queries to $G$ in expectation, where this expectation is taken over the random coins of $\mathcal{O}$ and over input edges $e \in E$.
Then, $\cA$ has average sensitivity at most $q(G)$.
       Moreover, given the promise that the input graphs satisfy $|E| \ge |V|$, the statement applies also to algorithms for which each solution is a subset of the set of vertices in the input graph.
 \end{restatable}
We use Theorem~\ref{thm:generic-transformation-from-local-oracle-to-stable-algorithm} to design a stable-on-average matching algorithm (Theorem~\ref{thm:matching-iterative-greedy}) based on a sublinear-time matching algorithm due to Yoshida et al.~\cite{YYI12}.

Closely related to Theorem~\ref{thm:generic-transformation-from-local-oracle-to-stable-algorithm} is Corollary~\ref{cor:LCA-to-stable-on-average-algorithm}, which says that if a graph problem has a \emph{local computation algorithm} (LCA), then one can design a stable-on-average algorithm for that problem.
LCAs, whose definition we give below, were introduced by Rubinfeld et al.~\cite{Rubinfeld:2011} and has been widely studied ever since~\cite{AlonRVX12,EvenMR18,HassidimMV16,LenzenL18,LeviRR20,LeviRY17,MansourPV18,MRVX12,MV13,ParterRVY19,ReingoldV16}. For more information on LCAs, we refer the interested reader to an excellent survey on the topic by Levi and Medina~\cite{LeviM17}.

\begin{definitiont}[Local Computation Algorithm (LCA)]\label{def:lca}
Consider a graph problem $\mathcal{P}: \cG \to \cS$, where the output to the problem is a subset of edges of the input graph.
%each solution in $\cS$ is a subset of edges of its preimage with respect to $\mathcal{P}$.
Let $\delta: \mathbb{N} \to [0,1]$ and $q, r:\mathbb{N} \to \mathbb{N}$.
A $(q,r,\delta)$-LCA for $\mathcal{P}$ is an algorithm $\mathcal{L}$ that, given query access to a graph $G = (V,E)$ (as in Definition~\ref{def:solution-oracle}), first generates a random string $\pi \in \{0,1\}^{r(|V|)}$, and satisfies:
\begin{itemize}
	\item given an input $e \in E$, the algorithm $\mathcal{L}$ makes at most $q(|V|)$ queries to $G$ and answers whether $e$ is part of a solution to the problem $\mathcal{P}$ on graph $G$, and
	\item the answers of $\mathcal{L}$ to all possible input edges are consistent with a single feasible solution to $\mathcal{P}$ on $G$.
\end{itemize}
For every graph $G$, the probability (over the choice of random string) that there exists an input edge for which $\mathcal{L}$ makes more than $q(|V|)$ queries is at most $\delta(|V|)$.
\end{definitiont}
%Mention that when the same edge is given as input again, the set of edges queried is identical to the previous instantiation.
We mention that $\mathcal{L}$ is not allowed to perform any preprocessing on the graph.
Additionally, the same set of edges is queried by $\mathcal{L}$ when the same edge is given as input multiple times.

Note that one can have an analogous definition of LCAs for graph problems where each solution is a subset of vertices.
We have the following result which is a direct corollary of Theorem~\ref{thm:generic-transformation-from-local-oracle-to-stable-algorithm}.
\begin{restatable}[LCAs imply stable-on-average algorithms]{corollary}{Ltsoaa}\label{cor:LCA-to-stable-on-average-algorithm}
Consider a graph problem $\mathcal{P}: \cG \to \cS$.
Let $\delta: \mathbb{N} \to [0,1]$ and $q, r:\mathbb{N} \to \mathbb{N}$.
If $\mathcal{P}$ has a $(q,r,\delta)$-LCA $\mathcal{L}$, then, there exists an algorithm $\mathcal{A}$ for $\mathcal{P}$, that on input $G = (V,E)$, has average sensitivity at most $q(|V|) + |E| \cdot \delta(|V|)$.
\end{restatable}

%\nvnote{Clarify what a query is. LCA upper bounds imply stable-on-average algorithms (compute solution using LCA repeatedly on each edge; if a valid solution is not output, output an arbitrary solution). LCAs without preprocessing. }
%\nvnote{Got a question for Yuichi here. Will ask during the meeting.}

Theorem~\ref{thm:generic-transformation-from-local-oracle-to-stable-algorithm} and Corollary~\ref{cor:LCA-to-stable-on-average-algorithm} cement the intuition that strong locality guarantees for solutions output by an algorithm imply that the removal of edges from a graph affects only the presence of a few elements (edges or vertices) in the solution, which in turn implies low average sensitivity.
On the contrapositive side, Corollary~\ref{cor:LCA-to-stable-on-average-algorithm} implies that a lower bound on the average sensitivity of algorithms for a problem implies a lower bound on the query complexity of an LCA (with failure probability $o(1/n^2)$) for the same problem, where $n$ denotes the number of vertices.
Exploiting this result, we show that every LCA for $2$-coloring has query complexity $\Omega(n)$, thereby answering an open question raised by Czumaj et al.~\cite{CMV18}.
We believe that this connection has the potential to shed more light on fundamental limits of LCAs and is of independent interest.
%Such connections highlight the importance of the study of average sensitivity.
%Due to its applicability in bounding the average sensitivity of algorithms, we think that Theorem~\ref{thm:generic-transformation-from-local-oracle-to-stable-algorithm} (and in particular, Corollary~\ref{cor:LCA-to-stable-on-average-algorithm}) could lead to further research in the design of local computation algorithms for various graph problems.

\subsection{Stable-on-average algorithms for concrete problems}

We summarize, in Table~\ref{tab:results}, the average sensitivity bounds that we obtain for various concrete problems.
We use $n$, $m$, $\mathsf{OPT}$ to denote the number of vertices, the number of edges, and the optimal value.

All of our algorithms run in polynomial time, and for $k \ge 1$, upper bounds on $k$-average sensitivity of these algorithms can be easily obtained using Theorem~\ref{thm:sensitivity-against-deleting-multiple-edges}.
Except in the case of our algorithm for the minimum spanning forest problem, our stable-on-average algorithms are all randomized.
Our lower bounds hold for any (randomized) algorithm to solve the respective problems with the specified approximation guarantee.

\begin{table}[t!]
	\centering
	\caption{Our results. Here $n$, $m$, $\mathsf{OPT}$ denote the number of vertices, the number of edges, and the optimal value, respectively, and $\epsilon \in (0,1)$ is an arbitrary constant. The notation $\widetilde{O}(\cdot)$ hides polylogarithmic factors in $n$.
  % An approximation guarantee of the form $(\alpha, \beta)$ indicates a multiplicative loss of $\alpha$ and an additive loss of $\beta$.
  Approximation guarantees are multiplicative unless specified otherwise.}\label{tab:results}

	\begin{tabular}{lllll}
		\toprule
		\multirow{2}{*}{Problem} & \multirow{2}{*}{Output} & Approximation & Average & \multirow{2}{*}{Reference}\\
		& & Guarantee & Sensitivity & \\
		\midrule
		Minimum Spanning  & \multirow{2}{*}{Edge set} & $1$ & $O(\frac{n}{m})$ & Sec.~\ref{sec:spanning-tree} \\
    Forest & & $< \infty$ & $\Omega(\frac{n}{m})$ & Sec.~\ref{sec:spanning-tree} \\ \cmidrule(lr){1-5}
		\multirow{3}{*}{Global Minimum Cut} & \multirow{3}{*}{Vertex set} & $2+\epsilon$ & $n^{O(\frac{1}{\epsilon\mathsf{OPT}})}$ & Sec.~\ref{subsec:minimum-cut-upper-bound}\\
		& & $1$ & $\Omega(n)$ & Sec.~\ref{subsec:minimum-cut-lower-bound} \\
    & & $<\infty$ & $\Omega\left(\frac{n^{\frac{1}{\mathsf{OPT}}}}{\mathsf{OPT}^2}\right)$ & Sec.~\ref{subsec:minimum-cut-lower-bound} \\
      \cmidrule(lr){1-5}
    Minimum $s$-$t$ Cut & Vertex set & additive $O(n^{2/3})$ & $O\left(n^{2/3}\right)$ & Sec.~\ref{sec:s-t-min-cut} \\ \cmidrule(lr){1-5}
    % Balanced Cut & Vertex set & $\left(\widehat{O}\left(\frac{\log^4 n}{\epsilon} \right), \widehat{O}\left(\frac{\log^5 n}{\epsilon}\right)\right)$ & $\widehat{O}\left(\frac{\log^9 n}{\epsilon}\right)$ & Sec.~\ref{sec:balanced-cut} \\ \cmidrule(lr){1-5}
		\multirow{3}{*}{Maximum Matching} & \multirow{3}{*}{Edge set} & $1/2$ & $1$ & Sec.~\ref{subsec:matching-greedy} \\
		& & $1-\epsilon$ & $\widetilde{O}\left({\left(\frac{\mathsf{OPT}}{\epsilon^3}\right)}^{\frac{1}{1+\Omega(\epsilon^2)}}\right)$ & Sec.~\ref{subsec:matching-iterative-greedy} \\
		& & $1$ & $\Omega(n)$ & Sec.~\ref{subsec:matching-exact-lower-bound} \\
		\cmidrule(lr){1-5}
		Minimum Vertex Cover & Vertex set & $2$ & $2$ & Sec.~\ref{subsec:matching-greedy} \\  \cmidrule(lr){1-5}
		2-Coloring & Vertex set & --- & $\Omega(n)$ & Sec.~\ref{sec:2-coloring} \\
		\bottomrule
	\end{tabular}
\end{table}

For the minimum spanning forest problem, we show that the classical Kruskal's algorithm~\cite{Kruskal:1956} has average sensitivity $O(n/m)$, which is at most $1$, and is quite small considering that Kruskal's algorithm is deterministic and that the spanning forest can have $\Omega(m)$ edges.
We also show a matching lower bound of $\Omega(n/m)$ for the minimum spanning forest problem, implying that the average sensitivity of Kruskal's algorithm is optimal.
In contrast, we show that Prim's algorithm can have average sensitivity $\Omega(m)$ for a natural (and deterministic) rule of breaking ties among edges.
%Additionally, it is not hard to show that the average sensitivities of the known polynomial time (approximation) algorithms for the other problems listed in Table~\ref{tab:results} are all $\Omega(n)$.

%\nvnote{The last sentence above needs a proof, right? One of the reviewers had pointed it out. Should we remove the mention of this?}

%\nvnote{In the following paragraph on global mincut, we were justifying our upper bound by saying that it is $o(OPT)$ as soon as $OPT$ is large enough. I am not sure if the comparison is strong, since we are measuring $OPT$-size in terms of the number of edges and the average sensitivity is measures in terms of the number of vertices. I have commented out the earlier paragraph and instead writing the following one. I am also removing the mention that we are measuring the average sensitivity of algorithms that output sets of vertices. The reason is that this is the only result where we mention that and we are already mentioning it in the Table. Additionally, we make a point about output representation in the section on discussions.}

For the global minimum cut problem, we show that every algorithm that outputs the exact mincut has to have average sensitivity $\Omega(n)$.
However, by allowing for a multiplicative approximation guarantee of $2 + \epsilon$ for $\epsilon > 0$, we design a global minimum cut algorithm with average sensitivity $n^{O(\frac{1}{\epsilon\mathsf{OPT}})}$.
If $\mathsf{OPT} = \Omega(\log n)$, the average sensitivity of our algorithm is $O(1)$, which is quite small.
We also prove a nearly tight lower bound on the average sensitivity of any algorithm that guarantees a purely multiplicative approximation to the minimum cut size.
In particular, when $\mathsf{OPT}$ is $o(\log n)$, our lower bound matches, up to a polylogarithmic factor, our upper bound for $3$-approximating the minimum cut size.
%When $\mathsf{OPT}$ is $\Omega(\log n)$, our algorithm to $3$-approximate the minimum cut size has average sensitivity $O(1)$.

%For the global minimum cut problem, our algorithm outputs a cut as a subset of the set of vertices in the input graph.
%As the approximation ratio of our algorithm is $2 + \epsilon$ for any $\epsilon > 0$, it is likely to output a cut of size close to $\mathsf{OPT}$, and hence we want to make its average sensitivity smaller than $\mathsf{OPT}$.
%We observe that the average sensitivity becomes smaller than $\mathsf{OPT}$ when $\mathsf{OPT} = \Omega(t\log \log t/\log t)$ for $t = \log(n)/\epsilon$, and it quickly decreases as $\mathsf{OPT}$ increases.

% Our stable-on-average algorithms for both the $s$-$t$ minimum cut problem outputs cuts as vertex sets.
% Both the algorithms have strong guarantees on the average sensitivity: specifically, their average sensitivities are polylogarithmic in $n$.

Our lower bound of $\Omega(n)$ on the average sensitivity of algorithms that output the exact global minimum cut also applies to the minimum $s$-$t$ cut problem.
In contrast, we show that it is possible to achieve average sensitivity of $O(n^{2/3})$ for the minimum $s$-$t$ cut problem by allowing for an additive $O(n^{2/3})$ approximation.
%In contrast, for every constant $\epsilon > 0$, we obtain an approximation algorithm for the minimum $s$-$t$ cut problem with average sensitivity $O(\log (n)/\epsilon)$ that outputs a cut of size at most $(1+\epsilon)\cdot \mathsf{OPT} + O(\log n)$.
%Note that when $\mathsf{OPT} = \Omega(\log n)$, our algorithm can be thought of as a $2$-approximation algorithm for the minimum $s$-$t$ cut problem.
% \ynote{Fix this paragraph.}
% \nvnote{Made basic fixes.}

% Our stable-on-average algorithm for the balanced cut problem incurs only polylogarithmic multiplicative and additive losses in $n$.
% The cut-sets corresponding to its output have cardinalities at least $\Omega\left(\frac{n}{\text{polylog} n}\right)$ ensuring that the cut is quite balanced.
% The previous best polynomial-time algorithm, whose average sensitivity is hard to bound, has approximation ratio $O(\sqrt{\log n})$ and outputs a vertex set of size $\Omega(n)$~\cite{arora2009expander}.
% These guarantees are only slightly better than those of our stable-on-average algorithm.
% \nvnote{How does the above approximation guarantee compare to that of the best algorithm for the balanced cut problem? Can we say that it is comparable?}

We show that the average sensitivity of every algorithm that outputs the exact maximum matching is $\Omega(n)$, implying that some approximation is essential to obtain nontrivial average sensitivity.
We also propose two stable-on-average approximation algorithms for maximum matching.
Our first algorithm has approximation ratio $1/2$ and average sensitivity at most $1$.
This result immediately implies a $2$-approximation algorithm for the minimum vertex cover problem with average sensitivity at most $2$.
Our second algorithm for maximum matching has approximation ratio $1-\eps$ and average sensitivity $\widetilde{O}\left({\left(\mathsf{OPT}/\eps^3\right)}^{1/(1+\Omega(\eps^2))}\right)$ for every constant $\eps \in (0,1)$.

%To help interpret our bounds on average sensitivity, we mention that for maximization problems whose optimal values are sufficiently Lipschitz with respect to edge removals, $O(\mathsf{OPT})$ is a trivial upper bound for the average sensitivity.
%However, this is not the case, in general, for minimization problems.

%This result shows that we do not have to sacrifice the approximation ratio a lot to obtain a nontrivial average sensitivity.

%For the minimum vertex cover problem, we propose two algorithms.
%The first algorithm has approximation ratio $2+o(1)$, which is close to the best we can hope for as obtaining $(2-\eps)$-approximation is NP-Hard assuming the Unique Games conjecture~\cite{Khot:2003hf}.
%Moreover, the average sensitivity of $\widetilde{O}(\mathsf{OPT}^{3/4})$ is much smaller than the trivial $O(\mathsf{OPT})$.
%The second algorithm has a worse approximation ratio but can achieve a better average sensitivity in some regimes.
%For example, when $\mathsf{OPT} = \Omega(n)$, $m = \Theta(n)$ and $\eps = 1/2$, the average sensitivity of the first algorithm is $\Omega(n^{3/4})$ whereas that of the second algorithm is $O(n^{1/2})$.

In the 2-coloring problem, given a bipartite graph, we are to output one part in the bipartition.
For this problem, we show a lower bound of $\Omega(n)$ for the average sensitivity, that is, there is no algorithm with nontrivial average sensitivity.

\paragraph{Implications on the query complexity of LCAs.}
Recall that, by Corollary~\ref{cor:LCA-to-stable-on-average-algorithm}, the average sensitivity lower bound for a problem implies an identical lower bound on the query complexity of an LCA (with failure probability $o(1/n^2)$) for the same problem.
This implies that every LCA for exact maximum matching, global minimum cut, and minimum $s$-$t$ cut has query compleity $\Omega(n)$.
Additionally, every LCA giving a purely multiplicative approximation guarantee for the global minimum cut has query complexity $\Omega\left(\frac{n^{\frac{1}{\mathsf{OPT}}}}{\mathsf{OPT}^2}\right)$.
Additionally, as mentioned earlier, every LCA for $2$-coloring has query complexity $\Omega(n)$, which answers an open question raised by Czumaj et al.~\cite{CMV18}.
\subsection{Discussions on average sensitivity}

\paragraph{Output representation.}
Average sensitivity is dependent on the output representation.
For example, we can double the average sensitivity by duplicating the output.
A natural idea for alleviating this issue is to normalize the average sensitivity by the maximum Hamming weight $\mathsf{H}$ of a solution.
However, for minimization problems where the optimal value $\mathsf{OPT}$ could be much smaller than $\mathsf{H}$, such a normalization can diminish subtle differences in average sensitivity, e.g., $O(\mathsf{OPT}^{1/2})$ vs $O(\mathsf{OPT})$.
It is an interesting open question whether there is a canonical way to normalize average sensitivity so that the resulting quantity is independent of the output representation.

\paragraph{Sensitivity against adversarial edge removals.}
It is also natural to take the maximum, instead of the average, over edges in definitions~\eqref{eq:deterministic-sensitivity} and~\eqref{eq:randomized-sensitivity}, which can be seen as sensitivity against adversarial edge removals.
Indeed a similar notion has been proposed to study algorithms for geometric problems~\cite{Meulemans:2018jr}.
However, in the case of graph algorithms, it is hard to guarantee that the output of an algorithm does not change much after removing an arbitrary edge. Moreover, by a standard averaging argument, one can say that for 99\% of arbitrary edge removals, the sensitivity of an algorithm is asymptotically equal to the average sensitivity, which is sufficient in most cases.

\paragraph{Average sensitivity w.r.t.\ edge additions.}
As another variant of average sensitivity, it is natural to consider incorporating edge additions in definitions~\eqref{eq:deterministic-sensitivity} and~\eqref{eq:randomized-sensitivity}.
If an algorithm is stable-on-average against edge additions, then in addition to the case of not knowing the true graph as we have discussed earlier, it will be useful for the case that the graph dynamically changes but we want to prevent the output of the algorithm from fluctuating too much.
However, in contrast to removing edges, it is not always clear how we should add edges to the graph in definitions~\eqref{eq:deterministic-sensitivity} and~\eqref{eq:randomized-sensitivity}.
A naive idea is sampling $k$ pairs of vertices uniformly at random and adding edges between them.
This procedure makes the graph close to a graph sampled from the Erd\H{o}s-R\'{e}nyi model~\cite{erdHos1959random}, which does not represent real networks such as social networks and road networks well.
To avoid this subtle issue, in this work, we focus on removing edges.

\paragraph{Alternative notion of average sensitivity for randomized algorithms.}
Consider a randomized algorithm $\mathcal{A}$ that, given a graph $G$ on $n$ vertices, generates a random string $\pi \in {\{0,1\}}^{r(n)}$ for some function $r:\mathbb{N} \to \mathbb{N}$, and then runs a deterministic algorithm $\mathcal{A}_\pi$ on $G$, where the algorithm $\mathcal{A}_\pi$ has $\pi$ hardwired into it.
Assume that $\mathcal{A}_\pi$ can be applied to any graph.
It is also natural to define the average sensitivity of $\mathcal{A}$ as
\begin{align}
  \E_{e \sim E} \left[ \E_\pi \left[d_{\mathrm{Ham}}\bigl(\mathcal{A}_\pi(G),\mathcal{A}_\pi(G-e)\bigr) \right]\right].
  \label{eq:alternative-randomized-sensitivity}
\end{align}
In other words, we measure the expected distance between the outputs of $\mathcal{A}$ on $G$ and $G-e$ when we feed the same string $\pi$ to $\mathcal{A}$, over the choice of $\pi$ and edge $e$.
Note that~\eqref{eq:alternative-randomized-sensitivity} upper bounds~\eqref{eq:randomized-sensitivity} because, in the definition of the earth mover's distance, we optimally transport probability mass from $\mathcal{A}(G)$ to $\mathcal{A}(G-e)$ whereas, in~\eqref{eq:alternative-randomized-sensitivity}, how the probability mass is transported is not necessarily optimal.

We can actually bound~\eqref{eq:alternative-randomized-sensitivity} for some of our algorithms.
In this work, however, we focus on the definition~\eqref{eq:randomized-sensitivity} because the assumption that $A_\pi$ can be applied to any graph does not hold in general, and bounding~\eqref{eq:alternative-randomized-sensitivity} is unnecessarily tedious and is not very enlightening.

\subsection{Overview of our techniques}

%\paragraph{Minimum spanning forest.} For the minimum spanning forest problem, we show that the classical Kruskal's algorithm has low average sensitivity; specifically, at most $1$.
%Interestingly, Kruskal's algorithm is deterministic and yet has low average sensitivity.
%In contrast, we show that Prim's algorithm can have average sensitivity $\Omega(m)$ for a natural (and deterministic) rule of breaking ties among edges.% selecting edges if there are ties.

\paragraph{Global minimum cut.} For the global minimum cut problem, our algorithm is inspired by a differentially private algorithm\footnote{We compare and contrast the definitions of average sensitivity and differential privacy in Section~\ref{sec:related-work}.} for the same problem by Gupta~et~al.~\cite{GLMRT10}.
Our algorithm, given a parameter $\eps > 0$ and a graph $G$ as input, first enumerates a list of cuts whose sizes are at most $(2+\eps)\cdot \mathsf{OPT}$; this enumeration can be done in polynomial time as shown by Karger's theorem~\cite{K93}.
The algorithm then outputs a cut from the list with probability exponentially small in the product of the size of the cut and $O(1/\eps \cdot \mathsf{OPT})$.
The main argument in analyzing the average sensitivity of the algorithm is that the aforementioned distribution is very close (in earth mover's distance) to a related Gibbs distribution on the set of all cuts in the graph. Therefore the average sensitivity of the algorithm is of the same order as that of the average sensitivity of sampling a cut from such a Gibbs distribution, where the latter sampling task requires exponential time.
We finally show that the average sensitivity of sampling a cut from this Gibbs distribution is at most $n^{O(1/\eps \mathsf{OPT})}$.

\paragraph{Minimum $s$-$t$ cut.}
%Our stable-on-average algorithm for minimum $s$-$t$ cut first solves an LP relaxation for the problem and then appropriately rounds the solution thus obtained.
%We show that our strategies that we use to solve the LP, and to round its solution have
% The main tool used in our stable-on-average algorithm for the $s$-$t$ minimum cut problem is an algorithm for solving the LP relaxation with small average sensitivity with respect to the $\ell_2$ norm (on a subset of variables used in rounding).
The first stage of our algorithm consists of solving an LP relaxation for the minimum $s$-$t$ cut problem in a stable way, for an appropriately defined notion of sensitivity.
Given a graph $G = (V,E)$ and vertices $s,t \in V$, the relaxation contains variables $\bm{d}(\{u,v\}) \in [0,1]$ for each $\{u,v\} \in \binom{V}{2}$, where these variables can be thought of as representing a \emph{pseudometric} over the vertices.
The constraints include triangle inequalities, and also a special constraint $\bm{d}(\{s,t\}) = 1$.
The objective is to minimize the sum of the variables associated with the edges in $G$.
Intuitively, if $\bm{d}(\{s,v\})$ is \emph{large} in a solution to the linear program, then the vertex $v$ falls on the \emph{$t$-side} of the $s$-$t$ cut represented by the solution.

Our stable LP solving strategy works by solving a related LP that is identical to the original LP, except for a regularization term added to the objective function.
This regularization term is $n^{-1/3}$ times the $\ell_2$ norm of the vector of variables $(\bm{d}(\{s,v\}): v \in V)$, where we use $\|\bm{d}\|_s$ to denote this norm.
It is easy to show that the value of the optimal solution to the regularized LP is within an additive $O(n^{2/3})$ of the value of the optimal solution to the original LP.
We also show that for all $e \in E$, the solutions output by our LP solver graphs $G$ and $G-e$ are close to each other with respect to $\|\cdot\|_s$.
Here, we use the fact that the regularized objective function is strongly convex with respect to $\|\cdot\|_s$.

Given a solution to the (regularized) LP relaxation, our rounding procedure samples a threshold $\tau \in [0,1]$ uniformly at random and outputs the set $S$ consisting of all vertices $u \in V$ such that $d(\{s,u\}) \le \tau$.
The approximation guarantee of this algorithm follows from the fact that we are rounding based on a near optimal solution to the linear programming relaxation.
To analyze the average sensitivity of the algorithm, we first show that the earth mover's distance (with respect to Hamming distance) between the outputs of the rounding procedure for inputs $\bm{d}, \bm{d}' \in {[0,1]}^{\binom{V}{2}}$ is bounded by the $\ell_1$ distance between the vectors $(\bm{d}(\{s,v\}): v \in V)$ and $(\bm{d}'(\{s,v\}): v \in V)$.
Combining this with the bound on the average sensitivity of our LP solving strategy, we obtain our final bound on the average sensitivity for our algorithm to approximate the minimum $s$-$t$ cut.

\paragraph{Maximum matching.}
Our stable-on-average $\frac{1}{2}$-approximation algorithm for the maximum matching problem first considers a uniformly random ordering of the edges in the input graph, and then greedily adds edges to the matching according to that ordering.
In the context of dynamic distributed algorithms, Censor-Hillel et al.~\cite{CHK16} showed that at most $1$ edge changes in the matching (in expectation) due to the removal of a uniformly random edge, where the expectation is taken over the edge removed and the ordering of edges.
This result immediately implies that the average sensitivity of this randomized greedy matching algorithm is at most $1$.
In addition, it implies a $2$-approximation algorithm for minimum vertex cover with average sensitivity at most $2$.

There are several components to the design and analysis of our stable-on-average $(1 - \eps)$-approximation algorithm.
Our starting point is the observation (Theorem~\ref{thm:generic-transformation-from-local-oracle-to-stable-algorithm}) that the existence of a sublinear-time solution oracle $\mathcal{O}$ (see Definition~\ref{def:solution-oracle}) for an algorithm $\cA$ implies that $\cA$ is stable-on-average.
We use Theorem~\ref{thm:generic-transformation-from-local-oracle-to-stable-algorithm} to bound the average sensitivity of a $(1-\eps')$-approximation algorithm $\cA$ for the maximum matching problem, where $\eps' = \Omega(\eps)$.
Specifically, $\cA$ constructs a matching by considering augmenting paths of increasing length, and augmenting the (initially empty) matching iteratively, where the paths of each length are considered in a uniformly random order.
Yoshida~et~al.~\cite{YYI12} constructed a sublinear-time solution oracle that, given a uniformly random edge $e \in E$ as input, makes $O\left(\Delta^{O(1/{(\eps')}^2)}\right)$ queries to $G$ in expectation and answers whether $e$ is in the matching output by $\cA$ on $G$, where the expectation is over the choice of input $e$ and the randomness in $\cA$, and $\Delta$ is the maximum degree of $G$.
Combined with Theorem~\ref{thm:generic-transformation-from-local-oracle-to-stable-algorithm}, this implies that the average sensitivity of $\cA$ is $O\left(\Delta^{O(1/{(\eps')}^2)}\right)$.

Next, we transform $\cA$ to also work for graphs of unbounded degree as follows.
The idea is to remove vertices of degree at least $\frac{m}{\eps' \mathsf{OPT}}$ from the graph and run $\cA$ on the resulting graph.
This transformation affects the approximation guarantee only by an additive $\eps' \mathsf{OPT}$ term, since the number of such \emph{high degree} vertices is small.
However, this thresholding procedure might itself have high average sensitivity, since the thresholds for $G$ and $G-e$ can be very different for all $e \in E$.

We circumvent this issue by using a Laplace random variable $L$ as the threshold, where the distribution of $L$ is tightly concentrated around $\frac{m}{\eps' \mathsf{OPT}}$.
We use our sequential composition theorem (Theorem~\ref{thm:sequential-composition-TV-and-EM}) in order to analyze the average sensitivity of the resulting procedure, where we consider the instantiation of the Laplace random threshold as the first algorithm, and the remaining steps in the procedure as the second algorithm.
The first term in the expression given by Theorem~\ref{thm:sequential-composition-TV-and-EM} turns out to be a negligible quantity and is easy to bound.
The main task in bounding the second term is to bound, for all $x \in \mathbb{R}$, the average sensitivity of a procedure $\cA_x$ that, on the input graph $G$, removes all vertices of degree at least $x$ from $G$ and runs the augmenting paths-based matching algorithm.
The heart of the argument in bounding this average sensitivity is that given a solution oracle ${\cal O}$ with query complexity $q(\Delta)$ for an algorithm ${\cal A}$, we can, for all $x \in \R$, construct a solution oracle ${\cal O}_x$ for the algorithm ${\cal A}_x$.
Moreover, the query complexity of ${\cal O}_x$ is at most $O(x^2q(x))$.
By Theorem~\ref{thm:generic-transformation-from-local-oracle-to-stable-algorithm}, this is also a bound on the average sensitivity of $\cA_x$.
Using this, we bound the second term in the expression given by Theorem~\ref{thm:sequential-composition-TV-and-EM} as $\E_{L}\left[O(L^2q(L))\right] = O\left({\left(\frac{m}{\eps' {\sf OPT}}\right)}^{O(1/{(\eps')}^2)}\right)$.

An issue with the aforementioned matching algorithm is that its average sensitivity is poor for graphs with small values of $\mathsf{OPT}$.
We observe that, in contrast to this, the algorithm that simply outputs the lexicographically smallest maximum matching has average sensitivity $O(\mathsf{OPT}^2/m)$, since the output matching stays the same unless an edge in the matching is removed.
We obtain our final stable-on-average $(1 - \eps)$-approximation algorithm for the maximum matching problem by running these two algorithms according to a probability distribution determined by the input graph.
Using our parallel composition theorem, we bound the average sensitivity of the resultant algorithm as $\widetilde{O}\left({\left(\mathsf{OPT}/\eps^3\right)}^{1/(1+\Omega(\eps^2))}\right)$.

%The design and analysis of our stable-on-average $(1-\eps)$-approximation algorithm for the maximum matching problem uses similar ideas as above.
%The only difference is that we replace the randomized greedy maximal matching algorithm above with a $(1-\eps)$-approximation algorithm that repeatedly improves a matching using greedily chosen augmenting paths.
%
%\ynote{Not to forget to mention the vertex cover problem.}

\paragraph{2-coloring.} To show our $\Omega(n)$ lower bound on the average sensitivity for $2$-coloring, consider the set of all paths on $n$ vertices and the set of all graphs obtained by removing exactly one edge from these paths (called $2$-part-paths). A path has exactly two ways of being $2$-colored and a $2$-path has four ways of being $2$-colored. A path and $2$-part-path are neighbors if the latter is obtained from the former by removing an edge. A $2$-part-path has at most four neighbors. The output distribution of any $2$-coloring algorithm $\cA$ on a $2$-part-path can be close (in earth mover's distance) only to those of at most $2$ of its neighboring paths. If $\cA$, however, has low average sensitivity, the output distributions of $\cA$ have to be close on a large fraction of pairs of neighboring graphs, which gives a contradiction.

%\subsection{Notation}\label{subsec:notations}

\subsection{Related work}\label{sec:related-work}

\paragraph{Average sensitivity of network centralities.}
\emph{(Network) centrality} is a collective name for indicators that measure importance of vertices or edges in a network.
Notable examples are closeness centrality~\cite{Bavelas:1950,Beauchamp:1965,Sabidussi:1966}, harmonic centrality~\cite{Marchiori:2000dx}, betweenness centrality~\cite{Freeman:1977ww}, and PageRank~\cite{Page:1999wg}.
To compare these centralities qualitatively, Murai and Yoshida~\cite{Murai:2019hG} recently introduced the notion of average-case sensitivity for centralities.
Fix a vertex centrality measure $c$; let $c_G(v)$ denote the centrality of a vertex $v \in V$ in a graph $G=(V,E)$.
Then, the \emph{average-case sensitivity} of $c$ on $G$ is defined as
\[
S_c(G) = \E_{e \sim E}\E_{v \sim V}\frac{|c_{G-e}(v)-c_G(v)|}{c_G(v)},
\]
where $e$ and $v$ are sampled uniformly at random.
They showed various upper and lower bounds for centralities.
See~\cite{Murai:2019hG} for details.

Since a centrality measure assigns real values to vertices, they studied the relative change of the centrality values upon removal of random edges.
As our focus in this work is on graph algorithms, our notion~\eqref{eq:randomized-sensitivity} measures the Hamming distance between solutions when one removes random edges.
%\nvtext{Additionally, the notion of average sensitivity for centralities is a ratio-based definition and is in contrast to our sensitivity definition for graph algorithms.}

\paragraph{Differential privacy.}
\emph{Differential privacy}~\cite{Dwork:2006dw} is a notion closely related to average sensitivity.
Assuming the existence of a neighbor relation over inputs, the definition of differential privacy requires that the distributions of outputs on neighboring inputs are similar.
The variant of differential privacy closest to our definition of average sensitivity is edge differential privacy introduced by Nissim et al.~\cite{Nissim:2007} and further studied by~\cite{Hay:2009a, GLMRT10,Karwa:2012,Kasiviswanathan:2013,Karwa:2014,Raskhodnikova:2016}. Here, the neighbors of a graph $G=(V,E)$ are defined to be ${\{G-e\}}_{e \in E}$.
For $\eps > 0$, we say that an algorithm is \emph{$\eps$-differentially private} if for all $e \in E$,
\begin{align}
\exp(-\eps) \cdot \Pr[\mathcal{A}(G-e) \in \mathcal{S}]
\leq \Pr[\mathcal{A}(G) \in \mathcal{S}]
\leq \exp(\eps) \cdot \Pr[\mathcal{A}(G-e) \in \mathcal{S}] \label{eq:differential-privacy}
\end{align}
for any set of solutions $\mathcal{S}$.

Differential privacy has stricter requirements than average sensitivity.
Firstly, differential privacy is a \emph{worst-case} sensitivity notion.
Moreover, since differential privacy guarantees that the probabilities of outputting a specific solution on $G$ and $G-e$ are close to each other, the total variation distance between the two distributions $\mathcal{A}(G)$ and $\mathcal{A}(G-e)$ must be small.
The earth mover's distance between two output distributions can be small even if the total variation distance between them is large, and therefore, even if an algorithm is not differentially private, it could still be stable-on-average.
Despite these differences, our stable-on-average algorithm for the global minimum cut problem is inspired by a differentially private algorithm for the same problem~\cite{GLMRT10}.

%\nvnote{I think that we should mention why our bound for average sensitivity of mincut is better than that directly implied by Gupta et al.}
%As differential privacy imposes the constraint~\eqref{eq:differential-privacy} for every $e \in E$, the requirement is sometimes too strong for graph problems.
%For example,
%In order to illustrate another key difference between the requirement placed on stable-on-average algorithms versus that on differentially private algorithms, consider the case of the minimum vertex cover problem.
%An algorithm that, given an input graph $G$, outputs a vertex cover of size smaller than $n-1$ is not differentially private.
%This is because the output reveals that there is no edge between two vertices that are not part of the vertex cover, thus violating privacy.
%The requirement~\eqref{eq:differential-privacy} implies that we must output a vertex cover for $G$ even for $G-e$, and it follows that we can only output a vertex cover of size at least $n-1$.
%It follows that we can only output a vertex cover of size at least $n-1$.
%To avoid this issue, Gupta et al.~\cite{GLMRT10} considered outputting an implicit representation of a vertex cover.

\paragraph{Generalization and stability of learning algorithms.}
Generalization~\cite{ShalevShwartz:2009ij} is a fundamental concept in statistical learning theory.
Given samples $\bm{z}_1,\ldots,\bm{z}_n$ from an unknown true distribution $\mathcal{D}$ over a dataset, the goal of a learning algorithm ${\cal L}$ is to output a parameter $\theta$ that minimizes \emph{expected loss} $\mathop{\E_{\bm{z} \sim \mathcal{D}}}[\ell(\bm{z};\theta)]$, where $\ell(\bm{z};\theta)$ is the loss incurred by
a sample $\bm{z}$ with respect to a parameter $\theta$.
As the true distribution $\mathcal{D}$ is unknown, a frequently used approach in learning is to compute a parameter $\theta$ that minimizes the \emph{empirical loss} $\frac{1}{n}\cdot\sum_{i=1}^n\ell(\bm{z}_i;\theta)$, which is an unbiased estimator of the expected loss and is purely a function of the available samples.
The \emph{generalization error} of a learner ${\cal L}$ is a measure of how close the empirical loss is to the expected loss as a function of the sample size $n$.

One technique to reduce the generalization error is to add a \emph{regularization} term to the loss function being minimized~\cite{Bousquet:2002wn}.
This also ensures that the learned parameter $\theta$ does not change much with respect to minor changes in the samples being used for learning.
Therefore, in a sense, learning algorithms that use regularization can be considered as being stable according to our definition of sensitivity.

Bousquet and Elisseeff~\cite{Bousquet:2002wn} defined a notion of stability for learning algorithms in relation to reducing the generalization error.
Their stability notion requires that the empirical loss of the learning algorithm does not change much by removing or replacing \emph{any} sample in the input data.
In contrast, in our definition of average sensitivity, we consider removing \emph{random} edges from a graph and measure the change in the output solution rather than that in the objective value.

\subsection{Organization}
We show our stable-on-average algorithms for the minimum spanning forest problem, the global minimum cut problem, the minimum $s$-$t$ cut problem, and the maximum matching problem problems in Sections~\ref{sec:spanning-tree},~\ref{sec:minimum-cut},~\ref{sec:s-t-min-cut}, and~\ref{sec:matching}, respectively.
Our lower bounds on the average sensitivity of algorithms for the global minimum cut problem and the maximum matching problem can also be found in Sections~\ref{sec:minimum-cut}, and~\ref{sec:matching}, respectively.
We show a linear lower bound for the 2-coloring problem in Section~\ref{sec:2-coloring}.
We discuss general properties of average sensitivity in Section~\ref{sec:general}.

%!TEX root=./stabilityOfAlgorithms.tex

\section{Preliminaries}\label{sec:preliminaries}
For a positive integer $n$, let $[n] = \{1,2,\ldots,n\}$.
Let $G=(V,E)$ be a graph whose vertex set is $V$ and edge set is $E$.
We denote by $\mathcal{G}$ the (infinite) set consisting of all graphs.
We often use the symbols $n$, $m$, $\Delta$ to denote the number of vertices, the number of edges, and the maximum degree of a vertex, respectively, in the input graph.
We use $\mathsf{OPT}(G)$ to denote the optimal value of a graph $G$ in the graph problem we are concerned with.
We simply write $\mathsf{OPT}$ when $G$ is clear from the context.
For an edge $e \in E$, we denote by $G-e$ the graph obtained by removing $e$ from $G$.
Similarly, for a subset of edges $F \subseteq E$, we denote by $G-F$ the graph obtained by removing every edge in $F$ from $G$.
For a subset of edges $F \subseteq E$, let $V(F)$ denote the set of vertices incident to an edge in $F$.
For a positive integer $k \le |E|$, we use the notation $\binom{E}{k}$ to denote the set of all subsets of $E$ of cardinality $k$.
For a subset of vertices $S$, let $G[S]$ be the subgraph of $G$ induced by $S$.
We denote by $\mathbb{R}_+$ the set of non-negative real numbers.
%For a set $S \subseteq V$, let $\bm{\chi}_S \in \mathbb{R}^V$ denote the characteristic vector of $S$.
% We write $x \lesssim y$ if there exists some constant $C > 0$ such that $Cx \geq y$.
For vectors $\bm{x}, \bm{y} \in \mathbb{R}^n$, we use $\inprod{\bm{x},\bm{y}}$ to denote the inner product of $\bm{x}$ and $\bm{y}$.

%\nvnote{TODO: Nithin will check if the last two sentences above are necessary for the new s-t mincut section.}
% A (convex) function $f\colon \bbR^n \to \bbR$ is \emph{$\alpha$-strongly convex} with respect to a norm $\norm{\cdot}$ if
% \begin{align*}
%     f(\bm{y}) \geq f(\bm{x}) + \inprod{\nabla f(\bm{x}), \bm{y} - \bm{x}} + \frac{\alpha}{2}\norm{\bm{y}-\bm{x}}^2.
% \end{align*}
% We say $f$ is \emph{strongly convex} if $f$ is $\alpha$-strongly convex for some $\alpha > 0$.

\subsection{Exponential Mechanism}

The \emph{exponential mechanism}~\cite{McSherry:2007dh} is an algorithm that, given a vector $\bm{x} \in \mathbb{R}^n$ and a real number $\eta > 0$, returns an index $i \in [n]$ with probability proportional to $e^{-\eta \bm{x}(i)}$.
Just as the exponential mechanism is useful to design differentially private algorithms, it is also useful to design stable-on-average algorithms.
Lemma~\ref{lem:exponential-mechanism} formalizes this statement.
\begin{lemma}\label{lem:exponential-mechanism}
	Let $\eta > 0$ and let $\mathcal{A}$ be the algorithm that, given a vector $\bm{x} \in \mathbb{R}^n$, applies the exponential mechanism to $\bm{x}$ and $\eta$.
	Then for any $t>0$, we have
	\[
	\Pr_{i \sim \mathcal{A}(\bm{x})}\left[\bm{x}(i) \geq \mathsf{OPT} + \frac{\log n}{\eta} + \frac{t}{\eta}\right] \leq e^{-t},
	\]
	where $\mathsf{OPT} = \min_{i \in [n]}\bm{x}(i)$.
	% Moreover, we have
	% \[
	%   \E_{i \sim A(\bm{x})}[\bm{x}(i)] \leq \left(1 + \frac{1}{\eta}\right)\left(\mathsf{OPT} + \frac{\log n}{\eta}\right) + \frac{1}{\eta^2}.
	% \]
	Moreover, for all $\bm{x}' \in \mathbb{R}^n$, we have
	\[
	d_{\mathrm{TV}}(\mathcal{A}(\bm{x}),\mathcal{A}(\bm{x}'))  = O\left(\eta\cdot \|\bm{x}-\bm{x}'\|_1\right).%\cdot( \eta\mathsf{OPT}' +\log n)\right), %= O\left(\|\bm{x}-\bm{x}'\|_1 \cdot \left( \eta + \left(1 + \frac{1}{\eta}\right) \log^2 n \right)\right).
	\]
\end{lemma}

% \nvnote{I am not expressing the TV distance above in $O()$ notation because, I think that different terms could dominate depending on the regime. We can't multiply the three terms and use the resulting expression since some of the terms could be very possibly smaller than $1$.}
The proof of Lemma~\ref{lem:exponential-mechanism} is deferred to Appendix~\ref{sec:exponential-appendix}.
By setting $\eta = \log n/ \epsilon$ and replacing $t$ with $t \log n$, we get the following:
\begin{lemma}\label{lem:exponential-mechanism-log}
	Let $\epsilon > 0$. There exists an algorithm $\mathcal{A}_\epsilon$ such that, given a vector $\bm{x} \in \mathbb{R}^n$ outputs $i \in [n]$ such that
	\[
	\Pr_{i \sim \mathcal{A}_\epsilon(\bm{x})}\left[\bm{x}(i) \geq \mathsf{OPT} + \epsilon (1+ t)\right] \leq n^{-t},
	\]
	for any $t > 0$, where $\mathsf{OPT} = \min_{i \in [n]}\bm{x}(i)$. Moreover, for all $\bm{x}' \in \mathbb{R}^n$, we have
	\[
	d_{\mathrm{TV}}(\mathcal{A}_\epsilon(\bm{x}),\mathcal{A}_\epsilon(\bm{x}')) = O\left(\|\bm{x}-\bm{x}'\|_1 \cdot \frac{\log n}{\epsilon}\right).
	\]
\end{lemma}

\section{Warm Up: Minimum Spanning Forest}\label{sec:spanning-tree}
To get intution about average sensitivity of algorithms, we start with the \emph{minimum spanning forest problem}.
In this problem, we are given a weighted graph $G=(V,E,w)$, where $w: E \to \mathbb{R}$ is a weight function on edges, and we want to find a forest of the minimum total weight including all the vertices.

Recall that Kruskal's algorithm~\cite{Kruskal:1956} works as follows: Iterate over edges in the order of increasing weights, where we break ties arbitrarily.
At each iteration, add the current edge to the solution if it does not form a cycle with the edges already added.
The following theorem states that this simple and deterministic algorithm is stable-on-average.
\begin{theorem}
  The average sensitivity of Kruskal's algorithm is $O(n/m)$.
\end{theorem}
\begin{proof}
  Let $G=(V,E)$ be the input graph and $T$ be the spanning forest obtained by running Kruskal's algorithm on $G$.
  We consider how the output changes when we remove an edge $e \in E$ from $G$.

  If the edge $e$ does not belong to $T$, clearly the output of Kruskal's algorithm on $G-e$ is also $T$.

  Suppose that the edge $e$ belongs to $T$.
  Let $T_1$ and $T_2$ be the two trees rooted at the endpoints of $e$ obtained by removing $e$ from $T$.
  If $G-e$ is not connected, that is, $e$ is a bridge in $G$, then Kruskal's algorithm outputs $T_1 \cup T_2$ on $G-e$.
  If $G-e$ is connected, then let $e'$ be the first edge considered by Kruskal's algorithm among all the edges connecting $G[V(T_1)]$ and $G[V(T_2)]$, where $V(T_i)$ is the vertex set of $T_i$ for $i \in [2]$.
  Then, Kruskal's algorithm outputs $T_1 \cup T_2 \cup \{e'\}$ on $G-e$.
  It follows that the Hamming distance between $T$ and the output of the algorithm on $G-e$ is at most $2$.

  Therefore, the average sensitivity of Kruskal's algorithm is at most
  \[
    \frac{m-|T|}{m}\cdot 0 + \frac{|T|}{m}\cdot 2 = O\left(\frac{n}{m}\right). \qedhere
  \]
\end{proof}
Indeed, it is not hard to show a matching lower bound.
\begin{theorem}
  The average sensitivity of a (possibly randomized) algorithm for the minimum spanning forest problem is $\Omega(n/m)$.
\end{theorem}
\begin{proof}
  Let $\mathcal{A}$ be an algorithm for the minimum spanning forest problem, and let $G=(V,E)$ be a connected graph with $n$ vertices and $m$ edges.
  For each $e \in E$, let $\mu_e$ be the probability distribution over $\mathcal{F}(G) \times \mathcal{F}(G-e)$ such that $d_{\mathrm{EM}}(\mathcal{A}(G),\mathcal{A}(G-e)) = \E_{(F,F_e) \sim \mu_e}|F \triangle F_e|$, where $\mathcal{F}(G)$ is the set of spanning forests of $G$.
  Note that the marginal distribution of $\mu_e$ on the first coordinate is identical for all $e \in E$.
  %\nvnote{I am not sure I get this above statement I think that the word "uniform' needs to be replaced with "identical"?.}
  Let $\mu_{e,F}$ be the distribution over $\mathcal{F}(G-e)$ defined as (the second coordinate of) the distribution $\mu_e$ conditioned on the first coordinate being $F$.
  Then, we have
  \begin{align*}
    & \E_{e \sim E}\left[d_{\mathrm{EM}}\bigl(\mathcal{A}(G),\mathcal{A}(G-e)\bigr)\right]
    = \E_{e \sim E}\E_{(F,F_e) \sim \mu_e}|F \triangle F_e|
    = \E_{F \sim \mathcal{F}(G)} \E_{e \sim E} \E_{F_e \sim \mu_{e,F}}|F \triangle F_e| \\
    & = \E_{F \sim \mathcal{F}(G)} \left[ \frac{1}{m}\sum_{e \in E}  \E_{F_e \sim \mu_{e,F}}|F \triangle F_e| \right]
    \geq \E_{F \sim \mathcal{F}(G)} \left[ \frac{1}{m}\sum_{e \in F}  \E_{F_e \sim \mu_{e,F}}|F \triangle F_e| \right]
    \geq \E_{F \sim \mathcal{F}(G)} \left[ \frac{1}{m}\sum_{e \in F}  \E_{F_e \sim \mu_{e,F}}1 \right]\\
    & = \frac{n-1}{m} = \Omega\left(\frac{n}{m}\right),
  \end{align*}
  where in the second inequality we used the fact that $e \in F$ and $e \not \in F_e$.
\end{proof}

In Appendix~\ref{sec:prim}, we show that Prim's algorithm, another classical algorithm for the minimum spanning forest problem, has average sensitivity $\Omega(m)$ for a certain natural tie breaking rule.
We mention that this lower bound holds even for unweighted graphs. 
%!TEX root=./stabilityOfAlgorithms.tex

\section{Global Minimum Cut}\label{sec:minimum-cut}
For a graph $G = (V,E)$ and a vertex set $S \subseteq V$, we define $\mathsf{cost}(G,S)$ to be the number of edges in $E$ that cross the cut $(S, V\setminus S)$.
Then in the \emph{global minimum cut problem}, given a graph $G=(V,E)$, we want to compute a vertex set $\emptyset \subsetneq S \subsetneq V$ that minimizes $\mathsf{cost}(G,S)$.
In this section, we discuss upper and lower bounds on the average sensitivity for the global minimum cut problem.

\subsection{Upper bound}\label{subsec:minimum-cut-upper-bound}

In this section, we show the following.
\begin{theorem}\label{thm:min-cut}
  For $\eps > 0$, there exists a polynomial time algorithm for the global minimum cut problem with approximation ratio $2+\eps$ and average sensitivity $n^{O(1/\eps\mathsf{OPT})}$.
\end{theorem}

%The inputs to our algorithm (Algorithm~\ref{alg:mincut-stable-algorithm}) are undirected graphs on $n$ vertices and the outputs are vertex sets.

Let $\mathsf{OPT}$ be the minimum size of a cut in $G$.
Our algorithm enumerates cuts of small size and then output a vertex set $S$ with probability $\exp(-\alpha \cdot \mathsf{cost}(G,S))$ for a suitable $\alpha$.
See Algorithm~\ref{alg:mincut-stable-algorithm} for details.

\begin{algorithm}
  \caption{\textsc{Stable Algorithm for Global Minimum Cut}}\label{alg:mincut-stable-algorithm}
  \Input{undirected graph $G = (V,E)$, $\eps > 0$}
    Compute the value $\mathsf{OPT}$\;
    Let $\alpha \gets \frac{(2+1/\eps)\log n}{\mathsf{OPT}}$ denote a parameter\;
    Enumerate all cuts of size at most $(2+7\eps)\mathsf{OPT}+2\eps$\;
    Sample a vertex set $S$ (from among the cuts enumerated) with probability proportional to $\exp(-\alpha \cdot \cost(G, S))$\;
    \Return $S$.
\end{algorithm}

The approximation ratio of the Algorithm~\ref{alg:mincut-stable-algorithm} is $2+9\eps$: It clearly holds when $\mathsf{OPT} \geq 1$, and it also holds when $\mathsf{OPT} = 0$ because we only output a cut of size zero (for $\eps < 1/2$).
The following theorem due to Karger~\cite{K93} directly implies that it runs in time polynomial in the input size for any constant $\eps>0$.
\begin{theorem}[\cite{K93}]\label{thm:kargers-mincut-theorem}
  Given a graph $G$ on $n$ vertices with the minimum cut size $c$ and a parameter $\alpha \geq 1$, the number of cuts of size at most $\alpha \cdot c$ is at most $n^{2\alpha}$ and can be enumerated in time polynomial (in $n$) per cut.
\end{theorem}

\noindent We now show that Algorithm~\ref{alg:mincut-stable-algorithm} is stable-on-average.
% The neighbors of a graph $G = (V, E)$ are graphs $G_e = (V, E \setminus \{e\})$ for $e \in E$.
\begin{lemma}\label{lem:min-cut-sensitivity}
  The average sensitivity of Algorithm~\ref{alg:mincut-stable-algorithm} is at most
  \[
    \beta(G) =\frac{n}{m} \cdot n^{(2+1/\eps)/\mathsf{OPT}} \cdot \left(\left(2+7\eps\right)\mathsf{OPT} + 2\eps \right) + o(1).
  \]
\end{lemma}
As we have $\mathsf{OPT} \leq 2m/n$, the average sensitivity can be bounded by $n^{O(1/\eps\mathsf{OPT})}$, and Theorem~\ref{thm:min-cut} follows by replacing $\eps$ with $\eps/9$.
\begin{proof}
  If $\mathsf{OPT} = 0$, then the claim trivially holds because the right hand size is infinity.
  Hence in what follows, we assume $\mathsf{OPT} \geq 1$.

  Let $\cA$ denote Algorithm~\ref{alg:mincut-stable-algorithm}.
  Consider an (inefficient) algorithm $\cA'$ that on input $G$, outputs a cut $S \subseteq V$ (from among all the cuts in $G$) with probability proportional to $\exp(-\alpha \cdot \cost(G, S))$.
  For a graph $G = (V,E)$, let $\cA(G)$ and $\cA'(G)$ denote the output distribution of algorithms $\cA$ and $\cA'$ on input $G$, respectively.
  For $G = (V,E)$ and $S \subseteq V$, let $p_G(S)$ and $p'_G(S)$ be shorthands for the probabilities that $S$ is output on input $G$ by algorithms $\cA$ and $\cA'$, respectively.

  We first bound the earth mover's distance between $\cA(G)$ and $\cA'(G)$ for a graph $G= (V,E)$.
  To this end, we define
  \[
    Z = \mathop{\sum_{S \subseteq V: \mathsf{cost}(G,S) \le \mathsf{OPT}+ b} \exp(-\alpha\cdot\cost(G,S))},
    \text{ and }
    Z' = \sum_{S \subseteq V} \exp(-\alpha\cdot\cost(G,S))
  \]
  where $b = (1+7\eps)\mathsf{OPT} + 2\eps$.
  Note that $Z \leq Z'$ and the quantity $\frac{Z' - Z}{Z'}$ is the total probability mass assigned by algorithm $\cA'$ to cuts $S \subseteq V$ such that $\mathsf{cost}(G,S) > \mathsf{OPT} + b$.

  Now, we start with $\cA'(G)$.
  For each $S \subseteq V$ such that $\mathsf{cost}(G,S) \le \mathsf{OPT}+ b$, keep at least $\frac{Z}{Z'} \cdot p'_G(S)$ mass with a cost of $0$ and move a mass of at most $p'_G(S) - \frac{Z}{Z'} \cdot p'_G(S)$ at a cost of $n \cdot (p'_G(S) - \frac{Z}{Z'} \cdot p'_G(S))$.
  For each $S \subseteq V$ such that $\mathsf{cost}(G,S) > \mathsf{OPT}+ b$, we move a mass of $p'_G(S)$ at a cost of $n \cdot p'_G(S)$.
  The total cost of moving masses is then equal to:
  \begin{align*}
  	d_{\mathrm{EM}}\left(\cA(G),\cA'(G)\right) &\le n\cdot \sum_{S \subseteq V: \mathsf{cost}(G,S) \le \mathsf{OPT}+ b} p'_G(S)\left(1 - \frac{Z}{Z'}\right) + n \cdot \sum_{S \subseteq V: \mathsf{cost}(G,S) > \mathsf{OPT}+ b} p'_G(S) \\
    & = \frac{n(Z'-Z)}{Z'} \left(\sum_{S \subseteq V: \mathsf{cost}(G,S) \le \mathsf{OPT}+ b} p'_G(S) + 1\right) \\
    & \leq \frac{2n(Z'-Z)}{Z'}.
  \end{align*}

  %Using Karger's theorem  and the fact that a cut of size $\mathsf{OPT}$ exists in $G$, we can see that $\cA'$ assigns a probability mass of at most $\exp(-\alpha t)$ to each cut of cost $\mathsf{OPT} + t$.
  Let $n_t$ stand for the number of cuts of cost at most $\mathsf{OPT} + t$ in $G$.
  By Karger's theorem (Theorem~\ref{thm:kargers-mincut-theorem}), we have that $n_t \le n^{2 + 2t/\mathsf{OPT}}$.
  Then, we have
  \begin{align*}
  	\frac{Z'-Z}{Z'} & \leq \sum_{t > b} \exp(-\alpha t)\cdot (n_t - n_{t-1}) \le (\exp(\alpha) - 1)\cdot \sum_{t > b} \exp(-\alpha t) n_t \\
  	& \le (\exp(\alpha) - 1)n^2 \cdot \sum_{t > b} n^{2t/\mathsf{OPT}} \cdot \exp(-\alpha t)\\
  	&\le (\exp(\alpha) - 1)n^2 \cdot \sum_{t > b} n^{-t/\eps \mathsf{OPT}}
  	\le (\exp(\alpha) - 1)n^2 \cdot \frac{n^{-(b+1)/\eps \mathsf{OPT}}}{1 - n^{-1/\eps \mathsf{OPT}}}\\
  	&= \left(n^{(2+1/\eps)/\mathsf{OPT}}-1\right) \cdot \left(1+\frac{1}{n^{1/\eps\mathsf{OPT}}-1}\right) \cdot \frac{n^2}{n^{(b+1)/\eps\mathsf{OPT}}} \\
    & \leq  n^{(2+1/\eps)/\mathsf{OPT}} \cdot \left(1+\frac{\eps n}{ \log n}\right) \cdot \frac{n^2}{n^{(b+1)/\eps \mathsf{OPT}}} \\
  	& = O\left(\frac{\eps n^{3+(2 + 1/\eps)/\mathsf{OPT}}}{ n^{(b+1)/\eps \mathsf{OPT}}}\right) = O\left(\frac{\eps}{n^{4+1/\eps}}\right).
  \end{align*}
  The last inequality above follows from our choice of $b$.
  Therefore, the earth mover's distance between $\cA(G)$ and $\cA'(G)$ is $d_{\mathrm{EM}}\bigl(\cA(G),\cA'(G)\bigr) \le O(\frac{\eps}{ n^{3+1/\eps}})$.

  In addition, we can bound the expected size of the cut output by $\cA'$ on $G$ as follows.
  The total probability mass assigned by algorithm $\cA'$ to cuts of size larger than $\mathsf{OPT} + b$ is equal to $\frac{Z' - Z}{Z'} = O\left(\frac{\eps}{n^{4+1/\eps}}\right)$.
  Hence, the expected size of the cut output by $\cA'$ on $G$ is at most $\mathsf{OPT} + b + m\cdot O(\frac{\eps}{ n^{4+1/\eps}}) = (2+7\eps)\mathsf{OPT} + 2\eps + O(\frac{\eps m}{n^{4+1/\eps}})$.

  We now bound the earth mover's distance between $\cA'(G)$ and $\cA'(G-e)$ for an arbitrary edge $e \in E$.
  Let $Z'_e$ denote the quantity $\mathop{\sum_{S \subseteq V} \exp(-\alpha\cdot\cost(G-e,S))}$.
  Since the cost of every cut in $G-e$ is at most the cost of the same cut in $G$, we have that $Z' \le Z'_e$ and therefore,
  $$p'_G(S) = \frac{\exp(-\alpha \cdot \mathsf{cost}(G,S))}{Z'} \le \frac{\exp(\alpha \cdot \mathsf{cost}(G-e,S))}{Z'_e} \cdot \frac{Z'_e}{Z'} = p'_{G-e}(S)\cdot \frac{Z'_e}{Z'}.$$

  We transform $\cA'(G)$ into $\cA'(G-e)$ as follows.
  For each $S \subseteq V$, we leave a probability mass of at most $p'_{G-e}(S)$ at $S$ with zero cost and move a mass of $\max\{0, p'_G(S) - p'_{G-e}(S)\}$ to any other point at a cost of at most $n \cdot \max\{0, p'_G(S) - p'_{G-e}(S)\} \le n \cdot \left(\frac{Z'_e}{Z'} - 1\right)\cdot p'_G(S)$.
  Hence,
  \begin{align*}
  	d_{\mathrm{EM}}\left(\cA'(G), \cA'(G-e)\right) \le n \cdot \left(\frac{Z'_e}{Z'} - 1\right) \cdot \sum_{S \subseteq V} p'_G(S) = n \cdot \left(\frac{Z'_e}{Z'} - 1\right).
  \end{align*}

  By the triangle inequality, the earth mover's distance between $\cA(G)$ and $\cA(G-e)$ can be bounded as
  \begin{align*}
  d_{\mathrm{EM}}\bigl(\cA(G), \cA(G-e)\bigr) &\le d_{\mathrm{EM}}\left(\cA(G), \cA'(G)\right) + d_{\mathrm{EM}}\left(\cA'(G), \cA'(G-e)\right) + d_{\mathrm{EM}}\left(\cA'(G-e), \cA(G-e)\right)\\
   &\le n \cdot \left(\frac{Z'_e}{Z'} - 1\right) + O\left(\frac{2\eps}{n^{2+1/\eps}}\right).
  \end{align*}

  Hence, the average sensitivity of $\cA$ is bounded as:
  \begin{align*}
  	\beta(G) &= \mathop{\E_{e \sim E}} d_{\mathrm{EM}}\bigl(\cA(G), \cA(G-e)\bigr)
  	\le O\left(\frac{2\eps}{n^{3+1/\eps}}\right) + n \cdot \E_{e \sim E} \left(\frac{Z'_e}{Z'} - 1\right)\\
  	&= O\left(\frac{2\eps}{n^{3+1/\eps}}\right) + \frac{n}{mZ'} \sum_{e \in E} (Z'_e - Z')\\
  	&= O\left(\frac{2\eps}{n^{3+1/\eps}}\right) + \frac{n}{mZ'} \sum_{e \in E} \sum_{S \subseteq V: e \text{ crosses }S} \exp(-\alpha \cdot \mathsf{cost}(G-e, S)) - \exp(-\alpha \cdot \mathsf{cost}(G, S))\\
  	&= O\left(\frac{2\eps}{n^{3+1/\eps}}\right) + \frac{n(\exp(\alpha) - 1)}{mZ'} \sum_{e \in E} \sum_{S \subseteq V: e \text{ crosses }S} \exp(-\alpha \cdot \mathsf{cost}(G, S))\\
  	&= O\left(\frac{2\eps}{n^{3+1/\eps}}\right) + \frac{n(\exp(\alpha) - 1)}{m} \sum_{S \subseteq V} \mathsf{cost}(G,S) \cdot \frac{\exp(-\alpha \cdot \mathsf{cost}(G, S))}{Z'}.
  \end{align*}

  The summation in the second term above is equal to the expected size of the cut output by algorithm $\cA'$ on input $G$.
  We argued that it is at most $(2+7\eps)\mathsf{OPT} + 2\eps + O(\frac{\eps m}{n^{4+1/\eps}})$.
  Hence, the average sensitivity of $\cA$ is at most
  \begin{align*}
    & \frac{n}{m} \cdot n^{(2+1/\eps)/\mathsf{OPT}} \cdot \left(\left(2+7\eps\right)\mathsf{OPT} + 2\eps\right) + O\left(\frac{\eps n^{(2+1/\eps)/\mathsf{OPT}} + 2}{n^{3+1/\eps}}\right) \\
    & =
   \frac{n}{m} \cdot n^{(2+1/\eps)/\mathsf{OPT}} \cdot \left(\left(2+7\eps\right)\mathsf{OPT} + 2\eps\right) + o(1)
  \end{align*}
  as $\mathsf{OPT} \geq 1$.
\end{proof}

\subsection{Lower bound}\label{subsec:minimum-cut-lower-bound}
In this section, we show that the average sensitivity of the algorithm given in Section~\ref{subsec:minimum-cut-upper-bound} is almost tight.
%We say that an algorithm for the global minimum cut problem is a non-trivial approximation algorithm if it outputs a cut of size zero when the graph is disconnected.
%Note that any approximation algorithm is a non-trivial approximation algorithm
Specifically, we show the following.
\begin{theorem}\label{thm:min-cut-lower-bound}
  Any algorithm for the global minimum cut problem with no additive error (and possibly an arbitrary large multiplicative error) has average sensitivity $\Omega(n^{1/\mathsf{OPT}}/\mathsf{OPT}^2)$ if $\mathsf{OPT} = o(\sqrt{n})$.
\end{theorem}
\begin{proof}
  We first show a lower bound for the case $\mathsf{OPT}=1$.
  Let $\mathcal{A}$ be an arbitrary algorithm for the global minimum cut problem with no additive error and let $G=([n+1],E)$ be a path on $n+1$ vertices, where $E = \{(i,i+1) : i \in [n]\}$.
  Note that for any $i \in [n]$, the graph $G-(i,i+1)$ is disconnected and $\mathcal{A}$ must output a vertex set $[i]$ or $[n+1] \setminus [i]$.
  For a vertex set $S \subseteq [n+1]$, let $p_S$ be the probability that $\mathcal{A}$ on $G$ outputs $S$.
  Then, the average sensitivity of $\mathcal{A}$ on $G$ is
  \begin{align}
    \E_{e \sim E}\left[ d_{\mathrm{EM}}(\mathcal{A}(G),\mathcal{A}(G-e))\right] =
   \frac{1}{n} \sum_{i \in [n]} \sum_{S \subseteq [n+1]} p_S \cdot \min\{d_{\mathrm{Ham}}(S,[i]),d_{\mathrm{Ham}}(S,[n+1]  \setminus [i] )\}.   \label{eq:min-cut-lower-bound}
  \end{align}
  Note that if two sets $S,T \subseteq [n+1]$ satisfy $|S| \leq |T| - n/10$ or $|S| \geq |T| + n/10$, then $d_{\mathrm{Ham}}(S,T) \geq n/10$ holds.
  Hence, we have $d_{\mathrm{Ham}}(S,[i]) \geq n/10$ for at least a $4/5$-fraction of $i \in [n]$.
  Similarly, we have $d_{\mathrm{Ham}}(S,[n+1] \setminus [i] ) \geq n/10$ for at least a $4/5$-fraction of $i \in [n]$.
  It follows that we have $\min\{d_{\mathrm{Ham}}(S,[i]),d_{\mathrm{Ham}}(S,[n+1] \setminus [i])\} \geq n/10$ for at least a $3/5$-fraction of $i \in [n]$.
  Then, we have
  \begin{align}
    \eqref{eq:min-cut-lower-bound} \geq
    \frac{3}{5} \cdot \frac{n}{10} \cdot \sum_{S \subseteq [n]}p_S = \frac{3n}{50} = \Omega(n).
    \label{eq:min-cut-lower-bound-2}
  \end{align}

  We now consider the case $t := \mathsf{OPT} \geq 2$.
  Consider a multigraph $G_t = ([n+1],E_t)$, where $E_t$ contains $t$ copies of the edge $(i,i+1)$ for every $i \in [n]$.
  For $k = {(tn)}^{1-1/t}$, the $k$-average sensitivity of $\mathcal{A}$ on $G$ without replacement is
  \begin{align}
  & \E_{\{e_1,\ldots,e_k\} \sim \binom{E_t}{k} }\left[d_{\mathrm{EM}}(\mathcal{A}(G),\mathcal{A}(G-\{e_1,\ldots,e_k\})) \right] \nonumber \\
  & \geq
  \E_{\{e_1,\ldots,e_k\} \sim \binom{E_t}{k}}\left[d_{\mathrm{EM}}(\mathcal{A}(G),\mathcal{A}(G-\{e_1,\ldots,e_k\})) \mid \mathcal{A}(G-\{e_1,\ldots,e_k\}) \text{ has two components}\right] \times \nonumber \\
  & \qquad \Pr_{\{e_1,\ldots,e_k\} \sim \binom{E_t}{k}} \left[\mathcal{A}(G-\{e_1,\ldots,e_k\}\text{ has two components}) \right].
  \label{eq:min-cut-lower-bound-3}
  \end{align}
  The first factor of~\eqref{eq:min-cut-lower-bound-3} is exactly equal to~\eqref{eq:min-cut-lower-bound}, which is $\Omega(n)$ by~\eqref{eq:min-cut-lower-bound-2}.
  Now we bound the second factor.
  For every $i \in [n]$, the probability that we cut all the edges between $i$-th  and $(i+1)$-th vertices is $\binom{k}{t}/\binom{tn}{t}$ from the property of the hypergeometric distribution.
  For every distinct $i,j \in [n]$, the probability that we cut all the edges between $i$-th and $(i+1)$-th vertices and all the edges between $j$-th and $(j+1)$-th vertices is $\binom{k}{2t}/\binom{tn}{2t}$.
  By the inclusion-exclusion principle, the probability that $G-\{e_1,\ldots,e_k\}$ has exactly two components is at least
  \begin{align*}
    & n \frac{\binom{k}{t}}{\binom{tn}{t}} - \binom{n}{2}\frac{\binom{k}{2t}}{\binom{tn}{2t}}
    \geq n \left(1-\frac{t^2}{k}\right){\left(\frac{k}{tn}\right)}^t - \binom{n}{2} \frac{n}{n-4t} {\left(\frac{k }{tn}\right)}^{2t} \tag{By $\left(1-\frac{k^2}{n}\right) \frac{n^k}{k!} \leq \binom{n}{k} \leq \frac{n^k}{k!}$ }\\
    & = n \left(1-\frac{t^2}{k}\right)\frac{1}{tn} - \binom{n}{2} \frac{n}{n-4t} \frac{1}{t^2n^2}
    \geq \frac{3}{4t} - \frac{1}{t^2} \geq \frac{1}{4t} = \Omega\left(\frac{1}{t}\right),
  \end{align*}
  where we used the fact that $2 \leq t = o(\sqrt{n})$.
  Hence, we have $\eqref{eq:min-cut-lower-bound-3} = \Omega(n/t)$.
  By Theorem~\ref{thm:sensitivity-against-deleting-multiple-edges}, the average sensitivity $\beta$ of $\mathcal{A}$ on $G$ must satisfy
  \[
    \beta \cdot {(tn)}^{1-1/t} \geq \Omega\left(\frac{n}{t}\right),
  \]
  which implies $\beta \geq \Omega(n^{1/t}/t^2)$.
\end{proof}

The proof of the following theorem is implicit in the first part of the proof of Theorem~\ref{thm:min-cut-lower-bound}.
\begin{theorem}\label{thm:exact-mincut-lowerbound}
	 Any algorithm that exactly outputs the global minimum cut has average sensitivity $\Omega(n)$.
\end{theorem}

\section{Minimum \texorpdfstring{$s$}{s}-\texorpdfstring{$t$}{t} Cut}\label{sec:s-t-min-cut}

In this section, we design a stable-on-average algorithm for the minimum $s$-$t$ cut problem.
We say that a pair $\{u,v\} \in \binom{V}{2}$ is \emph{cut by $S \subseteq V$} if $u \in S$ and $v \in V \setminus S$, or vice versa. %$u \in V \setminus S$ and $v \in S$.
The \emph{cut size} of a vertex set $S \subseteq V$ in a graph $G=(V,E)$ is the number of edges $e \in E$ cut by $S$.
In the minimum $s$-$t$ cut problem, given a graph $G = (V,E)$ and two vertices $s,t \in V$, we want to find a \emph{minimum $s$-$t$ cut}, that is, a vertex set $S$ with $s \in S$ and $t \not \in S$ that has the minimum cut size.
We show the following.
\begin{theorem}\label{thm:s-t-min-cut}
  There exists a polynomial time algorithm for the minimum $s$-$t$ cut problem with additive error $O(n^{2/3})$ and average sensitivity $O(n^{2/3})$.
\end{theorem}

In Section~\ref{subsec:min-cut-algorithm}, we describe our LP relaxation for the minimum $s$-$t$ cut problem and introduce a notion of average sensitivity for algorithms that solve the LP.
We then provide a stable-on-average LP solver in Section~\ref{subsec:regularized-LP} and discuss a rounding procedure in Section~\ref{subsec:s-t-min-cut-rounding}.
We prove Theorem~\ref{thm:s-t-min-cut} in Section~\ref{subsec:s-t-min-cut-analysis}.

\subsection{LP Relaxation and Average Sensitivity}\label{subsec:min-cut-algorithm}

Our algorithm is based on an LP relaxation for the minimum $s$-$t$ cut problem.
For each pair of vertices $\{u,v\}\in \binom{V}{2}$, we introduce a variable $\bm{d}(\{u,v\})$, which we regard as a distance between $u$ and $v$.
Roughly speaking, $\bm{d}(\{u,v\})$ is $1$ if $u$ and $v$ are on different sides of an $s$-$t$ cut. $\bm{d}(\{u,v\})$ is $0$ otherwise.
For notational simplicity, we often write $\bm{d}(u,v)$ to denote $\bm{d}(\{u,v\})$.
Intuitively, the distance between $s$ and $t$ should be at least one, and the distance $\bm{d}(\cdot,\cdot)$ should satisfy the triangle inequality.
Our LP relaxation is the following.
\begin{align}
  \begin{array}{lll}
  \text{minimize} & \mathrm{LP}_G(\bm{d}) :=  \displaystyle \sum_{\{u,v\} \in E}\bm{d}(u,v), \\
  \text{subject to} & \bm{d}(s,t) = 1, \\
  & \bm{d}(u,v) \leq \bm{d}(u,w) + \bm{d}(w,v) , & \displaystyle \forall \{u,v,w\} \in \binom{V}{3}\\
  & 0 \leq \bm{d}(u,v) \leq 1. & \displaystyle \forall \{u,v\}\in \binom{V}{2}
  \end{array}\label{eq:LP-min-cut}
\end{align}
It is easy to check that LP~\eqref{eq:LP-min-cut} is indeed a relaxation for the minimum $s$-$t$ cut problem.
Let $S \subseteq V$ be an $s$-$t$ cut.
Then for each $\{u,v\} \in \binom{V}{2}$, we set $\bm{d}(u,v) = 1$ if $\{u,v\}$ is cut by $S$, and set $\bm{d}(u,v) = 0$ otherwise.
It is clear that $\sum_{\{u,v\} \in E}\bm{d}(u,v) = |\{\{u,v\} \in E: \{u,v\} \text{ is cut by }S\}|$ is the cut size of $S$ and that $\bm{d}$ satisfies all the constraints.
%\nvnote{The above is not a fully specified solution, right? What about pairs $(u,v)$ that are not edges in the graph? } \ynote{$\{u,v\}$ should have been taken from $\binom{V}{2}$ instead of $E$. Fixed.}

%We intentionally imposed the triangle inequality only for triples of the form $\{s,u,v\} \in \binom{V}{3}$ so that the following holds.
%\begin{proposition}\label{pro:triangle-inequality}
%  Let $\bm{d}$ be an optimal solution to~\eqref{eq:LP-min-cut}.
%  For any $\{u,v\} \in \binom{V}{2}$, we have either $\bm{d}(u,v) = \bm{d}(s,u)-\bm{d}(s,v)$ or $\bm{d}(u,v) = \bm{d}(s,v)-\bm{d}(s,u)$.
%\end{proposition}
%\nvnote{It is not clear what solution $\bm{d}$ is being referred to in the above proposition.} \ynote{Fixed.}

Now, we introduce a notion of the average sensitivity of an algorithm for solving~\eqref{eq:LP-min-cut}.
First, for a vertex $s \in V$, we define a norm $\|\cdot\|_s$ as $\|\bm{d}\|_s := \sqrt{\sum_{v \in V}{\bm{d}(s,v)}^2}$.
Then, we define the \emph{average sensitivity} of a (deterministic) algorithm $\mathcal{A}$ for solving LP~\eqref{eq:LP-min-cut} as
\[
  \E_{e \sim E}\|\mathcal{A}(G)-\mathcal{A}(G-e)\|_s.
\]
We use the norm $\|\cdot\|_s$ instead of the standard $\ell_2$ norm because our rounding procedure uses only $\bm{d}(s,v)\;(v\in V)$ (see Section~\ref{subsec:s-t-min-cut-rounding} for the details of the rounding procedure), and the former norm gives a better approximation guarantee than the latter.

\subsection{Stable-on-Average LP Solver}\label{subsec:regularized-LP}
In this section, we give a stable-on-average solver for LP~\eqref{eq:LP-min-cut}.
\begin{theorem}\label{thm:lp}
  For any $\eta > 0$, there exists a polynomial-time algorithm for solving LP~\eqref{eq:LP-min-cut} with average sensitivity $O(1/\sqrt{\eta})$ such that the output $\bm{d} \in {[0,1]}^{\binom{V}{2}}$ satisfies
  \[
    \sum_{e \in E}\bm{d}(e) \leq \mathsf{OPT} + \frac{\eta n}{2},
  \]
  where $\mathsf{OPT}$ is the optimal value of LP~\eqref{eq:LP-min-cut}.
\end{theorem}

Given a parameter $\eta > 0$, our algorithm solves the following regularized LP and then returns the optimal solution.
\begin{align}
  \begin{array}{lll}
  \text{minimize} & \mathrm{LP}_G^\eta(\bm{d}) := \displaystyle \sum_{\{u,v\} \in E}\bm{d}(u,v)  + \frac{\eta}{2}\|\bm{d}\|_s^2 \\
  \text{subject to} & \bm{d}(s,t) = 1, \\
  & \bm{d}(u,v) \leq \bm{d}(u,w) + \bm{d}(w,v) & \displaystyle \forall \{u,v,w\} \in \binom{V}{3}, \\
  & 0 \leq \bm{d}(u,v) \leq 1 & \displaystyle \forall \{u,v\}\in \binom{V}{2},
  \end{array}\label{eq:regularized-LP}
\end{align}
The only difference from LP~\eqref{eq:LP-min-cut} is that we have a regularization term $\eta\|\bm{d}\|_s^2/2$ in the objective function.
Note that $\mathrm{LP}_G^\eta$ is $\eta$-strongly convex with respect to the norm $\|\cdot\|_s$. %and Proposition~\ref{pro:triangle-inequality} still holds for~\eqref{eq:regularized-LP}.
Theorem~\ref{thm:lp} follows from Lemmas~\ref{lem:optimality-regLP} and~\ref{lem:sensitivity-regLP}, given below.
\begin{lemma}\label{lem:optimality-regLP}
  Let $\bm{d}^*,\bm{d} \in {[0,1]}^{\binom{V}{2}}$ be the optimal solutions to LPs~\eqref{eq:LP-min-cut} and~\eqref{eq:regularized-LP}, respectively.
  Then, we have
  \begin{align*}
    \sum_{\{u,v\}\in E}\bm{d}(e) \leq \sum_{\{u,v\}\in E}\bm{d}^*(e) + \frac{\eta n}{2}.
  \end{align*}
\end{lemma}
\begin{proof}
  We have
  \begin{align*}
      & \sum_{e\in E}\bm{d}(e)  - \sum_{e\in E}\bm{d}^*(e)
      = \left(\sum_{e\in E}\bm{d}(e)  + \frac{\eta}{2}\|\bm{d}\|_s^2\right) - \left(\sum_{e \in E}\bm{d}^*(e) + \frac{\eta}{2} \|\bm{d}^*\|_s^2\right) +
      \frac{\eta}{2}\left(\|\bm{d}^*\|_s^2 - \|\bm{d}\|_s^2\right) \\
      & \leq \frac{\eta}{2}\left(\|\bm{d}^*\|_s^2 - \|\bm{d}\|_s^2\right) \tag{by the optimality of $\bm{d}$ for LP~\eqref{eq:regularized-LP}} \\
      & \leq \frac{\eta}{2}\|\bm{d}^*\|_s^2  \\
      & \leq \frac{\eta n}{2}, \tag{by $\bm{d}^* (s,v) \leq 1$ for any $v \in V$}
  \end{align*}
  as desired.
\end{proof}

\begin{lemma}\label{lem:sensitivity-regLP}
  Let $G = (V,E)$ be a graph, $\{u,v\} \in E$ be an edge, and let $\bm{d}$ and $\bm{d}'$ be the optimal solutions to LPs~\eqref{eq:regularized-LP} for $G$ and $G':=G-\{u,v\}$, respectively.
  Then, $\|\bm{d}-\bm{d}'\|_s = O(1/\sqrt{\eta})$.
\end{lemma}
\begin{proof}
  By the $\eta$-strong convexity of $\mathrm{LP}_G^\eta$ with respect to $\|\cdot\|_s$, we have
  \[
      \mathrm{LP}_G^\eta(\bm{d}') - \mathrm{LP}_G^\eta(\bm{d}) \geq \langle\nabla \mathrm{LP}_G^\eta(\bm{d}), \bm{d}' - \bm{d}\rangle +  \eta\|\bm{d}'-\bm{d}\|_s^2
      \geq \eta\|\bm{d}-\bm{d}'\|_s^2,
  \]
  where the second inequality holds because $\bm{d}$ is an optimal solution.
  %\nvnote{Should we define the inner product described above? It has to be consistent with the definition of the norm, right?}\ynote{No, it shoudn't be. That part comes from the convexity of $\mathrm{LP}_G^\eta$.}
  Similarly for $\mathrm{LP}_{G'}^\eta$, we have
  \[
    \mathrm{LP}_{G'}^\eta(\bm{d}) - \mathrm{LP}_{G'}^\eta(\bm{d}') \geq \eta\norm{\bm{d}-\bm{d}'}_s^2.
  \]
  Therefore, we have
  \begin{align*}
    & 2\eta\norm{\bm{d} - \bm{d}'}_s^2
    \leq \mathrm{LP}_G^\eta(\bm{d}') - \mathrm{LP}_G^\eta(\bm{d}) + \mathrm{LP}_{G'}^\eta(\bm{d}) - \mathrm{LP}_{G'}^\eta(\bm{d}') \\
    & = \sum_{e \in E}\left(\bm{d}'(e)-\bm{d}(e)\right) + \sum_{e \in E\setminus \{u,v\}}\left(\bm{d}(e)-\bm{d}'(e)\right)
    =\bm{d}'(u,v)-\bm{d}(u,v) \leq 1.
  \end{align*}
  % Now, we show $\bm{d}'(u,v)-\bm{d}(u,v) = O(\|\bm{d}'-\bm{d}\|_s)$ by considering the following four cases.
  % \begin{itemize}
  % \item If $\bm{d}'(u,v) = \bm{d}'(s,u)-\bm{d}'(s,v)$ and $\bm{d}(u,v) = \bm{d}(s,u)-\bm{d}(s,v)$, then we have
  % \begin{align*}
  %   & \bm{d}'(u,v)-\bm{d}(u,v) = (\bm{d}'(s,u)-\bm{d}(s,u)) + (\bm{d}'(s,v)-\bm{d}(s,v)) \\
  %   & \leq \sqrt{2} \cdot \sqrt{{(\bm{d}'(s,u)-\bm{d}(s,u))}^2 + {(\bm{d}'(s,v)-\bm{d}(s,v))}^2} \tag{by Cauchy-Schwarz}\\
  %   & = O\left(\|\bm{d}' - \bm{d}\|_s\right).
  % \end{align*}
  % \item If $\bm{d}'(u,v) = \bm{d}'(s,u)-\bm{d}'(s,v)$ and $\bm{d}(u,v) = - \bm{d}(s,u) + \bm{d}(s,v) $, then we have
  % \begin{align*}
  %   & \bm{d}'(u,v)-\bm{d}(u,v) = (\bm{d}'(s,u)-\bm{d}(s,v)) + (\bm{d}(s,u)-\bm{d}'(s,v)) \\
  %   & \leq (\bm{d}'(s,u)-\bm{d}(s,u)) + (\bm{d}(s,v)-\bm{d}'(s,v)) \tag{by $\bm{d}(s,v) - \bm{d}(s,u) \geq 0$} \\
  %   & = O\left(\|\bm{d}' - \bm{d}\|_s\right). \tag{similar argument to the first case}
  % \end{align*}
  % \item If $\bm{d}'(u,v) = -\bm{d}'(s,u)+\bm{d}'(s,v)$ and $\bm{d}(u,v) = \bm{d}(s,u)-\bm{d}(s,v)$, then $\bm{d}'(u,v)-\bm{d}(u,v) = O(\|\bm{d}' - \bm{d}\|_s)$ holds by a similar analysis to the first case.
  % \item If $\bm{d}'(u,v) = -\bm{d}'(s,u)+\bm{d}'(s,v)$ and $\bm{d}(u,v) = - \bm{d}(s,u) + \bm{d}(s,v)$, then $\bm{d}'(u,v)-\bm{d}(u,v) = O(\|\bm{d}' - \bm{d}\|_s)$ holds by a similar analysis to the second case.
  % \end{itemize}
  Hence, we have $\norm{\bm{d} - \bm{d}'}_s \leq O(1/\sqrt{\eta})$.
\end{proof}

% \nvnote{There seems to be an assumption here that $x(\pi(i)) \ge \bm{x}(\pi(i+1))$. I am unable to wrap my head around it.}
% \nvnote{Aren't we assuming above that all entries of $\bm{x}$ are distinct?}
\subsection{Rounding Procedure}\label{subsec:s-t-min-cut-rounding}

\begin{algorithm}[t!]
\caption{Rounding procedure for the minimum $s$-$t$ cut problem}\label{alg:thresh}
 \Procedure{\emph{\Call{Thresh}{$\bm{x}$}}}{
    \Input{$\bm{x} \in {[0,1]}^V$}
    Sample $\tau$ from $ [0,1]$ uniformly at random\;
    \Return the set $\{v \in V : \bm{x}(v) \geq \tau\}$\
  }
\end{algorithm}

Suppose we have obtained a solution $\bm{d}$ to LP~\eqref{eq:LP-min-cut}.
Then, we will round the vector $\bm{z} := {(\bm{d}(s,v))}_{v \in V} \in {[0,1]}^V$ using a thresholding procedure \Call{Thresh}{} (Algorithm~\ref{alg:thresh}), which returns a set of vertices $v \in V$ such that $\bm{d}(s,v)$ is at least a threshold $\tau$ sampled from $[0,1]$ uniformly at random.
We can also interpret \Call{Thresh}{$\bm{z}$} as follows.
Let $v_1,\ldots,v_n$ be the ordering of $V$ such that $\bm{z}(v_i) \geq \bm{z}(v_{i+1})$ for every $i \in [n-1]$.
Then, \Call{Thresh}{$\bm{z}$} outputs the set $\{v_1,\ldots,v_i\}$ with probability $\bm{z}(v_i)-\bm{z}(v_{i+1})$ for $i \in [n+1]$, where we define $\bm{z}(v_{n+1}) = 0$ for a dummy vertex $v_{n+1}$.

First, we analyze the solution quality of \Call{Thresh}{}.
\begin{lemma}\label{lem:thresholding-approximation}
  We have
  \[
    \E[\text{The cut size of } \Call{Thresh}{\bm{z}}] \leq \sum_{e \in E}\bm{d}(e),
  \]
  where $\bm{z} = {(\bm{d}(s,v))}_{v \in V}$.
\end{lemma}
\begin{proof}
  For each edge $e \in E$, the probability that it is cut by $S$ is
  \[
    |\bm{z}(u)- \bm{z}(v)|
    =
    |\bm{d}(s,u)- \bm{d}(s,v)|
    \leq \bm{d}(u,v),
  \]
  where $u,v \in V$ are the endpoints of $e$.
  The claim follows by the linearity of expectations.
\end{proof}
% \nvnote{Unclear what the dimension of $\bm{d}$ is in the above claim. In any case, if we associate an indicator variable $Y_e$ for each edge $e$ for whether $e$ is part of the minimum cut or not, $\Pr[Y_e = 1] = 1 -(\max\{d(s,u), d(s,v)\} - \min\{d(s,u), d(s,v)\})$, since $e = (u,v)$ is part of the cut only if one among $d(s,u)$ and $d(s,v)$ is less than $\tau$ and the other is at least $\tau$.}

Now, we bound the average sensitivity of \Call{Thresh}{} when only one coordinate differs.
\begin{lemma}\label{lem:thresholding-sensitivity-coordinatewise}
  Let $\bm{x} \in {[0,1]}^V$ and $\bm{x}' \in {[0,1]}^V$ be such that
  \[
    \bm{x}'(u) = \begin{cases}
    \bm{x}(u) + \Delta  & \text{if }u=v, \\
    \bm{x}(u)  & \text{otherwise},
    \end{cases}
  \]
  for some $v \in V$ and $0 \leq \bm{x}(v) \leq 1-\Delta$.
  Then, $d_{\mathrm{EM}}(\Call{Thresh}{\bm{x}},\Call{Thresh}{\bm{x}'}) \leq 2\Delta$.
\end{lemma}
%\nvnote{In the above, isn't is better to write $-\Delta \le \bm{x}(v) \le 1 - \Delta$?}
\begin{proof}
  We can assume $\Delta \geq 0$ as otherwise we can switch the roles of $\bm{x}$ and $\bm{x}'$.
  Starting with the vector $\bm{x}_0=\bm{x}$, we iteratively construct $\bm{x}_k \in {[0,1]}^V$ from $\bm{x}_{k-1}$ as
  \[
    \bm{x}_k(u) = \begin{cases}
      \displaystyle \min\left\{\bm{x}'(v), \min_{w \in V: \bm{x}_{k-1}(w) > \bm{x}_{k-1}(v)}\{\bm{x}_{k-1}(w)\}\right\} & \text{if } u = v, \\
      \bm{x}_{k-1}(u) & \text{otherwise}.
    \end{cases}
  \]
  Let $\ell$ be the smallest integer such that $\bm{x}_\ell(v) = \bm{x}'(v)$.
  Note that for every $k \in [\ell]$, there is an ordering $v_1,\ldots,v_n$ of $V$ such that both $\bm{x}_{k-1}(v_i) \geq \bm{x}_{k-1}(v_{i+1})$ and $\bm{x}_k(v_i) \geq \bm{x}_k(v_{i+1})$ hold for every $i \in [n-1]$.
%  permutation consistent with both $\bm{x}_k$ and $\bm{x}_{k-1}$.
  Note that $\sum_{k \in [\ell]}\|\bm{x}_k-\bm{x}_{k-1}\|_1 = \Delta$.

  Now we show that for each $k \in [\ell]$ we have $d_{\mathrm{EM}}(\Call{Thresh}{\bm{x}_k},\Call{Thresh}{\bm{x}_{k-1}}) \leq 2\|\bm{x}_k-\bm{x}_{k-1}\|_1$.
  Let $v_1,\ldots,v_n$ be an ordering of $V$ with the property mentioned above, and let $S = \{v_1,\ldots,v_{i-1}\}$ and $S' = \{v_1,\ldots,v_i\}$, where $i \in [n]$ is such that $v_i = v$.
  %$\pi$ be a permutation consistent with both $\bm{x}_k$ and $\bm{x}_{k-1}$, let $i = \pi^{-1}(v)$, and let $S = \{\pi(1),\ldots,\pi(i-1)\}$ and $S' = \{\pi(1),\ldots,\pi(i)\}$.
  %Then, the only difference in the output distributions of $\Call{Thresh}{\bm{x}_k}$ and $\Call{Thresh}{\bm{x}_{k-1}}$ is that the former outputs $S$ with probability $\bm{x}_k(\pi(i-1))-\bm{x}_k(\pi(i))$ and $S'$ with probability $\bm{x}_k(\pi(i))-\bm{x}_k(\pi(i+1))$ whereas the latter outputs $S$ with probability $\bm{x}_{k-1}(\pi(i-1))-\bm{x}_{k-1}(\pi(i))$ and $S'$ with probability $\bm{x}_{k-1}(\pi(i))-\bm{x}_{k-1}(\pi(i+1))$.
  Then, the only difference in the output distributions of $\Call{Thresh}{\bm{x}_k}$ and $\Call{Thresh}{\bm{x}_{k-1}}$ is that the former outputs $S$ with probability $\bm{x}_k(v_{i-1})-\bm{x}_k(v_i)$ and $S'$ with probability $\bm{x}_k(v_i)-\bm{x}_k(v_{i+1})$ whereas the latter outputs $S$ with probability $\bm{x}_{k-1}(v_{i-1})-\bm{x}_{k-1}(v_i)$ and $S'$ with probability $\bm{x}_{k-1}(v_i)-\bm{x}_{k-1}(v_{i+1})$.
  It follows that

  \begin{align*}
    & d_{\mathrm{EM}}(\Call{Thresh}{\bm{x}_k},\Call{Thresh}{\bm{x}_{k-1}})
    \le
    2\bigl( \bm{x}_k(v_i) - \bm{x}_{k-1}(v_i)\bigr) \cdot d_{\mathrm{Ham}}(S,S') \\
    & = 2\bigl( \bm{x}_k(v) - \bm{x}_{k-1}(v)\bigr) \cdot d_{\mathrm{Ham}}(S,S')
    = 2\|\bm{x}_k-\bm{x}_{k-1}\|_1.
  \end{align*}

  Then, we have
  \begin{align*}
    & d_{\mathrm{EM}}(\Call{Thresh}{\bm{x}},\Call{Thresh}{\bm{x}'})
    = d_{\mathrm{EM}}(\Call{Thresh}{\bm{x}_0},\Call{Thresh}{\bm{x}_{\ell}}) \\
    & \leq \sum_{k \in [\ell]}d_{\mathrm{EM}}(\Call{Thresh}{\bm{x}_k},\Call{Thresh}{\bm{x}_{k-1}})
    \le \sum_{k \in [\ell]} 2\|\bm{x}_k-\bm{x}_{k-1}\|_1 = 2\Delta. \qedhere
  \end{align*}
\end{proof}

\begin{corollary}\label{cor:thresholding-sensitivity}
  $d_{\mathrm{EM}}(\Call{Thresh}{\bm{x}},\Call{Thresh}{\bm{x}'}) \le 2\|\bm{x}-\bm{x}'\|_1$ for $\bm{x},\bm{x}' \in {[0,1]}^{V}$.
\end{corollary}
\begin{proof}
  The inequality can be obtained by iteratively applying Lemma~\ref{lem:thresholding-sensitivity-coordinatewise} to each coordinate of the vectors.
  %$d_{\mathrm{EM}}(\Call{Thresh}{\bm{x}},\Call{Thresh}{\bm{x}'}) = 2\|\bm{x}-\bm{x}'\|_1$ for $\bm{x},\bm{x}' \in {[0,1]}^V$.
\end{proof}

\subsection{Putting Things Together}\label{subsec:s-t-min-cut-analysis}

\begin{algorithm}[t!]
  \caption{Algorithm for the minimum $s$-$t$ cut problem}\label{alg:min-cut}
  \Procedure{\emph{\Call{MinCut}{$G,s,t,\epsilon$}}}{
    \Input{A graph $G=(V,E)$, two vertices $s,t \in V$.}
    Solve LP~\eqref{eq:LP-min-cut} using the algorithm given by Theorem~\ref{thm:lp} with parameter $\eta = n^{-1/3}$, and let $\bm{d} \in {[0,1]}^{\binom{V}{2}}$ be the solution obtained\;
    $\bm{z} \leftarrow {(\bm{d}(s,v))}_{v \in V}$\;
    \Return \Call{Thresh}{$\bm{z}$}.
  }
  % \Procedure{\emph{\Call{Thresh}{$\bm{x}$}}}{
  %   \Input{$\bm{x} \in {[0,1]}^V$}
  %   Sample $\tau$ from $ [0,1]$ uniformly at random\;
  %   \Return the set $\{v \in V : \bm{x}(v) \geq \tau\}$\
  % }
\end{algorithm}
Our algorithm is given in Algorithm~\ref{alg:min-cut}.
It simply computes a solution $\bm{d} \in {[0,1]}^{\binom{V}{2}}$ to LP~\eqref{eq:LP-min-cut} using the algorithm given in Theorem~\ref{thm:lp} with parameter $\eta = n^{-1/3}$, and then rounds $\bm{d}$ using the procedure $\Call{Thresh}{}$.

\begin{proof}[Proof of Theorem~\ref{thm:s-t-min-cut}]
  Let $\mathcal{A}$ be Algorithm~\ref{alg:min-cut}.
  First, we consider the approximation guarantee of $\mathcal{A}$.
  Let $\bm{d}^* \in \mathbb{R}^{\binom{V}{2}}$ be an optimal solution to LP~\eqref{eq:LP-min-cut}.
  By Theorem~\ref{thm:lp}, we get $\bm{d} \in {[0,1]}^{\binom{V}{2}}$ satisfying the constraints in LP~\eqref{eq:LP-min-cut} such that
  \[
    \sum_{e \in E}\bm{d}(e) \leq \mathsf{OPT} + \frac{\eta n}{2} = \mathsf{OPT} + O\left(n^{2/3}\right).
  \]
  By Lemma~\ref{lem:thresholding-approximation}, the expected cut size of the output set is the same.

  Now, we consider the average sensitivity of $\mathcal{A}$.
  Let $G=(V,E)$ be a graph, and let $\bm{d} \in {[0,1]}^{\binom{V}{2}}$ be the solution to LP~\eqref{eq:LP-min-cut} computed by $\mathcal{A}$ on $G$.
  For each edge $e \in E$, let $\bm{d}_e \in {[0,1]}^{\binom{V}{2}}$ be the solution to LP~\eqref{eq:LP-min-cut} computed by $\mathcal{A}$ on $G-e$.
  Then, we have
  \begin{align*}
    & \E_{e \sim E}d_{\mathrm{EM}}(\mathcal{A}(G),\mathcal{A}(G-e))
    \leq \E_{e \sim E} d_{\mathrm{EM}}\left(\Call{Thresh}{{(\bm{d}(s,v))}_{v \in V}},\Call{Thresh}{{(\bm{d}_e(s,v))}_{v \in V}}\right)\\
    & \leq 2\E_{e \sim E} \|{(\bm{d}(s,v))}_{v \in V} -{(\bm{d}_e(s,v))}_{v \in V}\|_1  \tag{by Corollary~\ref{cor:thresholding-sensitivity}} \\
    & = 2\sqrt{n}\E_{e \sim E} \|\bm{d}-\bm{d}_e\|_s  \\
    & = O\left(\sqrt{\frac{n}{\eta}}\right) \tag{by Theorem~\ref{thm:lp}} \\
    & = O\left(n^{2/3}\right). \qedhere
  \end{align*}
\end{proof}

\begin{comment}
\subsection{Lower Bounds}\label{subsec:min-cut-lower-bound}

The instance used to show lower bounds for the global minimum cut problem in~\cite{2019arXiv190403248V} can be reused to show the following same lower bound.
\begin{theorem}\label{thm:s-t-min-cut-lower-bound}
 Any algorithm for the $s$-$t$ minimum cut problem with no additive error (and possibly an arbitrary large multiplicative error) has average sensitivity $\Omega(n^{1/\mathsf{OPT}}/\mathsf{OPT}^2)$ if $\mathsf{OPT} = o(\sqrt{n})$.
\end{theorem}
\ynote{I'm not sure this is very useful.}
\end{comment}
% Theorem~\ref{thm:s-t-min-cut-lower-bound} implies a lower bound of $\widetilde{\Omega}(\log^c n) $ when $\mathsf{OPT} = O(\log n / \log \log^{c+2} n)$ for any $c > 0$.
% Theorem~\ref{thm:lp-log} shows that we can get a $O(\log \log^{c+2} n)$-approximation with average sensitivity of $O(\log n)$, which implies Theorem~\ref{thm:lp-log} is almost tight. \ynote{can we get rid of the $\mathsf{OPT}^2$ in the denominator?}

% \input{balanced-cut}
%!TEX root=./stabilityOfAlgorithms.tex

\section{Maximum Matching}\label{sec:matching}
A vertex-disjoint set of edges is called a \emph{matching}.
In the maximum matching problem, given a graph, we want to find a matching of the maximum size.
In this section, we describe several algorithms with low average sensitivity that approximate the maximum matching in a graph.

\subsection{Lexicographically smallest matching}
In this section, we describe an algorithm that computes a maximum matching in a graph with average sensitivity  at most $\mathsf{OPT}^2/m$ and prove Theorem~\ref{thm:matching-lexicographic}, where $\mathsf{OPT}$ is the maximum size of a matching.

First, we define some ordering among vertex pairs.
Then, we can naturally define the lexicographical order among matchings by considering a matching as a sorted sequence of vertex pairs.
Then, our algorithm simply outputs the lexicographically smallest matching.
Note that this can be done in polynomial time using Edmonds' algorithm~\cite{E65}.
\begin{theorem}\label{thm:matching-lexicographic}
  Let $\mathcal{A}$ be the algorithm that outputs the lexicographically smallest maximum matching.
  Then, the average sensitivity of $\mathcal{A}$ is at most $\mathsf{OPT}^2/m$, where $\mathsf{OPT}$ is the maximum size of a matching.
\end{theorem}
\begin{proof}
  For a graph $G=(V,E)$, let $M(G) \subseteq E$ be its lexicographically smallest maximum matching.
  As long as $e \not \in M$, we have $M(G) = M(G-e)$.
  Hence, the average sensitivity of the algorithm is at most
  \[
    \frac{\mathsf{OPT}}{m} \cdot \mathsf{OPT} + \left(1 - \frac{\mathsf{OPT}}{m}\right) \cdot 0 = \frac{\mathsf{OPT}^2}{m}.
    \qedhere
  \]
\end{proof}

\begin{remark}
Consider the path graph $P_n = (\{1,\ldots,n\}, E)$, where $E = \{(i, i+1): i \in [n-1]\}$. The average sensitivity of the above algorithm on $P_n$ is $\Omega(\frac{\mathsf{OPT}^2}{m})$.
Hence the above analysis of the average sensitivity is tight.
\end{remark}

\subsection{Greedy matching algorithm}\label{subsec:matching-greedy}

In this section, we analyze the average sensitivity of the randomized greedy algorithm (Algorithm~\ref{alg:matching-rand-greedy}) that outputs a maximal matching.
\begin{algorithm}[t!]
\caption{\textsc{Randomized Greedy Algorithm}}\label{alg:matching-rand-greedy}
\Input{undirected unweighted graph $G = (V,E)$}
Sample a uniformly random ordering $\pi$ of edges in $E$\;
Set $M \gets \emptyset$\;
Consider edges one by one according to $\pi$ and add an edge $(u,v)$ to $M$ only if both $u$ and $v$ are unmatched in $M$\;
\Return $M$.
\end{algorithm}
It is evident that Algorithm~\ref{alg:matching-rand-greedy} runs in polynomial time and that the matching it outputs has size at least $\frac{1}{2}$ the size of a maximum matching in the input graph.
\begin{theorem}\label{thm:greedy-sensitivity-guarantee}
Algorithm~\ref{alg:matching-rand-greedy} is a $\frac{1}{2}$-approximation algorithm for the maximum matching problem and has average sensitivity at most $1$.
\end{theorem}
\begin{proof}
For a permutation $\pi$ of edges in $E$, let $M_\pi(G)$ denote the matching obtained by running Algorithm~\ref{alg:matching-rand-greedy} on a graph $G$.
Using~\cite[Theorem 1]{CHK16}, we get that for every $e \in E(G)$, it holds that $\E_{\pi}\left[ \textsf{Ham}(M_\pi(G), M_\pi(G-e))\right] \le 1$.
This implies that the average sensitivity of Algorithm~\ref{alg:matching-rand-greedy} is at most $1$.
\end{proof}

A vertex set $S \subseteq V$ in a graph $G=(V,E)$ is called a \emph{vertex cover} if every edge in $E$ is incident to a vertex in $S$.
In the \emph{minimum vertex cover problem}, given a graph $G$, we want to compute a vertex cover of the minimum size.
It is well known that, for any maximal matching $M$, the vertex set consisting of all endpoints of edges in $M$ is a $2$-approximate vertex cover.
The following theorem is immediate from Theorem~\ref{thm:greedy-sensitivity-guarantee}.
\begin{theorem}\label{thm:vc-reduction-to-matching}
  There exists a $2$-approximation algorithm for the minimum vertex cover problem with average sensitivity at most $2$.
\end{theorem}

\subsection{Matching algorithm based on augmenting paths}\label{subsec:matching-iterative-greedy}

In this section, we describe a $(1-\eps)$-approximation algorithm for the maximum matching problem with average sensitivity $\tilde{O}\left(\mathsf{OPT}^{\frac{c}{c+1}}/{\eps^{\frac{3c}{c+1}}}\right)$ for $c = O(1/\eps^2)$ in Theorem~\ref{thm:matching-iterative-combined}.
The basic building block is a $(1-\eps)$-approximation algorithm (Algorithm~\ref{alg:matching-iterative-greedy}) for maximum matching that is based on iteratively augmenting a matching with greedily chosen augmenting paths of increasing lengths.
In Theorem~\ref{thm:matching-iterative-greedy}, we show that the average sensitivity of this algorithm is $\Delta^{O(1/\eps^2)}$, where $\Delta$ is the maximum degree of the input graph.
We obtain Theorem~\ref{thm:matching-iterative-greedy} by applying Theorem~\ref{thm:generic-transformation-from-local-oracle-to-stable-algorithm} to a result by Yoshida et al.~\cite{YYI12}.

We then apply Theorem~\ref{thm:generic-transformation-from-max-degree-to-avg-degree} to Theorem~\ref{thm:matching-iterative-greedy} in order to get rid of the dependence of the average sensitivity on the maximum degree and obtain Theorem~\ref{thm:matching-iterative-greedy-avg-degree-version}.
% \nvnote{Question for Yuichi: Should more explanation can be given (high level) about the thresholding technique since this is the first appearance? I shall write it if you think it will be a good addition.}
We then combine (using Theorem~\ref{thm:avg-sensitivity-distribution-of-multiple-algorithms}, the parallel composition theorem) the algorithm guaranteed by Theorem~\ref{thm:matching-iterative-greedy-avg-degree-version} with the algorithm guaranteed by Theorem~\ref{thm:matching-lexicographic} to obtain Theorem~\ref{thm:matching-iterative-combined}.
\subsubsection{Greedy matching algorithm based on augmenting paths}
In this section, we present an approximation algorithm that starts with an empty matching and then iteratively improves its size with augmenting paths of increasing lengths.
We show that the average sensitivity of this algorithm can be bounded using Theorem~\ref{thm:generic-transformation-from-local-oracle-to-stable-algorithm}.
\begin{algorithm}
  \caption{\textsc{Greedy Augmenting Paths Algorithm}}\label{alg:matching-iterative-greedy}
  \Input{undirected unweighted graph $G = (V,E)$, parameter $\eps \in (0,1)$}
  $M_0 \gets \emptyset$\;
  \For{$i \in \{1, 2, \dots \lceil \frac{1}{\eps} - 1\rceil\}$}{\label{stp:max-matching-for-loop-begin}
  Let $A_i$ denote the set of augmenting paths of length $2i - 1$ for the matching $M_{i-1}$\;
  Let $A'_i$ denote a maximal set of disjoint paths from $A_i$, where $A'_i$ is made from a random ordering of $A_i$\;
  $M_i \gets M_{i-1} \triangle A'_i$.
  }\label{stp:max-matching-for-loop-end}
  \Return $M_{\lceil \frac{1}{\eps} - 1\rceil}$.
\end{algorithm}

\begin{theorem}\label{thm:matching-iterative-greedy}
Algorithm~\ref{alg:matching-iterative-greedy} with parameter $\eps > 0$ has approximation ratio $1-\eps$ and average sensitivity $\Delta^{O(1/\eps^2)}$, where $\Delta$ is the maximum degree of the input graph.
\end{theorem}
\begin{proof}
For all $k \ge 0$, it is known that $|M_k| \ge \frac{k}{k+1} \cdot |M^*|$~\cite{GJ79}, where $M^*$ denotes a maximum matching in $G$. Hence, the matching $M_{\lceil \frac{1}{\eps} - 1\rceil}$ is a $(1-\eps)$-approximation to $M^*$.

Yoshida et al.~\cite[Theorem 3.7]{YYI12} show that for all $k \ge 0$, determining whether a uniformly random edge $e \sim E$ belongs to $M_k$ can be done by querying at most $\Delta^{O(k^2)}$ edges in expectation, where $\Delta$ is the maximum degree of $G$. Applying Theorem~\ref{thm:generic-transformation-from-local-oracle-to-stable-algorithm} to this result, we can see that the average sensitivity of Algorithm~\ref{alg:matching-iterative-greedy} with parameter $\eps > 0$ and input $G$ is $\Delta^{O(1/\eps^2)}$, where $\Delta$ is the maximum degree of $G$.
\end{proof}

\subsubsection{Stable-on-average thresholding transformation}
In this section, we show a transformation from matching algorithms whose average sensitivity is a function of the maximum degree to matching algorithms whose average sensitivity does not depend on the maximum degree.
This is done by adding to the algorithm, a preprocessing step that removes vertices from the input graph, where the removed vertices have degree at least an appropriate random threshold.
Such a transformation helps us to design stable-on-average algorithms for graphs with unbounded degree.
Let $\mathsf{Lap}(\mu, \phi)$ denote the Laplace distribution with a location parameter $\mu$ and a scale parameter $\phi$.

\begin{theorem}\label{thm:generic-transformation-from-max-degree-to-avg-degree}
	Let $\cA'$ be a randomized algorithm for the maximum matching problem such that the size of the matching output by $\cA'$ on a graph $G$ is always at least $a \cdot \mathsf{OPT}$ for some $a \geq 0$.
	In addition, assume that there exists a solution oracle $\mathcal{O}$ (see Definition~\ref{def:solution-oracle}) for $\cA'$ makes at most $q(\Delta)$ queries to $G$ in expectation, where $\Delta$ is the maximum degree of $G$, and the expectation is taken over the random coins of $\cA'$ and edges $e \in E$.
	Let $\delta > 0$ and $\tau$ be a non-negative function on graphs.
	Then, there exists an algorithm $\cA$ for the maximum matching problem with average sensitivity
	\[
	\beta(G) \le O\left(\frac{K_G}{\delta (\tau(G)-K_G)} + \exp\left(-\frac{1}{\delta}\right)\right)\cdot \mathsf{OPT} + \E_L\left[{(2L -2)}^2 q(L)\right],
	\]
	where $L$ is a random variable distributed as $\mathsf{Lap}(\tau(G), \delta \tau(G))$ and $K_G = \max_{e \in E(G)} |\tau(G)-\tau(G-e)|$.
	Moreover, the expected size of the matching output by $\mathcal{A}$ is at least
	$$a \cdot \mathsf{OPT} -  \frac{am}{(1 - \delta \ln (\mathsf{OPT}/2))\cdot\tau(G)} - a.$$
\end{theorem}
The following fact will be useful in the proof of Theorem~\ref{thm:generic-transformation-from-max-degree-to-avg-degree}.
\begin{proposition}\label{prop:laplace}
	Let $L$ be a random variable distributed as $\mathsf{Lap}(\mu, \phi)$.
	Then, $\Pr[L < (1-\eps)\mu] \leq \exp(-\eps \mu/\phi)/2$.
	Similarly, $\Pr[L > (1+\eps)\mu] \leq \exp(-\eps \mu/\phi)/2$.
\end{proposition}

\begin{proof}[Proof of Theorem~\ref{thm:generic-transformation-from-max-degree-to-avg-degree}]
	The algorithm $\cA$ is given below.

	\medskip
	\noindent {\bf Algorithm $\cA$}:
	On input $G = (V,E)$,
	\begin{enumerate}
		\item Sample a random variable $L$ according to the distribution $\mathsf{Lap}(\tau(G), \delta \tau(G))$.\label{stp:laplace-noise-addition}
		\item Let $[G]_L$ be the graph obtained after removing from $G$ all vertices of degree at least $L$.\label{stp:laplace-graph-modifcation}
		\item Run $\cA'$ on $[G]_L$.\label{stp:A-as-subroutine}
	\end{enumerate}

	\medskip

	\noindent We first bound the average sensitivity of $\cA$. We can think of $\cA$ as being sequentially composed of two algorithms, where the first algorithm takes in a graph $G = (V,E)$ and outputs a number $L \sim \mathsf{Lap}(\tau(G), \delta \tau(G))$.
	The second algorithm takes both $L$ and $G$ and runs $\cA'$ on $[G]_L$.

	\noindent Let $L_e$ for $e \in E$ denote a Laplace random variable distributed as $\mathsf{Lap}(\tau(G-e), \delta \tau(G-e))$.
	Using Theorem~\ref{thm:sequential-composition-TV-and-EM}, we get that the average sensitivity of $\cA$ is bounded by
	$${\sf OPT}\cdot \E_{e \sim E} \left[d_{\textrm{TV}}(L,L_e)\right] + \E_{L}\left[\E_{e \sim E}\left[d_{\textrm{EM}}(\cA'([G]_L), \cA'([G-e]_L))\right]\right].$$

	\begin{claim}
		For $x \in \R$, $\E_{e \sim E} \left[d_{\mathrm{EM}}\bigl(\cA'([G]_x), \cA'([G-e]_x)\bigr)\right] \le {(2x -2)}^2 q(x)$.
	\end{claim}
	\begin{proof}
		Fix $x \in \R$. In order to bound the term $\E_{e \sim E}\left[d_{\mathrm{EM}}\bigl(\cA'([G]_x), \cA'([G-e]_x)\bigr)\right]$, consider the following algorithm $\cA'_x$. On input $G = (V,E)$, the algorithm $\cA'_x$ first removes every vertex of degree at least $x$ from $G$ and then runs $\cA'$ on the resulting graph.
		Hence, the quantity $\E_{e \sim E}\left[d_{\mathrm{EM}}\bigl(\cA'([G]_x), \cA'([G-e]_x)\bigr)\right]$ denotes the average sensitivity of $\cA'_x$.

		In order to bound the average sensitivity of $\cA'_x$, construct a solution oracle $\mathcal{O}_x$ for $\cA'_x$ as follows. The oracle $\mathcal{O}_x$, when given access to a graph $G=(V,E)$ and input $e$ sampled uniformly at random from $E$, does the following. It first checks whether at least one of the endpoints of $e$ has degree at least $x$. If so, it returns that $e$ does not belong to the solution obtained by running $\cA'_x$ on $G$. Otherwise, it runs $\mathcal{O}$ with access to $[G]_x$ and $e$ as input and outputs the answer of $\mathcal{O}$.

		We can analyze the query complexity of $\mathcal{O}_x$ as follows. Call an edge $e \in E$ \emph{alive} if both the endpoints of $e$ have degree less than $x$. Otherwise, $e$ is {\em dead}.

		The oracle $\mathcal{O}_x$ can check whether an edge $e = (u,v)$ is alive or not by querying at most $2x - 2$ edges incident to $e$. In particular $\mathcal{O}_x$ examines the neighbors of $u$ and $v$ one by one, and, as soon $\mathcal{O}_x$ encounters $x - 1$ distinct neighbors (excluding $u$ or $v$ themselves) for either $u$ or $v$, $\mathcal{O}_x$ can declare $e$ to be a dead edge.

		If the edge $e \in E$ input to $\mathcal{O}_x$ is a dead edge, therefore, $\mathcal{O}_x$ queries at most $2x - 2$ edges and returns that $e$ cannot be part of a solution to running $\cA'_x$ on $G$.

		If the input edge $e \in E$ is alive, then we know that it is a uniformly random alive edge. By the guarantee on $\mathcal{O}$, we then know that $\mathcal{O}$ makes at most $q(x)$ queries to the alive edges in expectation over the randomness of $\cA'$ and the choice of the input alive edge, since the maximum degree of $[G]_x$ is at most $x$. In order for the oracle $\mathcal{O}_x$ to simulate oracle access to $[G]_x$ for the purpose of answering queries made by oracle $\mathcal{O}$, for each alive edge $e$ queried by $\mathcal{O}$, the oracle $\mathcal{O}_x$ has to query each edge incident to $e$ in $G$ and determine which among these are alive. Since $e$ is alive, both endpoints of $e$ have degrees less than $x$.
		Hence, $\mathcal{O}_x$ need only check whether at most $2x - 2$ edges incident to $e$ are alive or not. This can be done by querying ${(2x -2)}^2$ edges in $E$ in total.

		Combining all of the above, the expected query complexity of $\mathcal{O}_x$ is at most ${(2x -2)}^2 q(x)$, where the expectation is taken over the edges of $e \in E$ and the randomness in ${\cal A}_x$.

		Therefore, by Theorem~\ref{thm:generic-transformation-from-local-oracle-to-stable-algorithm}, we get that the average sensitivity of algorithm ${\cal A}_x$ is bounded by ${(2x -2)}^2 q(x)$.
	\end{proof}

	\noindent We now bound the quantity $\E_{e \sim E}  \left[d_{\text{TV}}(L, L_e)\right]$.
	\begin{claim}\label{clm:total-variation-between-neighboring-laplace}
		For any $e \in E$, we have
		\[
		d_{\mathrm{TV}}(L, L_e) \leq O\left(\frac{K}{\delta (\tau-K)} + \exp\left(-\frac{1}{\delta}\right) \right).
		\]
	\end{claim}
	\begin{proof}
		Let $f_L,f_{L_e}\colon \mathbb{R} \to \mathbb{R}$ be the probability density functions of the Laplace random variables $L$ and $L_e$, respectively.
		Let $\tau = \tau(G)$, $\tau_e = \tau(G_e)$, and $K=K_G$.
		Then
		\begin{align*}
		\frac{f_L(x)}{f_{L_e}(x)}
		& = \frac{\frac{1}{2\delta \tau} \exp\left(-\frac{|x-\tau|}{\delta \tau}\right)}{\frac{1}{2\delta \tau_e} \exp\left(-\frac{|x-\tau_e|}{\delta \tau_e}\right)}
		= \frac{\tau_e}{\tau}\exp\left(\frac{|x-\tau_e|}{\delta \tau_e}-\frac{|x-\tau|}{\delta \tau}\right) \\
		& = \left(1-\frac{\tau-\tau_e}{\tau}\right)\exp\left(\frac{\tau|x-\tau_e|-\tau_e|x-\tau|}{\delta \tau\tau_e}\right).
		%    & = \left(1 \pm \frac{K}{\tau}\right)\exp\left(\frac{\tau|x-\tau_e|-\tau_e|x-\tau|}{\delta \tau\tau_e}\right)
		\end{align*}
		A direct calculation shows that for $0 \leq x \leq 2\max\{\tau,\tau_e\}$, we have
		\begin{align*}
		\left(1 - \frac{K}{\tau}\right)\exp\left(\frac{-2K}{\delta (\tau-K)}\right)
		\leq \frac{f_L(x)}{f_{L_e}(x)}
		\leq \left(1 + \frac{K}{\tau}\right)\exp\left(\frac{2K}{\delta (\tau-K)}\right).
		\end{align*}
		This implies that for all $S \subseteq [0,2\max\{\tau,\tau_e\}]$,
		\[
		\left(1 - \frac{K}{\tau}\right)\exp\left(\frac{-2K}{\delta (\tau-K)}\right) - 1
		\leq \Pr[L\in S]-\Pr[L_e\in S]
		\leq \left(1 + \frac{K}{\tau}\right)\exp\left(\frac{2K}{\delta (\tau-K)}\right) - 1.
		\]
		By Proposition~\ref{prop:laplace}, the probability that $L$ (and $L_e$ as well) falls in the range $[-\infty,0] \cup [2\max\{\tau,\tau_e\},\infty]$ is bounded by $\exp(-1/\delta)$.
		Hence, total variation distance between $L$ and $L_e$ is
		\begin{align*}
		& \left(1 + \frac{K}{\tau}\right)\exp\left(\frac{2K}{\delta (\tau-K)}\right) - \left(1 - \frac{K}{\tau}\right)\exp\left(\frac{-2K}{\delta (\tau-K)}\right) + 2\exp\left(-\frac{1}{\delta}\right) \\
		& = \left(1 + \frac{K}{\tau}\right)\left(1+\frac{2K}{\delta (\tau-K)} + O\left(\frac{K^2}{\delta^2 {(\tau-K)}^2}\right) \right)  \\
		& \quad - \left(1 - \frac{K}{\tau}\right)\left(1-\frac{2K}{\delta (\tau-K)} - O\left(\frac{K^2}{\delta^2 {(\tau-K)}^2}\right) \right) + 2\exp\left(-\frac{1}{\delta}\right) \\
		& = \frac{2K}{\tau} + \frac{4K}{\delta (\tau-K)} + O\left(\frac{K^2}{\delta^2 {(\tau-K)}^2}\right)  + 2\exp\left(-\frac{1}{\delta}\right) \\
		& \leq \frac{6K}{\delta (\tau-K)} + 2\exp\left(-\frac{1}{\delta}\right) + O\left(\frac{K^2}{\delta^2 {(\tau-K)}^2}\right). \\
		& = O\left(\frac{K}{\delta (\tau-K)} + \exp\left(-\frac{1}{\delta}\right) \right). \qedhere
		\end{align*}
	\end{proof}

	\noindent Therefore, the average sensitivity of $\cA$ is bounded as
	\begin{align*}
	\beta(G) &= \E_{e \sim E} d_{\mathrm{EM}}\bigl(\cA(G), \cA(G-e)\bigr)\\
	&\le O\left(\frac{K}{\delta (\tau-K)} + \exp\left(-\frac{1}{\delta}\right) \right) \cdot \mathsf{OPT} + \E_L \left[  {(2x -2)}^2 q(x) \right].
	\end{align*}

	We now bound the approximation guarantee of $\cA$. By Proposition~\ref{prop:laplace},
	$$\Pr\left[L < \left(1 - \delta \ln \left(\frac{\mathsf{OPT}}{2}\right)\right)\cdot\tau(G)\right] \le \frac{1}{\mathsf{OPT}}.$$
	Therefore, with probability at least $1 - 1/\mathsf{OPT}$, only those vertices with degree at least $(1 - \delta \ln (\mathsf{OPT}/2))\cdot\tau(G)$ are removed from $G$.
	The number of such vertices is at most $\frac{m}{(1 - \delta \ln (\mathsf{OPT}/2))\cdot\tau(G)}$.
	Therefore, with probability at least $1 - 1/\mathsf{OPT}$, the size of a maximum matching in the resulting graph is at most $\frac{m}{(1 - \delta \ln (\mathsf{OPT}/2))\cdot\tau(G)}$ smaller than that of $G$.
	With probability at most $1/\mathsf{OPT}$, the size of a maximum matching in the resulting instance could be smaller by an additive term of at most $\mathsf{OPT}$.
	Hence, the expected size of a maximum matching in the new instance is at least $$\mathsf{OPT} - \frac{m}{(1 - \delta \ln (\mathsf{OPT}(G)/2))\cdot\tau(G)} - 1.$$ The statement on approximation guarantee follows.
\end{proof}
\subsubsection{Average sensitivity of the greedy augmenting paths algorithm with thresholding}

\begin{theorem}\label{thm:matching-iterative-greedy-avg-degree-version}
  Let $\eps \in (0,1)$ be a parameter.
  There exists an algorithm with approximation ratio $1-\eps$ and average sensitivity
  \[
    O\left(\frac{\eps}{1-\eps} \log n\right) + \left(\frac{m}{\eps^3 \mathsf{OPT}}\right)^{O(1/\eps^2)}.
  \]
\end{theorem}
\begin{proof}

The algorithm guaranteed by the theorem statement is as follows.

\medskip
\noindent {\bf Algorithm $\cA_\eps$}:
On input $G = (V,E)$,
\begin{enumerate}
	\item Compute $\mathsf{OPT}$.
	\item If $\mathsf{OPT} \le \frac{2}{\eps} + 1$ or $m \le \frac{1}{3\eps}$, then output an arbitrary maximum matching.
	\item Otherwise, run the algorithm obtained by applying Theorem~\ref{thm:generic-transformation-from-max-degree-to-avg-degree} with the setting $\tau:=\tau(G) = \frac{m}{\eps' \mathsf{OPT}}$ and $\delta := \frac{1}{2\ln n}$ to Algorithm~\ref{alg:matching-iterative-greedy} run with parameter $\eps'$, where $\eps' = \frac{\eps}{3} - \frac{1}{3\mathsf{OPT}}$.
\end{enumerate}

\medskip

  %The algorithm guaranteed by the theorem statement is the one that is obtained by applying Theorem~\ref{thm:generic-transformation-from-max-degree-to-avg-degree} to Algorithm~\ref{alg:matching-iterative-greedy} run with parameter $\eps'$ with the setting $\tau:=\tau(G) = m/\eps' \mathsf{OPT}(G)$ and $\delta := 1 / 2\ln n$, where $\eps' = \frac{\eps - \frac{1}{\mathsf{OPT}(G)}}{3}$.
  \smallskip

  \noindent \emph{Approximation guarantee}: If $\mathsf{OPT} \le \frac{1}{\eps} + 1$ or $m \le \frac{1}{2\eps}$, the approximation guarantee is clear.
  Otherwise, since Algorithm~\ref{alg:matching-iterative-greedy} outputs a maximal matching whose size is always at least $(1-\eps')\cdot\mathsf{OPT}$, the size of the matching output by $\cA_\eps$ is at least
  $(1-\eps')\cdot\mathsf{OPT} - \frac{\eps'\cdot (1-\eps')\cdot\mathsf{OPT}}{1 - \frac{\ln (\mathsf{OPT}/2)}{2\ln n}} - (1-\eps')$, which is at least $(1-\eps)\cdot\mathsf{OPT}$ by the setting of $\eps'$ and the fact that $\frac{\ln (\mathsf{OPT}/2)}{2\ln n} \le \frac{1}{2}$.
  %The analysis of the approximation ratio is straightforward.
\smallskip

\noindent \emph{Average sensitivity}: If $\mathsf{OPT} \le \frac{2}{\eps} + 1$ or $m \le \frac{1}{3\eps}$, the average sensitivity of $\cA_\eps$ is bounded by $O(\frac{1}{\eps})$, since the size of maximum matching in $G$ is small and it can decrease only by at most $1$ by the removal of an edge.

We now analyze the average sensitivity of $\cA_\eps$ for the case that $\mathsf{OPT} > \frac{2}{\eps} + 1$ and $m > \frac{1}{3\eps}$.
Let $c = O(1/\eps^2)$.
The average sensitivity of the algorithm resulting from applying Theorem~\ref{thm:generic-transformation-from-max-degree-to-avg-degree} to Algorithm~\ref{alg:matching-iterative-greedy} is bounded as:
\begin{align}
O\left(\frac{K_G}{\delta (\tau-K_G)} + \exp\left(-\frac{1}{\delta}\right)\right)\cdot \mathsf{OPT}
+ \int_0^\infty {(2x-2)}^2 \cdot x^{c}\cdot  \frac{1}{2\delta \tau} \cdot \exp\left(-\frac{|x - \tau|}{\delta \tau}\right) \, \mathrm{d}x. \label{eq:matching-rand-greedy-sensitivity}
\end{align}
To obtain the above expression, we used the fact (from~\cite[Theorem 3.7]{YYI12}) that $q(x) \le x^c$ when $x > 0$ and $q(x) = 0$ otherwise.

The second term of~\eqref{eq:matching-rand-greedy-sensitivity} can be bounded as:
\begin{align*}
& \int_0^\infty (2x-2)^2 x^{c} \frac{1}{2\delta \tau}\exp\left(-\frac{|x-\tau|}{\delta \tau}\right) \mathrm{d}x
= 4\int_\tau^\infty x^{c+2} \frac{1}{\delta \tau}\exp\left(-\frac{x-\tau}{\delta \tau}\right) \mathrm{d}x \\
& = \exp\left(\frac{1}{\delta}\right) (\delta \tau)^{c+2} \Gamma\left(c+3,\frac{1}{\delta}\right)
= (\delta \tau)^{c+2} (c+2)! \sum_{k=0}^{c+2}\frac{(1/\delta)^k}{k!}
= \left(\frac{m}{\eps^3 \mathsf{OPT}}\right)^{O(1/\eps^2)}
\end{align*}
where $\Gamma(\cdot,\cdot)$ is the incomplete Gamma function and we have used the fact that $\Gamma(s+1,x) = s! \exp(-x) \sum_{k=0}^s x^k/k!$ if $s$ is a non-negative integer.
Moreover, each term in the summation $\delta^{c+2}\cdot(c+2)! \sum_{k=0}^{c+2}\frac{(1/\delta)^k}{k!}$ is $o(1)$. Hence, the summation is $O(\frac{1}{\eps^2})$.

In order to bound the first term of~\eqref{eq:matching-rand-greedy-sensitivity}, note that
\begin{align*}
K_G &=\max_{e \in E} |\tau(G) - \tau(G-e)|\\
&= 3\max_{e \in E}\left|\frac{m}{\eps\mathsf{OPT}(G) - 1}-\frac{m-1}{\eps\mathsf{OPT}(G-e) - 1}\right|\\ &\leq 3\max\left\{\frac{m}{\eps\mathsf{OPT}(G)-1}-\frac{m-1}{\eps\mathsf{OPT}(G)-1}, \frac{m-1}{\eps(\mathsf{OPT}(G)-1)-1}-\frac{m}{\eps\mathsf{OPT}(G)-1} \right\}  \\
&= 3\max\left\{\frac{1}{\eps\mathsf{OPT}(G)-1}, \frac{\eps(m-\mathsf{OPT}(G))+1}{(\eps(\mathsf{OPT}(G)-1)-1)\cdot(\eps\mathsf{OPT}(G)-1)} \right\}\\
&= \frac{3}{\eps\mathsf{OPT}(G)-1}\max\left\{1, \frac{\eps(m-\mathsf{OPT}(G))+1}{\eps(\mathsf{OPT}(G)-1)-1} \right\}.
\end{align*}
The inequality above uses the fact that for numbers $u,v,w \ge 0$ such that $u \le v$ and $w(v-1) - 1 \ge 0$, we have that $\frac{u}{wv - 1} \le \frac{u - 1}{w(v-1)-1}$.

\noindent Since $\frac{\eps(m-\mathsf{OPT})+1}{\eps(\mathsf{OPT}-1)-1}$ is a nonincreasing function of $\mathsf{OPT}$ and $\mathsf{OPT} > \frac{2}{\eps} + 1$, we have that $$\frac{\eps(m-\mathsf{OPT})+1}{\eps(\mathsf{OPT}-1)-1} < \eps m.$$
Hence, $K_G < \frac{3}{\eps\mathsf{OPT}-1}\max\left\{1, \eps m\right\} = \frac{9\eps m}{\eps\mathsf{OPT}-1}$, since $m > \frac{1}{3\eps}$ and therefore, we have that $\tau - K_G \ge \tau(1-9\eps)$.
Hence, the first term of~(\ref{eq:matching-rand-greedy-sensitivity}) can be upper bounded by

\begin{align*}
O\left(\frac{K_G}{\delta \tau (1-9\eps)} + \exp\left(-\frac{1}{\delta}\right)\right)\cdot \mathsf{OPT}
& = O\left(\frac{9\eps m}{\eps \mathsf{OPT} - 1}\cdot \frac{1}{1 - 9\eps}\cdot \frac{2\ln n}{\frac{m}{\eps \mathsf{OPT} - 1}} + \frac{\mathsf{OPT}}{n^2}\right) \\
& = O\left(\frac{\eps}{1- \eps}\cdot \log n \right).
\end{align*}

Hence, the average sensitivity of the algorithm obtained can be bounded by:
\begin{align*}
\beta(G) &= \max\left\{O\left(\frac{1}{\eps}\right), O\left(\frac{\eps}{1- \eps}\cdot \log n\right) + \left(\frac{m}{\eps^3 \mathsf{OPT}}\right)^{O(1/\eps^2)}\right\}\\
&= O\left(\frac{\eps}{1- \eps}\log n\right)  + \left(\frac{m}{\eps^3 \mathsf{OPT}}\right)^{O(1/\eps^2)}. \qedhere
\end{align*}
\end{proof}

\subsubsection{Average sensitivity of a combined matching algorithm}
In this section, we combine the algorithms guaranteed by Theorems~\ref{thm:matching-lexicographic} and~\ref{thm:matching-iterative-greedy-avg-degree-version} in order to get a matching algorithm with improved sensitivity.
\begin{theorem}\label{thm:matching-iterative-combined}
  Let $\eps \in (0,1)$ be a parameter.
  There exists an algorithm with approximation ratio $1-\eps$ and average sensitivity
  \[
    \mathrm{OPT}(G)^{\frac{c}{c+1}}\cdot O\left(\left(\frac{\eps}{1-\eps} \cdot \log n\right)^{\frac{1}{c+1}} + \frac{1}{\eps^{\frac{3c}{c+1}}}  \right)
  \]
  for $c = O(1/\eps^2)$.
%  In particular when $\eps = 1/\log^{\frac{1}{3c+1}} n$, the average sensitivity is $\mathsf{OPT}(G)^{c/(c+1)}O\left(\log^{\frac{3c}{(3c+1)(c+1)}} n)   \right)$.
\end{theorem}
\begin{proof}
	Let $c = O(1/\eps^2)$.
	The algorithm guaranteed by the theorem is given as Algorithm~\ref{alg:combined-maximum-matching-tradeoff}.
	The bounds on approximation guarantee and average sensitivity are both straightforward when $\mathsf{OPT} < 2c$ or $m < 2c$.

\begin{algorithm}
	\caption{\textsc{Combined Algorithm to $\left(1 - \eps\right)$-Approximate Maximum Matching}}\label{alg:combined-maximum-matching-tradeoff}
	\Input{undirected unweighted graph $G = (V,E)$}
	Compute $\mathsf{OPT}$.\;
	\If{$\mathsf{OPT} < 2c$ or $m < 2c$}{
		\Return an arbitrary maximum matching in $G$.
	}\Else{
		Let $f(G) \gets \frac{\mathsf{OPT}^2}{m}$ and $g(G) \gets \frac{\eps}{(1- \eps)}\cdot \log n + \left(\frac{m}{\eps^3 \mathsf{OPT}}\right)^{c}$\;
		Run the algorithm given by Theorem~\ref{thm:matching-lexicographic} with probability $\frac{g(G)}{f(G) + g(G)}$ and run the algorithm given by Theorem~\ref{thm:matching-iterative-greedy-avg-degree-version} with the remaining probability.
	}
\end{algorithm}
The approximation guarantee in the case when $\mathsf{OPT} \ge 2c$ and $m \ge 2c$ is also straightforward since Algorithm~\ref{alg:combined-maximum-matching-tradeoff} is simply a distribution over algorithms guaranteed by Theorem~\ref{thm:matching-lexicographic} and Theorem~\ref{thm:matching-iterative-greedy-avg-degree-version}.

We now bound the average sensitivity of Algorithm~\ref{alg:combined-maximum-matching-tradeoff} when  $\mathsf{OPT} \ge 2c$ and $m \ge 2c$.
Let $\rho(G)$ denote the probability $\frac{g(G)}{f(G) + g(G)}$.
By Theorem~\ref{thm:avg-sensitivity-distribution-of-algorithms}, the average sensitivity is at most
\begin{equation} \label{eqn:sensitivity-combined-iterative-greedy}
\frac{O(f(G))\cdot g(G) + O(g(G))\cdot f(G)}{f(G) + g(G)} + 2\mathsf{OPT} \cdot \E_{e \sim E} \left[|\rho(G) - \rho(G-e)|\right].
\end{equation}

We first bound the quantity $\E_{e \sim E} \left[|\rho(G) - \rho(G-e)|\right]$.
  \begin{claim}\label{clm:lipschitz-constant-g-augmenting-greedy}
	For every graph $G = (V,E)$ such that $\mathsf{OPT} \ge c+1$, and for every $e \in E$,
	$$\left(1 - \frac{c}{m}\right)\cdot g(G) \le g(G-e) \le \left(1 + \frac{c}{\mathsf{OPT} - c}\right)\cdot g(G).$$
\end{claim}
\begin{proof}
	We first prove the upper bound. We know that
	\begin{align*}
	\frac{g(G-e)}{g(G)} &\le \left(1 + \frac{\left(\frac{m-1}{\eps^3 (\mathsf{OPT}-1)}\right)^c - \left(\frac{m}{\eps^3\mathsf{OPT}}\right)^c}{\frac{\eps\log n}{(1- \eps)} + \frac{m^c}{\eps^{3c}\mathsf{OPT}^c}}\right)
	\le  \left(1 + \frac{\left(\frac{m-1}{\eps^3 (\mathsf{OPT}-1)}\right)^c - \left(\frac{m}{\eps^3\mathsf{OPT}}\right)^c}{\frac{m^c}{\eps^{3c}\mathsf{OPT}^c}}\right)\\
	&= \left(1 - \frac{1}{m}\right)^c \cdot \left(1 + \frac{1}{\mathsf{OPT} - 1}\right)^c
	\le \left(1 + \frac{c}{\mathsf{OPT} - c}\right).
	\end{align*}
	Note that the last inequality holds whenever $\textsf{OPT} > c$, because $(1 + x)^r \le 1 + \frac{rx}{1 - (r-1)x}$ for $x \in [0,\frac{1}{r-1})$ and $r > 1$.

	\noindent For the lower bound,
	\begin{align*}
	\frac{g(G-e)}{g(G)} &\ge \left(1 - \frac{\left(\frac{m}{\eps^3\mathsf{OPT}}\right)^c - \left(\frac{m-1}{\eps^3 \mathsf{OPT}}\right)^c}{\frac{\eps\log n}{(1- \eps)} + \frac{m^c}{\eps^{3c}\mathsf{OPT}^c}}\right)
	\ge \left(1 - \frac{\left(\frac{m}{\eps^3\mathsf{OPT}}\right)^c - \left(\frac{m-1}{\eps^3 \mathsf{OPT}}\right)^c}{\frac{m^c}{\eps^{3c}\mathsf{OPT}^c}}\right)\\
	&=  \left(1 - \frac{1}{m}\right)^c \ge 1 - \frac{c}{m}.\qedhere
	\end{align*}
\end{proof}
\begin{claim}\label{clm:lipschitz-constant-f-augmenting-greedy}
	For every graph $G = (V,E)$ and every $e \in E$, $$f(G) \cdot \left(1 - \frac{2}{\mathsf{OPT}}\right)\le f(G -e) \le f(G) \cdot \left(1 + \frac{1}{m-1}\right).$$
\end{claim}
\begin{proof}
	To prove the upper bound,
	\begin{align*}
	\frac{f(G-e)}{f(G)} \le \left(\frac{m}{m-1}\right) = \left(1 + \frac{1}{m-1}\right).
	\end{align*}

	\noindent For the lower bound,
	\begin{align*}
	\frac{f(G-e)}{f(G)} &\ge \left(\frac{\mathsf{OPT} - 1}{\mathsf{OPT}}\right)^2 \cdot \left(\frac{m}{m-1}\right)^2
	\ge \left(\frac{\mathsf{OPT} - 1}{\mathsf{OPT}}\right)^2 \ge 1 - \frac{2}{\mathsf{OPT}}.\qedhere
	\end{align*}
\end{proof}
\begin{claim}\label{clm:lipschitz-constant-rho-augmenting-greedy}
For every graph $G = (V,E)$ such that $\mathsf{OPT} \ge 2c$ and $m \ge 2c$, and for every $e \in E$,
$$\rho(G)\cdot \left(1 - \frac{2c}{\mathsf{OPT}-c}\right) \le \rho(G-e) \le \rho(G) \cdot \left(1 + \frac{5c}{\mathsf{OPT} - c}\right).$$
\end{claim}
\begin{proof}
\noindent Note that $\left(1 - \frac{2}{\mathsf{OPT}}\right)^{-1} \le 1 + \frac{4}{\mathsf{OPT}}$ and $\left(1 - \frac{c}{m}\right)^{-1} \le 1 + \frac{2c}{m}$ for $\mathsf{OPT} \ge 4$ and $m \ge 2c$.
We also have  $\left(1+\frac{c}{\mathsf{OPT}-c}\right)^{-1} \ge 1 - \frac{c}{\mathsf{OPT} - c}$ and $\left(1 + \frac{1}{m-1}\right)^{-1} \ge 1 - \frac{1}{m-1}$ for $\mathsf{OPT} \ge 2c$ and $m \ge 2$.

%Therefore,
%$$\frac{g^{1/4}(G)}{f^{1/4}(G)}\cdot \left(1 - \frac{15}{4m}\right) \le \frac{g^{1/4}(G-e)}{f^{1/4}(G-e)} \le \frac{g^{1/4}(G)}{f^{1/4}(G)} \cdot \left(1 + \frac{9}{\mathsf{OPT}-1}\right),$$

%and

%$$\frac{f^{3/4}(G)}{g^{3/4}(G)}\cdot \left(1 - \frac{5}{\mathsf{OPT}}\right) \le \frac{f^{3/4}(G-e)}{g^{3/4}(G-e)} \le \frac{f^{3/4}(G)}{g^{3/4}(G)} \cdot \left(1 + \frac{87}{4(m-1)}\right).$$
\smallskip

\noindent Combining all of the above,
\begin{align*}
\rho(G-e) &= \frac{g(G-e)}{f(G-e) + g(G-e)}\\
&\le \frac{g(G)\cdot \left(1 + \frac{c}{\mathsf{OPT} - c}\right)}{(f(G) + g(G))\cdot \min \left\{1 - \frac{c}{m}, 1 - \frac{2}{\mathsf{OPT}}\right\}}\\
&\le \rho(G) \cdot \left(1 + \frac{c}{\mathsf{OPT} - c}\right) \cdot \max \left\{1 + \frac{2c}{m}, 1 + \frac{4}{\mathsf{OPT}}\right\}\\
&\le \rho(G) \cdot \left(1 + \frac{5c}{\mathsf{OPT} - c}\right).
\end{align*}

%The second to last inequality holds whenever $\mathsf{OPT} \ge 4$ and $m \ge 6$.

\noindent Using similar calculations, we can see that
\begin{align*}
\rho(G-e) &\ge \rho(G) \cdot \left(1 - \frac{c}{m}\right)\cdot \min \left\{1 - \frac{c}{\mathsf{OPT} - c}, 1 - \frac{1}{m-1}\right\}
\ge \rho(G)\cdot \left(1 - \frac{2c}{\mathsf{OPT}-c}\right).\qedhere
\end{align*}
\end{proof}
\noindent Thus, for all $e \in E$, we have that $|\rho(G) - \rho(G-e)| \le \max\left\{\frac{2c}{\mathsf{OPT}-c}, \frac{5c}{\mathsf{OPT} - c}\right\}\cdot \rho(G) = \frac{5c\rho(G)}{\mathsf{OPT} - c}$. Hence, $\E_{e \sim E} [|\rho(G) - \rho(G-e)|] \le \frac{5c\rho(G)}{\mathsf{OPT} - c}$.

\medskip

\noindent Therefore, the average sensitivity of  Algorithm~\ref{alg:combined-maximum-matching-tradeoff} is at most
\begin{align*}
& \frac{O(f(G))\cdot g(G) + O(g(G))\cdot f(G)}{f(G) + g(G)} + 2\mathsf{OPT} \cdot \E_{e \sim E} [|\rho(G) - \rho(G-e)|]\\
&= O\left(\frac{f(G)^{c/(c+1)}g(G)^{1/(c+1)}}{\frac{g(G)^{1/(c+1)}}{f(G)^{1/(c+1)}}+\frac{f(G)^{c/(c+1)}}{g(G)^{c/(c+1)}}} \right) + O\left(\frac{\mathsf{OPT}c\rho(G)}{\mathsf{OPT}}\right)\\
&=
O\left(f(G)^{c/(c+1)}g(G)^{1/(c+1)}\right) + O(1/\eps^2)\\
& =
O\left( \left(\frac{\mathsf{OPT}^2}{m}\right)^{c/(c+1)} \cdot \left((\frac{\eps}{1-\eps} \log n)^{1/(c+1)} + \left(\frac{m^c}{\eps^{3c}\mathsf{OPT}^c}\right)^{1/(c+1)}\right)\right) + O(1/\eps^2)\\
& =
O\left( \left(\frac{\mathsf{OPT}^{2c/(c+1)}}{m^{c/(c+1)}} \left(\frac{\eps}{1-\eps}\right)^{1/(c+1)} \log^{1/(c+1)} n + \frac{\mathsf{OPT}^{2c/(c+1)}}{m^{c/(c+1)}}\frac{m^{c/(c+1)}}{\eps^{3c/(c+1)}\mathsf{OPT}^{c/(c+1)}}\right)\right) + O(1/\eps^2)\\
& = O\left(\mathsf{OPT}^{c/(c+1)} \left(\frac{\eps}{1-\eps}\right)^{1/(c+1)} \log^{1/(c+1)} n + \frac{\mathsf{OPT}^{c/(c+1)}}{\eps^{3c/(c+1)}}\right) + O(1/\eps^2)\\
& = O\left(\mathsf{OPT}^{c/(c+1)} \left( \left(\frac{\eps}{1-\eps}\right)^{1/(c+1)} \log^{1/(c+1)} n + \frac{1}{\eps^{3c/(c+1)}}\right)\right). \qedhere
\end{align*}

To obtain the first term of the expression resulting from the first equality, we divide both the numerator and denominator by $f(G)^{\frac{1}{c+1}}\cdot g(G)^{\frac{c}{c+1}}$. % upper bound the resulting fraction by its numerator $g(G)^{\frac{1}{c+1}}\cdot f(G)^{\frac{c}{c+1}}$.
The second term of the first equality above follows since $\frac{\mathsf{OPT}}{\mathsf{OPT} - c} \le 2$ as $\mathsf{OPT} \ge 2c$.
   %$\frac{\frac{g(G)^{1/(c+1)}}{f(G)^{1/(c+1)}}}{\frac{g(G)^{1/(c+1)}}{f(G)^{1/(c+1)}}+\frac{f(G)^{c/(c+1)}}{g(G)^{c/(c+1)}}}$
\end{proof}

\subsection{Lower bound}\label{subsec:matching-exact-lower-bound}

In this section, we show a lower bound of $\Omega(n)$ for the problem of exactly computing the maximum matching in a graph.

\begin{theorem}\label{thm:matching-exact-lower-bound}
Every algorithm that exactly computes the maximum matching in a graph has average sensitivity $\Omega(n)$.
\end{theorem}
\begin{proof}
	Let $n \in \mathbb{N}$ be even.
	Consider the cycle $C_n$ on $n$ vertices.
	$C_n$ has exactly two maximum matchings $M_1$ and $M_2$ of size $n/2$ each.
	Both $M_1$ and $M_2$ consist of alternating edges of the cycle.
	Let $A$ be an algorithm that outputs $M_1$ with probability $p$ and $M_2$ with probability $1-p$.
	Assume, without loss of generality, that $p \ge \frac{1}{2}$.
	For every edge $e \in M_1$, the unique maximum matching in the odd-length path $G-e$ has Hamming distance $n-1$ from $M_1$.
	Thus, for each $e \in M_1$, the earth mover's distance between $A(G)$ and $A(G-e)$ is at least $\frac{n-1}{2}$.
	Hence, the average sensitivity of $A$ is at least $\frac{1}{n}\sum_{e \in M_1} \frac{n-1}{2} = \Omega(n)$.
\end{proof}

\section{2-Coloring}\label{sec:2-coloring}

In the \emph{$2$-coloring problem}, given a bipartite graph $G=(V,E)$, we are to output a (proper) $2$-coloring on $G$, that is, an assignment $f:V\to \{0,1\}$ such that $f(u) \neq f(v)$ for every edge $(u,v) \in E$.
Clearly this problem can be solved in linear time.
In this section, however, we show that there is no stable-on-average algorithm for the $2$-coloring problem.

\begin{theorem}\label{thm:2-coloring}
  Any (randomized) algorithm for the $2$-coloring problem has average sensitivity $\Omega(n)$.
\end{theorem}
\begin{proof}
  Suppose that there is a (randomized) algorithm $\mathcal{A}$ whose average sensitivity is at most $\beta n$ for $\beta < 1/256$.
  In what follows, we assume that $n$, that is, the number of vertices in the input graph, is a multiple of $16$.

  Let $\mathcal{P}_n$ be the family of all possible paths on $n$ vertices, and let $\mathcal{Q}_n$ be the family of all possible graphs on $n$ vertices consisting of two paths.
  Note that $|\mathcal{P}_n| = n!/2$ and $|\mathcal{Q}_n| = (n-1)n!/4$.
  Consider a bipartite graph $H = (\mathcal{P}_n,\mathcal{Q}_n; E)$, where a pair $(P,Q)$ is in $E$ if and only if $Q$ can be obtained by removing an edge in $P$.
  Note that each $P \in \mathcal{P}_n$ has $n-1$ neighbors in $H$ and each $Q \in \mathcal{Q}_n$ has four neighbors in $H$.

  We say that an edge $(P,Q) \in E$ is \emph{intimate} if $d_{\mathrm{EM}}\bigl(\mathcal{A}(P),\mathcal{A}(Q)\bigr) \leq 8\beta n$.
  We observe that for every $P \in \mathcal{P}_n$, at least a $7/8$-fraction of the edges incident to $P$ are intimate;
  otherwise
  \[
    \E_{e \sim E(P)}\left[d_{\mathrm{EM}}\bigl(\mathcal{A}(P),\mathcal{A}(P-e)\bigr)\right] > \frac{1}{8} \cdot 8\beta n = \beta n,
  \]
  which is a contradiction, where $E(P)$ denotes the set of edges in $P$.

  We say that a graph $Q \in \mathcal{Q}_n$ is \emph{heavy} if both components of $Q$ have at least $n/16$ vertices, and say that an edge $(P,Q) \in E$ is \emph{heavy} if $Q$ is heavy.
  We observe that for every $P \in \mathcal{P}_n$, at least a $7/8$-fraction of the edges incident to $P$ are heavy.

  We say that an edge $(P,Q) \in E$ is \emph{good} if it is intimate and heavy.
  Observe that for every $P \in \mathcal{P}_n$, by the union bound, at least a $3/4$-fraction of the edges incident to $P$ are good.
  In particular, this means that the fraction of good edges in $H$ is at least $3/4$.
  Hence, there exists $Q^* \in \mathcal{Q}_n$ that has at least three good incident edges;
  otherwise the fraction of good edges in $H$ is at most $2/4 = 1/2$, which is a contradiction.

  Let $f_1,\ldots,f_4$ be the four $2$-colorings of $Q^*$.
  As $Q^*$ has three good incident edges, without loss of generality, there are adjacent paths $P_1, P_2 \in \mathcal{P}_n$ such that both $(P_1,Q^*)$ and $(P_2,Q^*)$ are good, and there is no assignment that is a $2$-coloring for both $P_1$ and $P_2$.
  Without loss of generality, we assume that $f_1,f_2$ are $2$-colorings of $P_1$, and $f_3,f_4$ are $2$-colorings of $P_2$.
  Note that $d_{\mathrm{Ham}}(f_i,f_j) \geq n/16$ for $i \neq j$ because $Q$ is heavy.
  Let $q_i = \Pr[\mathcal{A}(Q^*) = f_i]$ for $i \in [4]$.
  As the edge $(P_1,Q^*)$ is intimate, we have
  \begin{align*}
    8\beta n
    & \geq d_{\mathrm{EM}}\bigl(\mathcal{A}(P_1),\mathcal{A}(Q^*)\bigr)
    \geq \frac{n}{16}\left( \bigl|\Pr[\mathcal{A}(P_1) = f_1] - q_1\bigr| + \bigl|\Pr[\mathcal{A}(P_1) = f_2] - q_2\bigr|  + q_3 + q_4\right)  \\
    & = \frac{n}{16}\left( \bigl|\Pr[\mathcal{A}(P_1) = f_1] - q_1\bigr| + \bigl|\Pr[\mathcal{A}(P_1) = f_2] - q_2\bigr|  + 1 - q_1 - q_2 \right)
  \end{align*}
  and hence we must have $q_1 + q_2 \geq 1 - 128 \beta$.
  Considering $d_{\mathrm{EM}}\bigl(\mathcal{A}(P_2),\mathcal{A}(Q^*)\bigr)$, we also have $q_3 + q_4 \geq 1-128 \beta$.
  However,
  \[
    1 = q_1 + q_2 + q_3 + q_4 \geq (1 - 128\beta) + (1 - 128 \beta) = 2 - 256\beta > 1
  \]
  as $\beta < 1/256$, which is a contradiction.
\end{proof}

%!TEX root=./stabilityOfAlgorithms.tex

\section{General Results on Average Sensitivity}\label{sec:general}
In this section, we state and prove some basic properties of average sensitivity and show that locality guarantees of solutions output by an algorithm imply low average sensitivity for that algorithm.
% \nvnote{Pink text is a little inaccurate. But I guess it is okay.}
\subsection{Bounds on \texorpdfstring{$k$}{k}-average sensitivity from bounds on average sensitivity}

In this section, we prove Theorem~\ref{thm:sensitivity-against-deleting-multiple-edges}, which says that, if an algorithm is stable-on-average against deleting a single edge, it is also stable-on-average against deleting multiple edges.
We restate the theorem here.
\sadme*
% \begin{theorem}\label{thm:sensitivity-against-deleting-multiple-edges}
% 	Let $\mathcal{A}$ be an algorithm for a graph problem with average sensitivity given by $f(n,m)$.
% 	Then, for any integer $k \geq 1$, the algorithm $\mathcal{A}$ has $k$-average sensitivity $\sum_{i=1}^k f(n,m-i+1)$.
% \end{theorem}
\begin{proof}
	We have
	\begin{align*}
	& \E_{\{e_1,\ldots,e_k\} \sim \binom{E}{k}}\left[d_{\mathrm{EM}}\bigl(\mathcal{A}(G),\mathcal{A}(G-\{e_1,\ldots,e_k\})\bigr)\right] \\
	\leq & \E_{\{e_1,\ldots,e_k\} \sim \binom{E}{k}}\left[\sum_{i=1}^{k}d_{\mathrm{EM}}\bigl(\mathcal{A}(G-\{e_1,\ldots,e_{i-1}\}),\mathcal{A}(G-\{e_1,\ldots,e_i \}) \bigr) \right]\\
	= & \E_{e_1 \sim E} \Bigl[ d_{\mathrm{EM}}\bigl(\mathcal{A}(G),\mathcal{A}(G-\{e_1\})\bigr) + \E_{e_2\sim E \setminus \{e_1\}}\Bigl[d_{\mathrm{EM}}\bigl(\mathcal{A}(G-\{e_1\}),\mathcal{A}(G-\{e_1,e_2\})\bigr) + \cdots \\
	& \qquad + \E_{e_k \sim E\setminus \{e_1,\dots,e_{k-1}\}}\Bigl[d_{\mathrm{EM}}\bigl(\mathcal{A}(G-\{e_1,\ldots,e_{k-1}\}),\mathcal{A}(G-\{e_1,\ldots,e_k \}) \Bigr)\dots\Bigr]\Bigr]\Bigr] \\
	\le & f(n,m) + \E_{e_1 \sim E}\Bigl[\beta(G-\{e_1\}) + \E_{e_2 \sim E\setminus \{e_1\}}\Bigl[\beta(G-\{e_1,e_2\}) + \cdots \\&+ \E_{e_{k-1} \sim E\setminus\{e_1,\dots e_{k-2}\}}\Bigl[\beta(G-\{e_1,\ldots,e_{k-1}\})\dots\Bigr]\Bigr]\Bigr]\\
	\leq & \sum_{i=1}^kf(n,m-i+1).
	\end{align*}
	Here, the first inequality is due to the triangle inequality.
\end{proof}
\subsection{Sequential composition}

In this section, we state and prove our two sequential composition theorems Theorem~\ref{thm:sequential-composition-TV-and-EM} and Theorem~\ref{thm:composition}.
%It will be useful if we can sequentially apply averagely stable algorithms on the input to get a solution and the whole algorithm is again averagely stable.
%We show two different sequential composition theorems for average sensitivity.

\scTaE*
\begin{proof}
	Consider $G = (V,E)$ and let $e \in E$. We bound the earth mover's distance between $\cA(G)$ and $\cA(G-e)$ as follows.
	For a distribution $\cD$, we use $f_\cD$ to denote its probability mass function.
	We know that for all $S_1 \in \mathcal{S}_1$ and $S_2 \in \mathcal{S}_2$
	$$f_{(\cA_1(G), \cA_2(G,S_1))}(S_1, S_2) = f_{\cA_1(G)}(S_1) \cdot f_{\cA_2(G,S_1)}(S_2),$$
	where $(\cA_1(G), \cA_2(G,S_1))$ denotes the joint distribution of $\cA_1(G)$ and $\cA_2(G,S_1)$.
	Fix $S_1 \in \mathcal{S}_1$.
	For each $S_2 \in \mathcal{S}_2$, we transform probabilities of the form $f_{(\cA_1(G), \cA_2(G,S_1))}(S_1, S_2)$ to $f_{\cA_1(G)}(S_1)\cdot f_{\cA_2(G-e,S_1)}(S_2)$. This incurs a total cost of $f_{\cA_1(G)}(S_1) \cdot d_{\textrm{EM}}(\cA_2(G, S_1), \cA_2(G-e, S_1))$.
	We can now, for each $S_1 \in \mathcal{S}_1$ and $S_2 \in \mathcal{S}_2$, transform the probability $f_{\cA_1(G)}(S_1)\cdot f_{\cA_2(G-e,S_1)}(S_2)$ into $f_{\cA_1(G-e)}(S_1)\cdot f_{\cA_2(G-e,S_1)}(S_2)$ at a cost of at most $d_{\textrm{TV}}(\cA_1(G), \cA_1(G-e)) \cdot \mathsf{H}$, where $\mathsf{H}$ denotes the maximum Hamming weight among those of solutions obtained by running $\cA$ on $G$ and $\{G-e\}_{e \in E}$.
	Thus, the earth mover's distance between $\cA(G)$ and $\cA(G-e)$ is at most $$d_{\textrm{TV}}(\cA_1(G), \cA_1(G-e)) \cdot \mathsf{H} + \int_{\mathcal{S}_1} f_{\cA_1(G)}(S_1) \cdot d_{\textrm{EM}}\bigl(\cA_2(G, S_1), \cA_2(G-e, S_1)\bigr) \; \mathrm{d}S_1.$$
	Hence, the average sensitivity of $\cA$ can be bounded as:
	\begin{align*}
	\E_{e \sim E} \left[d_{\textrm{EM}}(\cA(G), \cA(G-e))\right] &\le \mathsf{H}\cdot \E_{e \sim E} \left[d_{\textrm{TV}}(\cA_1(G), \cA_1(G-e))\right] \\&+ \E_{e \sim E} \left[\int_{S_1 \in \mathcal{S}_1} f_{\cA_1(G)}(S_1) \cdot d_{\textrm{EM}}(\cA_2(G, S_1), \cA_2(G-e, S_1))\; \mathrm{d}S_1\right]\\
	&\le \mathsf{H} \gamma_1(G) + \E_{S_1 \sim \cA_1(G)}\left[d_{\textrm{EM}}(\cA_2(G, S_1), \cA_2(G-e, S_1))\right] \\
	&= \mathsf{H} \gamma_1(G) + \E_{S_1 \sim \cA_1(G)}\left[\E_{e \sim E}d_{\textrm{EM}}(\cA_2(G, S_1), \cA_2(G-e, S_1))\right]\\
	&= \mathsf{H} \gamma_1(G) + \E_{S_1 \sim \cA_1(G)}\left[\beta_2^{(S_1)}(G)\right].
	\end{align*}
	We are able to interchange the order of expectations because of Fubini's theorem~\cite{F1907}.
\end{proof}

\noindent The following theorem states the composition of average sensitivity with respect to the total variation distance.
\c*
%We start by proving the composition of average sensitivity for two algorithms.

\noindent Theorem~\ref{thm:composition} can be immediately obtained by iteratively applying Lemma~\ref{lem:composition}.

\begin{lemma}\label{lem:composition}
	Consider two randomized algorithms $\mathcal{A}_1: \mathcal{G} \to \mathcal{S}_1,\mathcal{A}_2: \mathcal{G} \times \mathcal{S}_1 \to \mathcal{S}_2$ for a graph problem.
	Suppose that the average sensitivity of $\mathcal{A}_1$ is $\gamma_1(G)$ and the average sensitivity of $\mathcal{A}_2(\cdot,S_1)$ is $\gamma_2(G)$ for any $S_1 \in \mathcal{S}_1$, both with respect to the total variation distance.
	Let $\mathcal{A}: \mathcal{G} \to \mathcal{S}_2$ be a randomized algorithm obtained by composing $\mathcal{A}_1$ and $\mathcal{A}_2$, that is, $\mathcal{A}(G) = \mathcal{A}_2(G,\mathcal{A}_1(G))$.
	Then, the average sensitivity of $\mathcal{A}$ is $\gamma_1(G)+\gamma_2(G)$ with respect to the total variation distance.
\end{lemma}
\begin{proof}
	For a distribution $\mathcal{D}$, we use $f_\mathcal{D}$ to denote its probability mass function.
	Consider a graph $G = (V,E)$.
	\noindent Note that
	$$f_{\mathcal{A}(G)}(S_2) = \int_{\mathcal{S}_1} f_{\mathcal{A}_2(G,S_1)}(S_2) f_{\mathcal{A}_1(G)}(S_1) \; \mathrm{d}S_1.$$ Then we have that, for $e \in E$,
	\begin{align*}
	& d_{\mathrm{TV}}\bigl(\mathcal{A}(G),\mathcal{A}(G-e)\bigr)  \\
	& =
	\frac{1}{2}\int_{\mathcal{S}_2} \left|\int_{\mathcal{S}_1} f_{\mathcal{A}_2(G,S_1)}(S_2) f_{\mathcal{A}_1(G)}(S_1)\;\mathrm{d}S_1
	-
	\int_{\mathcal{S}_1} f_{\mathcal{A}_2(G-e,S_1)}(S_2) f_{\mathcal{A}_1(G-e)}(S_1)\; \mathrm{d}S_1 \right|\; \mathrm{d}S_2\\
	& =
	\frac{1}{2}\int_{\mathcal{S}_2} \biggl|\int_{\mathcal{S}_1} f_{\mathcal{A}_2(G,S_1)}(S_2) \Bigl(f_{\mathcal{A}_1(G)}(S_1)-f_{\mathcal{A}_1(G-e)}(S_1)\Bigr) \; \mathrm{d}S_1
	-  \\
	& \qquad \quad \int_{\mathcal{S}_1} \Bigl(f_{\mathcal{A}_2(G-e,S_1)}(S_2) - f_{\mathcal{A}_2(G,S_1)}(S_2)\Bigr) f_{\mathcal{A}_1(G-e)}(S_1) \; \mathrm{d}S_1\biggr| \; \mathrm{d}S_2\\
	& \leq
	\frac{1}{2}\int_{\mathcal{S}_1}\biggl|f_{\mathcal{A}_1(G)}(S_1)-f_{\mathcal{A}_1(G-e)}(S_1) \biggr|\;\mathrm{d}S_1 \cdot \int_{\mathcal{S}_2} f_{\mathcal{A}_2(G,S_1)}(S_2)\;\mathrm{d}S_2
	+  \\
	& \qquad \int_{\mathcal{S}_1} f_{\mathcal{A}_1(G-e)}(S_1)\;\mathrm{d}S_1 \cdot \frac{1}{2}\int_{\mathcal{S}_2}  \biggl| f_{\mathcal{A}_2(G-e,S_1)}(S_2) - f_{\mathcal{A}_2(G,S_1)}(S_2)\biggr|\;\mathrm{d}S_2   \\
	& = \frac{1}{2}\int_{\mathcal{S}_1}\biggl|f_{\mathcal{A}_1(G)}(S_1)-f_{\mathcal{A}_1(G-e)}(S_1) \biggr|\;\mathrm{d}S_1
	+ \\
	& \qquad \int_{\mathcal{S}_1} f_{\mathcal{A}_1(G-e)}(S_1) \;\mathrm{d}S_1 \cdot \frac{1}{2}\int_{\mathcal{S}_2}  \biggl| f_{\mathcal{A}_2(G-e,S_1)}(S_2) - f_{\mathcal{A}_2(G,S_1)}(S_2)\biggr| \;\mathrm{d}S_2  \\
	& = d_{\mathrm{TV}}\bigl(\cA_1(G), \cA_1(G-e)\bigr) + \int_{\mathcal{S}_1} f_{\mathcal{A}_1(G-e)}(S_1) \cdot d_{\mathrm{TV}}\bigl(\cA_2(G, S_1), \cA_2(G-e, S_1)\bigr) \;\mathrm{d}S_1.
	\end{align*}

	Hence, the average sensitivity of $\cA$ with respect to the total variation distance can be bounded as,
	\begin{align*}
	\E_{e \sim E} \left[d_{\mathrm{TV}}\bigl(\mathcal{A}(G),\mathcal{A}(G-e)\bigr)\right] &\le \E_{e \sim E} \left[d_{\mathrm{TV}}\bigl(\cA_1(G), \cA_1(G-e)\bigr)\right] +\\
	& \qquad \E_{e \sim E} \left[\int_{\mathcal{S}_1} f_{\mathcal{A}_1(G-e)}(S_1) \cdot d_{\mathrm{TV}}\bigl(\cA_2(G, S_1), \cA_2(G-e, S_1)\bigr)\;\mathrm{d}S_1\right]\\
	&\le \gamma_1(G) + \int_{\mathcal{S}_1} f_{\mathcal{A}_1(G-e)}(S_1)\;\mathrm{d}S_1 \cdot \gamma_2(G) = \gamma_1(G) + \gamma_2(G). \qedhere
	\end{align*}
\end{proof}

\subsection{Parallel composition}

%It is often the case that there are multiple algorithms that solve the same problem with some having better average sensitivities than others in some regimes. Such averagely stable algorithms can be composed by running them according to a distribution determined by the input graph. The advantage of such a composition is that the average sensitivity of the resulting algorithm might be better than the component algorithms in all regimes.
In this section, we prove Theorem~\ref{thm:avg-sensitivity-distribution-of-multiple-algorithms}, which bounds the average sensitivity of an algorithm obtained by running different algorithms according to a distribution in terms of the average sensitivities of the component algorithms.
We restate the theorem here.
\asdma*
\begin{proof}
	Consider a graph $G = (V,E)$.
	For a solution $S$, let $p^G(S)$ denote the probability that $S$ is output on input $G$ by $\cA$.
	Let $p_i^G(S)$ denote the probability that $S$ is output on input $G$ by $\cA_i$.
	For every solution $S$, we know that $p^G(S) = \sum_{i \in [\ell]} \rho_i(G) \cdot p_i^G(S)$.

	Let $\cA(G)$ denote the output distribution of $\cA$ on $G$.
	Fix $e \in E$.
	We first bound the earth mover's distance between $\cA(G)$ and $\cA(G-e)$.
	In order to transform $\cA(G)$ into $\cA(G-e)$, we first transform $p^G(S)$, for each solution $S$, into $\sum_{i \in [\ell]} \rho_i(G) \cdot p_i^{G-e}(S)$. This can be done at a cost of at most $\sum_{i \in [\ell]} \rho_i(G) \cdot d_{\text{EM}}(\cA_i(G), \cA_i(G-e))$.

	We now convert $\sum_{i \in [\ell]} \rho_i(G) \cdot p_i^{G-e}(S)$, for each solution $S$, into $\sum_{i \in [\ell]}\rho_i(G-e) \cdot p_i^{G-e}(S)$ at a cost of at most $2\mathsf{H}\cdot \frac{1}{2}\sum_{i \in [\ell]}|\rho_i(G) - \rho_i(G-e)|$, where $\frac{1}{2}\sum_{i \in [\ell]}|\rho_i(G) - \rho_i(G-e)|$ is the total variation distance between the probability distributions with which $\cA$ selects the algorithms on inputs $G$ and $G-e$.
	Hence, the average sensitivity of $\cA$ is at most
	\[
	\sum_{i \in [\ell]}\rho_i(G)\cdot \beta_i(G) + \mathsf{H}\cdot \E_{e \sim E}\left[\sum_{i \in [\ell]}|\rho_i(G) - \rho_i(G-e)|\right].\qedhere
	\]
\end{proof}

\noindent We separately state the special case of Theorem~\ref{thm:avg-sensitivity-distribution-of-multiple-algorithms} for $\ell = 2$.
\begin{theorem}\label{thm:avg-sensitivity-distribution-of-algorithms}
	Let $\cA_1$ and $\cA_2$ be two algorithms for a graph problem with average sensitivities $\beta_1(G)$ and $\beta_2(G)$, respectively.
	Let $\cA$ be an algorithm that, given a graph $G$, runs $\cA_1$ with probability $\rho(G)$ and runs $\cA_2$ with the remaining probability.
	Let $\mathsf{H}$ denote the maximum Hamming weight among those of solutions obtained by running $\cA$ on $G$ and $\{G-e\}_{e \in E}$.
	Then the average sensitivity of $\cA$ is at most $\rho(G)\cdot \beta_1(G) + (1- \rho(G))\cdot \beta_2(G) + 2\mathsf{H}\cdot \E_{e \sim E}\left[|\rho(G) - \rho(G-e)|\right]$.
\end{theorem}

\subsection{Sublinearity implies low average sensitivity}
In this section, we prove Theorem~\ref{thm:generic-transformation-from-local-oracle-to-stable-algorithm}, which show that the existence of a sublinear-time solution oracle (Definition~\ref{def:solution-oracle}) for an algorithm $\cA$ implies that the average sensitivity of $\cA$ is bounded by the query complexity of that oracle.

\gtflotsa*
\begin{proof}
	We prove the theorem for the case that solutions output by $\cA$ are subsets of edges of the input graph. It can be easily modified to work for the case that the solutions output by $\cA$ are subsets of vertices of the input graph in which case, we will use the technical condition that $n \le m$.

	Without loss of generality, assume that $\cA$ uses $r(n)$ random bits when run on graphs of $n$ vertices\footnote{If $r(G)$ is the length of the random string used for $G$, we can simply set $r(n) = \max\{r(G): G=(V,E), |V| = n\}$. If we do not need $r(n)$ bits for some particular graph $G$ on $n$ vertices, we can just throw away the unused bits.}.
	Consider a graph $G = (V,E)$ that ${\cal O}$ gets access to.
	For $e \in E$ and a string $\pi \in \{0,1\}^{r(n)}$, let $Q_{e, \pi}$ denote the set of edges in $E$ queried by ${\cal O}$ on input $e$, while simulating the run of $\cA$ with $\pi$ as the random string.
	The set $Q_{e, \pi}$ denotes the set of edges $e'$ such that the status of $e$ in the solutions output by $\cA$ with randomness $\pi$ on inputs $G$ and $G-e'$ could be different.
	For each edge $e' \in E$ and string $\pi \in \{0,1\}^{r(n)}$, define $R_{e', \pi}$ as the set of edges $e \in E$ such that $e' \in Q_{e, \pi}$.

	By definition, for each $\pi \in \{0,1\}^{r(n)}$, we have $\sum_{e \in E} |R_{e, \pi}| = \sum_{e \in E} |Q_{e, \pi}|$.
	Hence we have: $$\sum_{\pi \in \{0,1\}^{r(n)}}\sum_{e \in E} |R_{e, \pi}|= \sum_{\pi \in \{0,1\}^{r(n)}}\sum_{e \in E} |Q_{e, \pi}|,$$ and
	$$\E_{\pi \sim \{0,1\}^{r(n)}}\E_{e \sim E} |R_{e, \pi}| = \E_{\pi \sim \{0,1\}^{r(n)}}\E_{e \sim E} |Q_{e, \pi}|\le q(G),$$
	where the last inequality follows from our assumption on ${\cal O}$.

	For $\pi \in \{0,1\}^{r(n)}$ and $e \in E$, the set $R_{e, \pi}$ contains the set of edges whose presence in the solution could be affected by the removal of $e$ from $G$. Therefore, it is a superset of the set of edges contained in the symmetric difference between the outputs of $\cA$ on inputs $G$ and $G-e$ when run with $\pi$ as the random string.

	Let $\cH_{\cA, \pi}(G,G')$ denote the Hamming distance between the outputs of the algorithm $\cA$ on inputs $G$ and $G'$ when run with $\pi$ as the random string.
	As per this notation, for each $e\in E$, $$\E_{\pi \in \{0,1\}^{r(n)}} \cH_{\cA, \pi}(G, G-e) \le \E_{\pi \in \{0,1\}^{r(n)}} |R_{e, \pi}|.$$

\noindent The following claim relates the quantity on the left hand side of the above inequality with the average sensitivity of $\cA$.

	\begin{claim}\label{clm:sensitivity-atmost-expected-Hamming-distance}
		The average sensitivity of $\cA$ is bounded as $$\beta(G) \le \mathop{\E_{e\in E(G)}} \mathop{\E_{\pi \in \{0,1\}^{r(n)}}} \cH_{\cA, \pi}(G,G-e).$$
	\end{claim}
	\begin{proof}
		Fix $G \in \mathcal{G}$ and $e \in E(G)$.
		We first bound the earth mover's distance between $\cA(G)$ and $\cA(G-e)$, where $\cA(G)$ and $\cA(G-e)$ are the output distributions of $\cA$ on inputs $G$ and $G-e$, respectively.
		For $S \in \cS$, let $p_{G}(S)$ and $p_{G-e}(S)$ denote the probabilities that $\cA$ outputs $S$ on $G$ and $G-e$, respectively.
		We start with $\cA(G)$.
		Consider a string $\pi \in \{0,1\}^{r(n)}$.
		Let $S \in \cS$ denote the output of $\cA$ on input $G$ when using the string $\pi$ as its random string.
		Let $S'$ denote the output that is generated when running $\cA$ on input $G-e$ with $\pi$ as the random string.
		We move a mass of $\frac{1}{2^{r(n)}}$ (corresponding to the string $\pi$) from $p_{G}(S)$ to $p_G(S')$ at a cost of $\frac{\ham(S,S')}{2^{r(n)}}$.
		Moving masses corresponding to every string $\pi \in \{0,1\}^{r(n)}$ this way, we can transform $\cA(G)$ to $\cA(G-e)$.
		The total cost incurred during this transformation is $\mathop{\E_{\pi \in \{0,1\}^{r(n)}}} \cH_{\cA, \pi}(G,G-e)$.
		Therefore the earth mover's distance between $\cA(G)$ and $\cA(G-e)$ is at most $\mathop{\E_{\pi \in \{0,1\}^{r(n)}}} \cH_{\cA, \pi}(G,G-e)$.
		Therefore the average sensitivity of $\cA$ is $\beta(G) \le \mathop{\E_{e \in E(G)}} \mathop{\E_{\pi \in \{0,1\}^{r(n)}}} \cH_{\cA, \pi}(G,G-e)$.
	\end{proof}

	\noindent Therefore, the average sensitivity of $\cA$ is:
	\[
	\beta(G) \le \mathop{\E_{e \sim E}} \mathop{\E_{\pi \in \{0,1\}^{r(n)}}} \cH_{\cA, \pi}(G, G-e)
	\le \mathop{\E_{e \sim E}} \E_{\pi \in \{0,1\}^{r(n)}} |R_{e, \pi}| \le q(G).\qedhere
	\]
\end{proof}

We now prove Corollary~\ref{cor:LCA-to-stable-on-average-algorithm} which says that the existence of an LCA (Definition~\ref{def:lca}) for a graph problem implies the existence of a stable-on-average algorithm for the same problem.

\Ltsoaa*
\begin{proof}
	Assume without loss of generality that each solution in $\cS$ is a subset of edges of its preimage with respect to $\mathcal{P}$.
	Consider the algorithm $\cA$ that, on input $G = (V,E)$, constructs a solution to $\mathcal{P}$ by running $\mathcal{L}$ on each edge $e \in E$ and combining the outputs of $\mathcal{L}$.
	It is clear	that $\mathcal{L}$ is a solution oracle (Definition~\ref{def:solution-oracle}) for the algorithm $\mathcal{A}$.
	Hence, the average sensitivity of $\cA$ is upper bounded by the expected number of queries made by $\mathcal{L}$, which is at most $q(|V|) + |E|\cdot \delta(|V|)$.
\end{proof}

\paragraph{Acknowledgments.}
We are grateful to anonymous reviewers for suggesting a major improvement to the average sensitivity analysis of Algorithm~\ref{alg:matching-rand-greedy}.
We thank Tasuku Soma and Samson Zhou for several helpful discussions.
We extend our gratitude to Sofya Raskhodnikova for helpful comments that improved the presentation of this article.
% \nvnote{Stating this separately because the form is convenient for use and we are referring to this statement at multiple places in the manuscript.}

%\input{matching-improved}

\bibliographystyle{plain}
\bibliography{stabilityOfAlgorithms}

\begin{thebibliography}{10}

\bibitem{AlonRVX12}
N.~Alon, R.~Rubinfeld, S.~Vardi, and N.~Xie.
\newblock Space-efficient local computation algorithms.
\newblock In {\em Proceedings of the 23rd Annual {ACM-SIAM} Symposium on
  Discrete Algorithms (SODA)}, pages 1132--1139, 2012.

\bibitem{Bavelas:1950}
A.~Bavelas.
\newblock Communication patterns in task-oriented groups.
\newblock {\em The Journal of the Acoustical Society of America},
  22(6):725--730, 1950.

\bibitem{Beauchamp:1965}
M.~A. Beauchamp.
\newblock An improved index of centrality.
\newblock {\em Behavioral Science}, 10(2):161--163, 1965.

\bibitem{Bousquet:2002wn}
O.~Bousquet and A.~Elisseeff.
\newblock Stability and generalization.
\newblock {\em Journal of Machine Learning Research}, pages 499--526, 2002.

\bibitem{CHK16}
K.~Censor{-}Hillel, E.~Haramaty, and Z.~S. Karnin.
\newblock Optimal dynamic distributed {MIS}.
\newblock In {\em Proceedings of the 2016 {ACM} Symposium on Principles of
  Distributed Computing (PODC)}, pages 217--226, 2016.

\bibitem{CMV18}
A.~Czumaj, Y.~Mansour, and S.~Vardi.
\newblock Sublinear graph augmentation for fast query implementation.
\newblock In {\em Proceedings of the 16th International Workshop on
  Approximation and Online Algorithms (WAOA)}, pages 181--203, 2018.

\bibitem{Dwork:2006dw}
C.~Dwork, F.~McSherry, K.~Nissim, and A.~Smith.
\newblock Calibrating noise to sensitivity in private data analysis.
\newblock In {\em Proceedings of the 3rd Theory of Cryptography Conference
  (TCC)}, pages 265--284, 2006.

\bibitem{E65}
J.~Edmonds.
\newblock Paths, trees, and flowers.
\newblock {\em Canadian Journal of mathematics}, pages 449--467, 1965.

\bibitem{erdHos1959random}
P.~Erd{\H{o}}s and A.~R{\'e}nyi.
\newblock On random graphs.
\newblock {\em Publicationes Mathematicae}, 6:290--297, 1959.

\bibitem{EvenMR18}
G.~Even, M.~Medina, and D.~Ron.
\newblock Best of two local models: Centralized local and distributed local
  algorithms.
\newblock {\em Inf. Comput.}, 262(Part):69--89, 2018.

\bibitem{Freeman:1977ww}
L.~C. Freeman.
\newblock A set of measures of centrality based on betweenness.
\newblock {\em Sociometry}, 40(1):35--41, 1977.

\bibitem{F1907}
G.~{Fubini}.
\newblock {Sugli integrali multipli.}
\newblock {\em Accademia dei Lincei, Rendiconti}, 16(1):608--614, 1907.

\bibitem{GJ79}
M.~R. Garey and D.~S. Johnson.
\newblock {\em Computers and Intractability: A Guide to the Theory of
  NP-Completeness}.
\newblock W. H. Freeman \& Co., New York, NY, USA, 1979.

\bibitem{GLMRT10}
A.~Gupta, K.~Ligett, F.~McSherry, A.~Roth, and K.~Talwar.
\newblock Differentially private combinatorial optimization.
\newblock In {\em Proceedings of the 21st Annual {ACM-SIAM} Symposium on
  Discrete Algorithms (SODA)}, pages 1106--1125, 2010.

\bibitem{HassidimKNO09}
A.~Hassidim, J.~A. Kelner, H.~N. Nguyen, and K.~Onak.
\newblock Local graph partitions for approximation and testing.
\newblock In {\em Proceedings of the 50th Annual {IEEE} Symposium on
  Foundations of Computer Science (FOCS)}, pages 22--31, 2009.

\bibitem{HassidimMV16}
A.~Hassidim, Y.~Mansour, and S.~Vardi.
\newblock Local computation mechanism design.
\newblock {\em {ACM} Trans. Economics and Comput.}, 4(4):21:1--21:24, 2016.

\bibitem{Hay:2009a}
M.~Hay, C.~Li, G.~Miklau, and D.~D. Jensen.
\newblock Accurate estimation of the degree distribution of private networks.
\newblock In {\em Proceedings of the 9th {IEEE} International Conference on
  Data Mining (ICDM)}, pages 169--178, 2009.

\bibitem{K93}
D.~R. Karger.
\newblock Global min-cuts in rnc, and other ramifications of a simple min-cut
  algorithm.
\newblock In {\em Proceedings of the 4th Annual ACM-SIAM Symposium on Discrete
  Algorithms (SODA)}, pages 21--30, 1993.

\bibitem{Karwa:2014}
V.~Karwa, S.~Raskhodnikova, A.~D. Smith, and G.~Yaroslavtsev.
\newblock Private analysis of graph structure.
\newblock {\em ACM Transactions on Database Systems}, 39(3):22:1--22:33, 2014.

\bibitem{Karwa:2012}
V.~Karwa and A.~B. Slavkovic.
\newblock Differentially private graphical degree sequences and synthetic
  graphs.
\newblock In {\em Proceedings of the International Conference on Privacy in
  Statistical Databases (PSD)}, pages 273--285, 2012.

\bibitem{Kasiviswanathan:2013}
S.~P. Kasiviswanathan, K.~Nissim, S.~Raskhodnikova, and A.~D. Smith.
\newblock Analyzing graphs with node differential privacy.
\newblock In {\em Proceedings of the 10th Theory of Cryptography (TCC)}, pages
  457--476, 2013.

\bibitem{Kempe:2003iu}
D.~Kempe, J.~Kleinberg, and {\'E}.~Tardos.
\newblock Maximizing the spread of influence through a social network.
\newblock In {\em Proceedings of the 9th ACM SIGKDD International Conference on
  Knowledge Discovery and Data Mining (KDD)}, pages 137--146, 2003.

\bibitem{Kruskal:1956}
J.~B. Kruskal.
\newblock On the shortest spanning subtree of a graph and the traveling
  salesman problem.
\newblock {\em Proceedings of the American Mathematical Society}, 7(1):48--50,
  1956.

\bibitem{LenzenL18}
C.~Lenzen and R.~Levi.
\newblock A centralized local algorithm for the sparse spanning graph problem.
\newblock In {\em Proceedings of the 45th International Colloquium on Automata,
  Languages, and Programming, {(ICALP)}}, pages 87:1--87:14, 2018.

\bibitem{LeviM17}
R.~Levi and M.~Medina.
\newblock A (centralized) local guide.
\newblock {\em Bulletin of the {EATCS}}, 122, 2017.

\bibitem{LeviRR20}
R.~Levi, D.~Ron, and R.~Rubinfeld.
\newblock Local algorithms for sparse spanning graphs.
\newblock {\em Algorithmica}, 82(4):747--786, 2020.

\bibitem{LeviRY17}
R.~Levi, R.~Rubinfeld, and A.~Yodpinyanee.
\newblock Local computation algorithms for graphs of non-constant degrees.
\newblock {\em Algorithmica}, 77(4):971--994, 2017.

\bibitem{MansourPV18}
Y.~Mansour, B.~Patt{-}Shamir, and S.~Vardi.
\newblock Constant-time local computation algorithms.
\newblock {\em Theory Comput. Syst.}, 62(2):249--267, 2018.

\bibitem{MRVX12}
Y.~Mansour, A.~Rubinstein, S.~Vardi, and N.~Xie.
\newblock Converting online algorithms to local computation algorithms.
\newblock In {\em Proceedings of the 39th International Colloquium on Automata,
  Languages, and Programming (ICALP)}, pages 653--664, 2012.

\bibitem{MV13}
Y.~Mansour and S.~Vardi.
\newblock A local computation approximation scheme to maximum matching.
\newblock In {\em Proceedings of 16th International Workshop on Approximation
  Algorithms for Combinatorial Optimization (APPROX)}, pages 260--273, 2013.

\bibitem{Marchiori:2000dx}
M.~Marchiori and V.~Latora.
\newblock Harmony in the small-world.
\newblock {\em Physica A: Statistical Mechanics and its Applications},
  285(3-4):539--546, 2000.

\bibitem{McSherry:2007dh}
F.~McSherry and K.~Talwar.
\newblock Mechanism design via differential privacy.
\newblock In {\em Proceedings of the 48th Annual IEEE Symposium on Foundations
  of Computer Science (FOCS)}, pages 94--103, 2007.

\bibitem{Meulemans:2018jr}
W.~Meulemans, B.~Speckmann, K.~Verbeek, and J.~Wulms.
\newblock A framework for algorithm stability and its application to kinetic
  euclidean {MSTs}.
\newblock In {\em Proceedings of the 13th Latin American Symposium on
  Theoretical Informatics (LATIN)}, pages 805--819, 2018.

\bibitem{Murai:2019hG}
S.~Murai and Y.~Yoshida.
\newblock Sensitivity analysis of centralities on unweighted networks.
\newblock In {\em Proceedings of the 2019 World Wide Web Conference (WWW)},
  pages 1332--1342, 2019.

\bibitem{Newman:2004jh}
M.~E.~J. Newman.
\newblock Fast algorithm for detecting community structure in networks.
\newblock {\em Physical Review E}, 69(6):066133, 2004.

\bibitem{Newman:2006iq}
M.~E.~J. Newman.
\newblock Modularity and community structure in networks.
\newblock {\em Proceedings of the National Academy of Sciences},
  103(23):8577--8582, 2006.

\bibitem{Nguyen:2008fr}
H.~N. Nguyen and K.~Onak.
\newblock Constant-time approximation algorithms via local improvements.
\newblock In {\em Proceedings of the 49th Annual IEEE Symposium on Foundations
  of Computer Science (FOCS)}, pages 327--336, 2008.

\bibitem{Nissim:2007}
K.~Nissim, S.~Raskhodnikova, and A.~D. Smith.
\newblock Smooth sensitivity and sampling in private data analysis.
\newblock In {\em Proceedings of the 39th Annual {ACM} Symposium on Theory of
  Computing (STOC)}, pages 75--84, 2007.

\bibitem{Page:1999wg}
L.~Page, S.~Brin, R.~Motwani, and T.~Winograd.
\newblock The pagerank citation ranking: Bringing order to the web.
\newblock Technical report, Stanford InfoLab, 1999.

\bibitem{ParterRVY19}
M.~Parter, R.~Rubinfeld, A.~Vakilian, and A.~Yodpinyanee.
\newblock Local computation algorithms for spanners.
\newblock In {\em Proceedings of 10th Innovations in Theoretical Computer
  Science Conference {(ITCS)}}, pages 58:1--58:21, 2019.

\bibitem{Raskhodnikova:2016}
S.~Raskhodnikova and A.~D. Smith.
\newblock Lipschitz extensions for node-private graph statistics and the
  generalized exponential mechanism.
\newblock In {\em Proceedings of the {IEEE} 57th Annual Symposium on
  Foundations of Computer Science (FOCS)}, pages 495--504, 2016.

\bibitem{ReingoldV16}
O.~Reingold and S.~Vardi.
\newblock New techniques and tighter bounds for local computation algorithms.
\newblock {\em J. Comput. Syst. Sci.}, 82(7):1180--1200, 2016.

\bibitem{Rubinfeld:2011}
R.~Rubinfeld, G.~Tamir, S.~Vardi, and N.~Xie.
\newblock Fast local computation algorithms.
\newblock In {\em Proceedings of the 1st Symposium on Innovations in Computer
  Science (ICS)}, pages 223--238, 2011.

\bibitem{Sabidussi:1966}
G.~Sabidussi.
\newblock The centrality index of a graph.
\newblock {\em Psychometrika}, 31(4):581--603, 1966.

\bibitem{ShalevShwartz:2009ij}
S.~Shalev-Shwartz and S.~Ben-David.
\newblock {\em Understanding Machine Learning}.
\newblock From Theory to Algorithms. Cambridge University Press, Cambridge,
  2009.

\bibitem{YYI12}
Y.~Yoshida, M.~Yamamoto, and H.~Ito.
\newblock Improved constant-time approximation algorithms for maximum matchings
  and other optimization problems.
\newblock {\em SIAM Journal on Computing}, 41(4):1074--1093, 2012.

\end{thebibliography}

%!TEX root=./stabilityOfAlgorithms.tex

\appendix
%%%%%%%%%%%%%%%%%%%%%%%%%%%%%%%%%%%%%%%%%%%%%%%%%%
\section{Average Sensitivity of Exponential Mechanism}\label{sec:exponential-appendix}
%%%%%%%%%%%%%%%%%%%%%%%%%%%%%%%%%%%%%%%%%%%%%%%%%%
In this section, we prove Lemma~\ref{lem:exponential-mechanism}.
\begin{proof}[Proof of Lemma~\ref{lem:exponential-mechanism}]
Let $t > 0$ be a parameter.
	Note that any index $i \in [n]$ with $\bm{x}(i) > \mathsf{OPT} + \log n/\eta + t/\eta$ has probability at most $e^{-t}/n$ of being sampled by $\mathcal{A}$.
	Hence, by a union bound, for every $t > 0$
	\[
	\Pr_{i \sim \mathcal{A}(\bm{x})}\left[\bm{x}(i) \geq \mathsf{OPT} + \frac{\log n}{\eta} + \frac{t}{\eta}\right] \leq e^{-t}.
	\]

	Next, we analyze the distance between the output distributions.
	Let $\bm{x},\bm{x}' \in \mathbb{R}^n$ be vectors, and let $Z = \sum_{i \in [n]}e^{-\eta \bm{x}(i)}$ and $Z' = \sum_{i \in [n]}e^{-\eta \bm{x}'(i)}$.
	Without loss of generality, we assume that $Z \geq Z'$.
  First, note that for all $i \in [n]$ such that $\bm{x}(i) \geq \bm{x}'(i)$, we have
  \[
	  0 \leq e^{-\eta \bm{x}'(i)} - e^{-\eta \bm{x}(i)}
	  = e^{-\eta \bm{x}'(i)} \left(1 - e^{-\eta (\bm{x}(i) -\bm{x}'(i))}\right)
		\leq \eta e^{-\eta \bm{x}'(i)} (\bm{x}(i) -\bm{x}'(i)).
  \]
  Hence for any $i \in [n]$, we have
  \begin{align*}
	  & |e^{-\eta \bm{x}(i)} - e^{-\eta \bm{x}'(i)}|
	  \leq \max\left\{\eta e^{-\eta \bm{x}(i)} (\bm{x}'(i) -\bm{x}(i)), \eta e^{-\eta \bm{x}'(i)} (\bm{x}(i) -\bm{x}'(i))\right\} \\
	  & \leq \eta |\bm{x}(i)-\bm{x}'(i)| \max\{e^{-\eta \bm{x}(i)}, e^{-\eta \bm{x}'(i)}\}
	  \leq \eta |\bm{x}(i)-\bm{x}'(i)| \left(e^{-\eta \bm{x}(i)} +  e^{-\eta \bm{x}'(i)}\right).
  \end{align*}
  Then, we have
  \begin{align}
    & \frac{1}{Z}\sum_{i \in [n]}|e^{-\eta \bm{x}(i)} - e^{-\eta \bm{x}'(i)}|
    \leq \frac{\eta}{Z}\sum_{i \in [n]} |\bm{x}(i)-\bm{x}'(i)| \left( e^{-\eta \bm{x}(i)} + e^{-\eta \bm{x}'(i)}\right) \nonumber \\
    & \leq \frac{\eta}{Z} \max_{i \in [n]}|\bm{x}(i)-\bm{x}'(i)| \sum_{i \in [n]}  \left( e^{-\eta \bm{x}(i)} + e^{-\eta \bm{x}'(i)}\right)
    = \frac{\eta(Z+Z')}{Z} \max_{i \in [n]}|\bm{x}(i)-\bm{x}'(i)| \nonumber \\
    &\leq 2\eta\|\bm{x}-\bm{x}'\|_1. \label{eq:exponential-mechanism-2}
  \end{align}

	% First, note that
	% \begin{align}
	% & \frac{Z - Z'}{Z}
	% %    = \frac{1}{Z}\sum_{i \in [n]}\left(e^{-\eta \bm{x}(i)} - e^{-\eta \bm{x}'(i)}\right)
	% = \frac{\sum_{i \in [n]}e^{-\eta \bm{x}(i)}\left(1 - e^{\eta (\bm{x}(i) -\bm{x}'(i))}\right)}{\sum_{i \in [n]}e^{-\eta \bm{x}(i)}}
	% \leq \max_{i \in [n]}\frac{e^{-\eta \bm{x}(i)}\left(1 - e^{\eta (\bm{x}(i) -\bm{x}'(i))}\right)}{e^{-\eta \bm{x}(i)}} \nonumber \\
	% & \leq \eta \max_{i \in [n]}\bigl( \bm{x}(i) - \bm{x}'(i)\bigr)
	% \leq \eta \|\bm{x} - \bm{x}'\|_1. \label{eq:exponential-mechanism-2}
	% \end{align}

	% \nvnote{The claimed upper bound on the total variation distance is actually $0$ below. The problem is a missing absolute value sign. I am thinking about how to bound the correct expression.}

	Then, the total variation distance between $\mathcal{A}(\bm{x})$ and $\mathcal{A}(\bm{x}')$ is at most
	\begin{align*}
	& \sum_{i \in [n]}\left|\frac{\exp(-\eta \bm{x}(i))}{Z} - \frac{\exp(-\eta \bm{x}'(i))}{Z'}\right|
	= \sum_{i \in [n]}\left|\frac{\exp(-\eta \bm{x}(i))}{Z} - \frac{\exp(-\eta \bm{x}'(i))}{Z}\left(\frac{Z-Z'+Z'}{Z'}\right)\right| \\
	&= \sum_{i \in [n]}\left|\frac{\exp(-\eta \bm{x}(i))}{Z} - \frac{\exp(-\eta \bm{x}'(i))}{Z} - \frac{\exp(-\eta \bm{x}'(i))}{Z} \left(\frac{Z-Z'}{Z'}\right)\right| \\
	&\leq \frac{1}{Z}\sum_{i \in [n]}\left|e^{-\eta \bm{x}(i)} - e^{-\eta \bm{x}'(i)} \right| + \frac{Z-Z'}{Z} \frac{1}{Z'}\sum_{i \in [n]}\exp(-\eta \bm{x}'(i))  \\
	& \leq \frac{2}{Z}\sum_{i \in [n]}\left|e^{-\eta \bm{x}(i)} - e^{-\eta \bm{x}'(i)} \right|
	 \leq 4\eta\|\bm{x}-\bm{x}'\|_1.\qedhere
%	& \leq \frac{\eta }{Z}\sum_{i \in [n]}|\bm{x}(i)-\bm{x}'(i)| + \frac{Z-Z'}{Z}\left(\frac{1}{Z'}\sum_{i \in [n]}\exp(-\eta \bm{x}'(i))\right) \\
%	&= O\left( \eta \|\bm{x}-\bm{x}'\|_1 \cdot \left(1 + \left(\frac{\log n}{\eta} + \frac{1}{\eta}\right)\left(\mathsf{OPT}' + \frac{\log n}{\eta}\right) + \frac{\Gamma(2)}{\eta^2} \right)\right) \tag{by~\eqref{eq:exponential-mechanism-1} and~\eqref{eq:exponential-mechanism-2}} \\
%	& \le O\left(\|\bm{x}-\bm{x}'\|_1 \cdot \left(\eta + \frac{\Gamma(2)}{\eta} + \left(\log n + 1\right)\left(1 + \frac{\log n}{\eta}\right)  \right)\right) \\
%	& = O\left( \|\bm{x}-\bm{x}'\|_1 \cdot \left(\eta + \left(1 + \frac{1}{\eta}\right) \log^2 n \right)\right).
%	\qedhere
	\end{align*}
	%  The last inequality in the statement can be obtained by setting $\eta = \sqrt{(\log n + 1)/2}$
\end{proof}
%!TEX root=./stabilityOfAlgorithms.tex

\section{Average Sensitivity of Prim's algorithm}\label{sec:prim}

In this section, we show that Prim's algorithm (with a simple tie-breaking rule, as described in Algorithm~\ref{alg:prim}) has high average sensitivity even on unweighted graphs. This is in contrast to the low average sensitivity of Kruskal's algorithm that we discussed in Section~\ref{sec:spanning-tree}.

\begin{algorithm}
	\caption{\textsc{Prim's Algorithm}}\label{alg:prim}
	\Input{undirected graph $G = ([n],E)$}
	Let $T \gets \{1\}$\;
	\While{\emph{there exists a vertex not spanned by }$T$}{
	Let $E'$ be the set of edges with the smallest weight among all the edges in $E$ that have exactly one endpoint in $T$\;
	Add to $T$, an edge from $E'$ that has lexicographically smallest $T$-endpoint among all edges in $E'$, breaking further ties arbitrarily.
	}
	\Return Output $T$.
\end{algorithm}
\begin{figure}[!htb]
	\center{\includegraphics[scale=0.5]{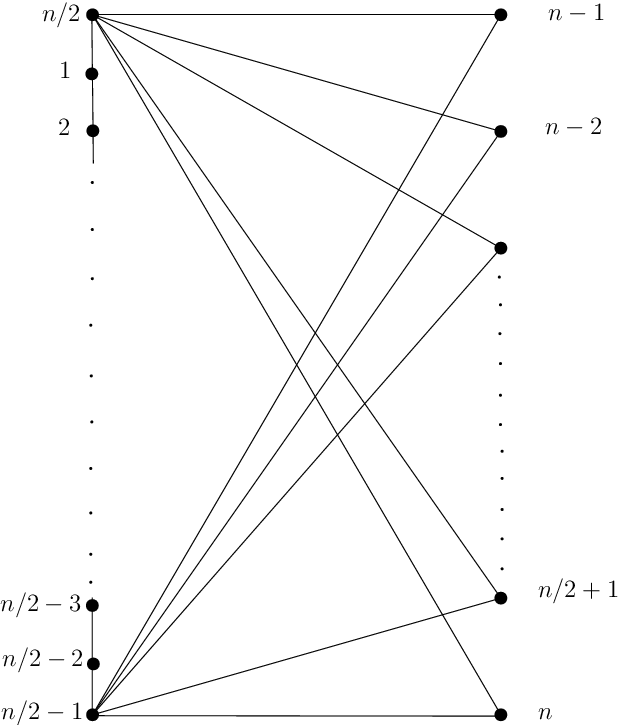}}
	\caption{The graph family ${\{G_{n}\}}_{n \in 2\mathbb{N}}$.}\label{fig:prim-sensitivity-lower-bound}
\end{figure}

\begin{lemma}
The average sensitivity of Prim's algorithm is $\Omega(m)$.
\end{lemma}
\begin{proof}
	Consider the graph family ${\{G_{n}\}}_{n \in 2\mathbb{N}}$ in Figure~\ref{fig:prim-sensitivity-lower-bound}. For a large enough $n \in 2\mathbb{N}$, consider running Algorithm~\ref{alg:prim} on $G_{n}$. The tree $T$ output will consist of the edges $(i, i+1)$ for all $i \in [n/2 -2]$, the edges $(n/2 -1, j)$ for all $j \in \{n/2+1, \dots n\}$, and the edge $(n/2, 1)$.

	If we remove an edge $(i', i'+1)$ for $i' \in [n/2 -2]$ from $G_n$ and run Algorithm~\ref{alg:prim} on the resulting graph, the tree, say $T_{i'}$, output will consist of all edges of the form $(i, i + 1)$ for $i \in [n/2 -1] \setminus \{i'\}$, all edges of the form $(n/2, j)$ for all $j \in \{n/2+1, \dots n\}$, and the edges $(n/2+1, n/2-1)$ and $(n/2,1)$. The Hamming distance of $T_{i'}$ from $T$ is equal to $n/2$.

	Since a uniformly random edge removed from $G_n$ is of the form $(i, i+1)$ for $i \in [n/2 -2]$ with probability $\frac{n/2 - 2}{3n/2 - 1}$, the average sensitivity of Algorithm~\ref{alg:prim} is at least $\frac{n}{2} \cdot \frac{n/2 - 2}{3n/2 - 1}$, which is at least $\frac{n}{6} - 1 = \Omega(m)$ for the family ${\{G_n\}}_{n \in 2\mathbb{N}}$.
\end{proof}

\end{document}